\documentclass{article}

\usepackage[utf8]{inputenc}

\def\showauthornotes{1}

\def\showkeys{0}
\def\showdraftbox{0}
\def\showcolorlinks{1}
\def\usemicrotype{1}
\def\showfixme{0}

\providecommand{\RR}{\mathbb{R}}

\usepackage{etex}

\usepackage[l2tabu, orthodox]{nag}

\usepackage{xspace,enumerate}

\usepackage[dvipsnames]{xcolor}

\usepackage[T1]{fontenc}
\usepackage[full]{textcomp}

\usepackage[american]{babel}

\usepackage{mathtools}

\usepackage{amsthm}
\usepackage{amsmath}

\newtheorem{theorem}{Theorem}[section]
\newtheorem*{theorem*}{Theorem}

\newtheorem{proposition}[theorem]{Proposition}
\newtheorem*{proposition*}{Proposition}
\newtheorem{lemma}[theorem]{Lemma}
\newtheorem*{lemma*}{Lemma}
\newtheorem{corollary}[theorem]{Corollary}
\newtheorem*{corollary*}{Corollary}
\newtheorem*{conjecture*}{Conjecture}

\newtheorem*{fact*}{Fact}

\newtheorem*{hypothesis*}{Hypothesis}

\theoremstyle{definition}
\newtheorem{definition}[theorem]{Definition}
\newtheorem*{definition*}{Definition}

\newtheorem{example}[theorem]{Example}

\theoremstyle{remark}
\newtheorem{claim}[theorem]{Claim}
\newtheorem*{claim*}{Claim}
\newtheorem{remark}[theorem]{Remark}
\newtheorem*{remark*}{Remark}

\newtheorem*{observation*}{Observation}

\usepackage[a4paper,
top=1.3in,
bottom=1.3in,
left=1.6in,
right=1.6in]{geometry}

\usepackage[varg]{pxfonts} % varg - uses nicer g,v,y,w

\usepackage[sc,osf]{mathpazo}
\usepackage{lmodern}

\ifnum\showkeys=1
\usepackage[color]{showkeys}
\fi

\ifnum\showcolorlinks=1
\usepackage[
colorlinks=true,
urlcolor=blue,
linkcolor=blue,
citecolor=OliveGreen,
]{hyperref}
\fi

\ifnum\showcolorlinks=0
\usepackage[
colorlinks=false,
pdfborder={0 0 0}
]{hyperref}
\fi

\usepackage{prettyref}

\newcommand{\savehyperref}[2]{\texorpdfstring{\hyperref[#1]{#2}}{#2}}

\newrefformat{eq}{\savehyperref{#1}{\textup{(\ref*{#1})}}}
\newrefformat{lem}{\savehyperref{#1}{Lemma~\ref*{#1}}}
\newrefformat{def}{\savehyperref{#1}{Definition~\ref*{#1}}}
\newrefformat{thm}{\savehyperref{#1}{Theorem~\ref*{#1}}}
\newrefformat{cor}{\savehyperref{#1}{Corollary~\ref*{#1}}}
\newrefformat{cha}{\savehyperref{#1}{Chapter~\ref*{#1}}}
\newrefformat{sec}{\savehyperref{#1}{Section~\ref*{#1}}}
\newrefformat{app}{\savehyperref{#1}{Appendix~\ref*{#1}}}
\newrefformat{tab}{\savehyperref{#1}{Table~\ref*{#1}}}
\newrefformat{fig}{\savehyperref{#1}{Figure~\ref*{#1}}}
\newrefformat{hyp}{\savehyperref{#1}{Hypothesis~\ref*{#1}}}
\newrefformat{alg}{\savehyperref{#1}{Algorithm~\ref*{#1}}}
\newrefformat{rem}{\savehyperref{#1}{Remark~\ref*{#1}}}
\newrefformat{item}{\savehyperref{#1}{Item~\ref*{#1}}}
\newrefformat{step}{\savehyperref{#1}{step~\ref*{#1}}}
\newrefformat{conj}{\savehyperref{#1}{Conjecture~\ref*{#1}}}
\newrefformat{fact}{\savehyperref{#1}{Fact~\ref*{#1}}}
\newrefformat{prop}{\savehyperref{#1}{Proposition~\ref*{#1}}}
\newrefformat{prob}{\savehyperref{#1}{Problem~\ref*{#1}}}
\newrefformat{claim}{\savehyperref{#1}{Claim~\ref*{#1}}}
\newrefformat{relax}{\savehyperref{#1}{Relaxation~\ref*{#1}}}
\newrefformat{red}{\savehyperref{#1}{Reduction~\ref*{#1}}}
\newrefformat{part}{\savehyperref{#1}{Part~\ref*{#1}}}
\newrefformat{ex}{\savehyperref{#1}{Example~\ref*{#1}}}

\newcommand{\Sref}[1]{\hyperref[#1]{\S\ref*{#1}}}

\usepackage{nicefrac}
\ifnum\usemicrotype=1
\usepackage{microtype}
\fi

\ifnum\showauthornotes=1
\newcommand{\Authornote}[2]{{\sffamily\small\color{red}{[#1: #2]}}}
\newcommand{\Authornotecolored}[3]{{\sffamily\small\color{#1}{[#2: #3]}}}
\newcommand{\Authorcomment}[2]{{\sffamily\small\color{gray}{[#1: #2]}}}
\newcommand{\Authorstartcomment}[1]{\sffamily\small\color{gray}[#1: }

\newcommand{\Authorfnote}[2]{\footnote{\color{red}{#1: #2}}}
\newcommand{\Authorfixme}[1]{\Authornote{#1}{\textbf{??}}}
\newcommand{\Authormarginmark}[1]{\marginpar{\textcolor{red}{\fbox{\Large #1:!}}}}
\else
\newcommand{\Authornote}[2]{}
\newcommand{\Authornotecolored}[3]{}
\newcommand{\Authorcomment}[2]{}
\newcommand{\Authorstartcomment}[1]{}

\newcommand{\Authorfnote}[2]{}
\newcommand{\Authorfixme}[1]{}
\newcommand{\Authormarginmark}[1]{}
\fi

\newcommand{\Ynote}{\Authornotecolored{Purple}{Y}}
\newcommand{\Inote}{\Authornotecolored{ForestGreen}{I}}
\newcommand{\inote}{\Inote}

\ifnum\showfixme=0

\fi

\usepackage{boxedminipage}

\newcommand{\brac}[1]{[#1]}
\newcommand{\Brac}[1]{\left[#1\right]}

\newcommand{\abs}[1]{\lvert#1\rvert}
\newcommand{\Abs}[1]{\left\lvert#1\right\rvert}

\newcommand{\card}[1]{\lvert#1\rvert}

\newcommand\sett[2]{\left\{ #1 \left| \; \vphantom{#1 #2} \right. #2  \right\}}
\newcommand{\set}[1]{\{#1\}}

\newcommand{\norm}[1]{\lVert#1\rVert}

\newcommand{\iprod}[1]{\langle#1\rangle}

\newcommand{\Esymb}{\mathbb{E}}
\newcommand{\Psymb}{\mathbb{P}}

\DeclareMathOperator*{\E}{\Esymb}

\DeclareMathOperator*{\ProbOp}{\Psymb}

\renewcommand{\Pr}{\ProbOp}
\def\one{{\mathbf{1}}}

\newcommand{\prob}[1]{\Pr \left[ {#1} \right] }
\newcommand{\Prob}[2][]{\Pr_{{#1}}\left[#2\right]} % use by \Prob[x]{event}
\newcommand{\cProb}[3]{\Pr_{{#1}}\left[ #2 \left| \; \vphantom{#2 #3} \right. #3  \right]} % use by \Prob[x]{event}

\newcommand{\ex}[1]{\E\brac{#1}}
\newcommand{\Ex}[2][]{\E_{{#1}}\Brac{#2}}

\newcommand{\plu}[2]{\PluOp_{{#1}}{[{#2}}]}

\newcommand{\ve}{\;\hbox{and}\;}

 \usepackage{dsfont}
\usepackage{mathrsfs}

\newcommand{\textparen}[1]{\text{(#1)}}

\ifx\because\undefined
\newcommand{\because}[1]{\textparen{because #1}}
\else
\renewcommand{\because}[1]{\textparen{because #1}}
\fi

\newcommand{\defeq}{\stackrel{\mathrm{def}}=}

\newcommand\bdot\bullet

\DeclareMathOperator{\dist}{dist}

\renewcommand{\leq}{\leqslant}
\renewcommand{\le}{\leqslant}
\renewcommand{\geq}{\geqslant}
\renewcommand{\ge}{\geqslant}

\ifnum\showdraftbox=1

\else

\fi

\let\epsilon=\varepsilon

\numberwithin{equation}{section}

\newcommand{\MYstore}[2]{%
  \global\expandafter \def \csname MYMEMORY #1 \endcsname{#2}%
}

\newcommand{\MYload}[1]{%
  \csname MYMEMORY #1 \endcsname%
}

\newcommand{\MYnewlabel}[1]{%
  \newcommand\MYcurrentlabel{#1}%
  \MYoldlabel{#1}%
}

\newcommand{\MYdummylabel}[1]{}

\newcommand{\torestate}[1]{%
  \let\MYoldlabel\label%
  \let\label\MYnewlabel%
  #1%
  \MYstore{\MYcurrentlabel}{#1}%
  \let\label\MYoldlabel%
}

\newcommand{\restatetheorem}[1]{%
  \let\MYoldlabel\label
  \let\label\MYdummylabel
  \begin{theorem*}[Restatement of \prettyref{#1}]
    \MYload{#1}
  \end{theorem*}
  \let\label\MYoldlabel
}

\newcommand{\restatelemma}[1]{%
  \let\MYoldlabel\label
  \let\label\MYdummylabel
  \begin{lemma*}[Restatement of \prettyref{#1}]
    \MYload{#1}
  \end{lemma*}
  \let\label\MYoldlabel
}

\newcommand{\restateprop}[1]{%
  \let\MYoldlabel\label
  \let\label\MYdummylabel
  \begin{proposition*}[Restatement of \prettyref{#1}]
    \MYload{#1}
  \end{proposition*}
  \let\label\MYoldlabel
}

\newcommand{\restateclaim}[1]{%
  \let\MYoldlabel\label
  \let\label\MYdummylabel
  \begin{claim*}[Restatement of \prettyref{#1}]
    \MYload{#1}
  \end{claim*}
  \let\label\MYoldlabel
}

\newcommand{\restatecorollary}[1]{%
  \let\MYoldlabel\label
  \let\label\MYdummylabel
  \begin{corollary*}[Restatement of \prettyref{#1}]
    \MYload{#1}
  \end{corollary*}
  \let\label\MYoldlabel
}

\newcommand{\restatefact}[1]{%
  \let\MYoldlabel\label
  \let\label\MYdummylabel
  \begin{fact*}[Restatement of \prettyref{#1}]
    \MYload{#1}
  \end{fact*}
  \let\label\MYoldlabel
}

\newcommand{\restatedefinition}[1]{% % overwrite label command with dummy
\let\MYoldlabel\label 
\let\label\MYdummylabel 
\begin{definition*}[Restatement of \prettyref{#1}] 
    \MYload{#1} 
\end{definition*} 
\let\label\MYoldlabel 
} 

\newcommand{\restate}[1]{%
  \let\MYoldlabel\label
  \let\label\MYdummylabel
  \MYload{#1}
  \let\label\MYoldlabel
}

\newcommand{\eps}{\epsilon}

\let\origparagraph\paragraph
\renewcommand{\paragraph}[1]{\origparagraph{#1.}}

\allowdisplaybreaks

\newcommand{\dunion}{\mathbin{\mathaccent\cdot\cup}}

\let\pref=\prettyref

\renewcommand{\restriction}{\mathord{\upharpoonright}}
\newcommand{\rest}[2]{{#1}\restriction_{{#2}}}
\newcommand{\sqbrack}[1]{\lbrack {#1} \rbrack}
\newcommand{\bigO}[1]{O \left( #1 \right)}

\def\withComments{1}
\def\withDarkBackground{0}

\newcommand{\Up}{U}
\newcommand{\Down}{D}
\newcommand{\cx}[2]{ \ell_2 ( {#1} ({#2})) }
\newcommand{\CX}[1]{\cx{X}{{#1}} }

\if\withDarkBackground 1
\pagecolor{gray!70}

\fi

\if\withComments 1
\newcommand{\todo}[1]{ {\marginpar{\color{red} \par TODO: {#1}}} }

\else
\newcommand{\todo}{}
\fi

\newcommand{\scloc}[2]{ {#1}^{#2} }
\newcommand{\scres}[2]{ {#1}_{#2} }

\newcommand{\col}{col}

\renewcommand{\epsilon}{\varepsilon}
\renewcommand{\phi}{\varphi}
\providecommand{\eps}{\epsilon}

\providecommand{\RR}{\mathbb{R}}

\newcommand{\F}{\mathbb{F}}

\newif\ifWithFKN

\WithFKNfalse

\newcommand{\FF}{f}

\newcommand{\dl}[2]{#1}
\newcommand{\ul}[2]{\set{{#2} \supset {#1}}}
\newcommand{\adj}[1]{{\textrm{ reach}}_{{#1}}} % adjacent vertices
\newcommand{\ane}[1]{\overset{{#1}}{\ne}}

\newcommand{\Grassmann}[3]{Gr_{lin}({#2},{#3})}
\newcommand{\Graff}[3]{Gr_{aff}({#2},{#3})}

\newrefformat{ass}{\savehyperref{#1}{Assumption~\ref*{#1}}}

\usepackage[normalem]{ulem}
\usepackage[shortlabels]{enumitem}

\renewcommand{\plu}{ {\textrm {pop}} }

\newcommand{\VASA}{{ VASA}}
\newcommand{\vASA}[1]{_{#1}ASA}
\newcommand{\VAS}[1]{VAS_{#1}}
\newcommand{\AV}[1]{AV_{#1}}

\newcommand\rem[1]{{}}

\newcommand{\myconst}{\frac{1}{40}}
\newcommand{\twomyconst}{\frac{1}{20}}
\newcommand{\fourmyconst}{\frac{1}{20}}

\newcommand{\comp}[1]{M^{#1}}

\usepackage{appendix}

\makeatletter
\newcommand{\labeleditem}[1]{%
\item[#1]\protected@edef\@currentlabel{#1}%
}
\makeatother

\usepackage{cleveref}
\crefname{lemma}{lemma}{lemmata}
\crefname{claim}{claim}{claims}
\crefname{proposition}{proposition}{propositions}
\crefname{section}{section}{sections}
\crefname{subsection}{subsection}{subsections}
\crefname{definition}{definition}{definitions}
\usepackage{newunicodechar}

\usepackage{csquotes}
\title{Agreement testing theorems on layered set systems}
\author{Yotam Dikstein\thanks{Weizmann Institute of Science, ISRAEL. email: {\texttt yotam.dikstein@weizmann.ac.il}.}
\and Irit Dinur\thanks{Weizmann Institute of Science, ISRAEL. email: {\texttt irit.dinur@weizmann.ac.il}.}}

\date{\today}

\begin{document}
\maketitle
\begin{abstract}
We introduce a framework of layered subsets, and give a sufficient condition for when a set system supports an agreement test. Agreement testing is a certain type of property testing that generalizes PCP tests such as the plane vs. plane test. 
Previous work has shown that high dimensional expansion is useful for agreement tests. We extend these results to more general families of subsets, beyond simplicial complexes.
These include
%We prove several new agreement testing results:
\begin{itemize}
  \item Agreement tests for set systems whose sets are faces of high dimensional expanders. Our new tests apply to all dimensions of complexes both in case of two-sided expansion and in the case of one-sided partite expansion.  This improves and extends an earlier work of Dinur and Kaufman (FOCS 2017) and applies to matroids, and potentially many additional complexes.
  \item Agreement tests for set systems whose sets are neighborhoods of vertices in a high dimensional expander. This family resembles the expander neighborhood family used in the gap-amplification proof of the PCP theorem. This set system is quite natural yet does not sit in a simplicial complex, and demonstrates some versatility in our proof technique. %It could potentially be useful towards new and more efficient constructions of locally testable codes and PCPs.
  \item Agreement tests on families of subspaces (also known as the Grassmann poset). This extends the classical low degree agreement tests beyond the setting of low degree polynomials.
\end{itemize}
Our analysis relies on a new random walk on simplicial complexes which we call the ``complement random walk'' and which may be of independent interest. This random walk generalizes the non-lazy random walk on a graph to higher dimensions, and has significantly better expansion than previously-studied random walks on simplicial complexes.

\end{abstract}

\newpage
\tableofcontents{}
\newpage
\def\glob{\mathcal{G}}
\section{Introduction}
%
%\paragraph{What is an agreement test?}

Agreement testing is a certain type of property testing. The first agreement testing theorems are the line versus line or plane versus plane low degree agreement tests \cite{RuSu96,ArSu,RazS1997} that play an important part in various PCP constructions. We discuss the history and evolution of these tests further below.

Abstractly, an agreement test is the following. Let $V$ be a ground set and let $S$ be a family of subsets of $V$. The object being tested is an ensemble of local functions $\sett{f_s\in \Sigma^s}{s\in S}$ with one function per set $s\in S$. The domain of $f_s$ is $s$ itself. A {\em perfect ensemble} is an ensemble that comes from a global function $g:V\to\Sigma$ whose domain is the entire vertex set. In a perfect ensemble the local function at $s$ is the restriction of $g$ to the set $s$, that is, $f_s= \rest  g s$ for all $s\in S$.

We let $\glob$ be the set of all perfect ensembles. An agreement test is a property tester for $\glob$. It is specified by a distribution over pairs\footnote{In some cases the test can query more than two subsets, as in the so-called Z-test of \cite{ImpagliazzoKW2012}, but in this paper we restrict attention only to two query tests.} of intersecting subsets, $s_1,s_2\in S$, and the test accepts if the respective local functions agree on the intersection: $\rest {f_{s_1}}t = \rest{f_{s_2}}t$ where  $t=s_1\cap s_2$.
A perfect ensemble is clearly accepted with probability $1$. The test is {\em $c$-sound} if
\begin{equation}\label{eq:c-soundness}
  \dist(f,\glob) \le c\cdot \Pr_{s_1,s_2} [ \rest {f_{s_1}}t = \rest{f_{s_2}}t ]\,.
\end{equation}
Here the distance $\dist(f,\glob)$ is the minimal fraction of sets $s\in S$ that we need to change in $f$ in order to get a function in $\glob$.

It is well known (see \pref{ex:noisyensemble}) that in some cases exact soundness is impossible and we must allow a slightly weaker notion, called $\gamma$-approximate soundness. The $\gamma$-approximate distance between two ensembles $f$ and $g$, denoted $\dist_\gamma(f,g)$, is the fraction of sets $s$ in which $\dist(f_s,g_s)>\gamma$. An agreement test is {\em $\gamma$-approximately $c$-sound} if
\begin{equation}\label{eq:approx-soundness}
  \dist_\gamma(f,\glob) \le c\cdot \Pr_{s_1,s_2} [ \rest {f_{s_1}}t = \rest{f_{s_2}}t ]\,.
\end{equation}
This means that if the test succeeds with probability $1-\eps$ there must be a global function $g:V\to\Sigma$ such that for all but $c\cdot\eps$ of the sets $s$, $\dist(f_s,\rest g s)\le \gamma$.%
\paragraph{Why study agreement tests}
The original motivation for agreement tests comes from PCP proof composition: a key step in this construction is to combine many small proofs into one global proof, but without knowing whether the small proofs are consistent with each other. The agreement test ensures that they can be combined together coherently. Indeed, agreement tests are the basis of the ``inner verifier'' constructed in recent works on $2:2$ games \cite{KhotMS2017,DinurKKMS2018-2to1,BarakKS19,KhotMS18}.

Recent work \cite{DinurFH19} used agreement tests in a different context, for proving structure theorems for Boolean functions. The idea is to prove structure for small restrictions of the function, often an easier task, and then apply an agreement testing theorem to combine these structures together.

Agreement tests are a natural family of tests that seems interesting in its own right. This work makes a step towards developing a theory that explains which set systems have agreement tests.

\subsection*{The STAV layered set system}
We describe a three layered set system which we call a STAV. %Our main technical theorem says that every set system that supports a STAV must support an agreement test. This reduces the task of proving an agreement test to the much simpler task of uncovering a STAV underneath the set system.

Looking closely at agreement tests, we can always model them with three layers: the vertices ($V$), the sets ($S$) and the possible intersections between sets ($T$). The STAV has an additional so-called ``Amplification'' layer ($A$) that captures an amplification  property that occurs in many interesting settings: given that we know that two local functions agree on part of the intersection, the probability that they will agree on the whole intersection rises significantly.

We give an informal description of STAV, for the detailed formal definition please see \pref{sec:stav}. A STAV is a tuple $(S,T,A,V)$ together with the following three distributions
\begin{itemize}
  \item The STAV distribution - a distribution over $(s,t,a,v)$, $s \supset t \supset a\dunion v$.
  \item The STS distribution - a distribution over $s_1,t,s_2$ that gives the agreement testing distribution and in addition a subset $t\subseteq s_1\cap s_2$.
  \item The VASA distribution - a distribution over $v,a,s,a'$ whose role will be made clear in the analysis.
\end{itemize}
A STAV is called $\gamma$-good if these distributions (and some local views of them) satisfy certain spectral conditions.

\paragraph{The {\em surprise} parameter} Based on the STAV structure, it is natural to define a parameter which we call the {surprise}. This parameter depends both on the ensemble $f=\set{f_s}$ and on the STAV, and in some cases, it can be bounded independently of $f$  (this is the case for simplicial complexes). The surprise parameter is a measure of how much amplification the $A$ layer gives us. It is the probability that two intersecting sets agree on $a$ given that they disagree on $t$ (See \pref{def:surprise}). This parameter gives a unified way to address different agreement scenarios.
\subsection*{Main Results}
Our main technical theorem (\pref{thm:main-STAV-agreement-theorem}) says that every set system that supports a $\gamma$-good STAV must support a sound agreement test. This reduces the task of proving an agreement test to the much simpler task of uncovering a STAV underneath the set system.\\
%\begin{theorem*}[Informal, see formal \pref{thm:main-STAV-agreement-theorem}]
%Let $S$ be a set system over a ground set $V$. If there are set systems $A,T \subset 2^V$ and distributions such that $S,T,A,V$ is a $\gamma$-good STAV then $S$ supports a $\gamma$-approximate sound agreement test for ensembles with surprise at most $\gamma$.
%\end{theorem*}

We list here a few applications of this theorem, starting with agreement tests for high dimensional expanders.
Introducing high dimensional expanders is beyond the current scope and we refer the reader to \pref{sec:HDX} for more introductory definitions.%
\begin{theorem}[Agreement for two-sided HDX - short version of \pref{thm:main-agreement-theorem-two-sided}]
There exists a constant $c>0$ such that for every $d$-dimensional simplicial complex $X$ the following holds. If $X$ is a $\frac{1}{d^3}$-two-sided $d$-dimensional HDX, then $X(d)$ supports a $c$-sound agreement test.
\end{theorem}
In \pref{sec:agreement-on-hdx} we describe some corollaries of this theorem for matroids. 

The only known constructions of {\em sparse} two-sided HDXs are by truncating one-sided HDXs, see the Ramanujan complexes of \cite{LSV} as well as the construction of HDXs due to \cite{KaufmanO18}. It is natural to study agreement tests for the (non-truncated) one-sided HDX itself. The following theorem gives such a result in the special case that the complex is also $d+1$-partite. Many Ramanujan complexes are naturally $d+1$-partite, as are the complexes constructed in \cite{KaufmanO18}.
\begin{theorem}[Agreement for partite one-sided HDX - short version of \pref{thm:main-agreement-theorem-one-sided}]
There exists a constant $c > 0$ such that the following holds. Suppose $X$ is a $(d+1)$-Partite complex that is a $\frac{1}{d^3}$-one sided HDX. Then $X(d)$ supports a $c$-sound agreement test.
\end{theorem}
Our next agreement theorem is for a family of subsets that is derived from a high dimensional expander, although itself it does not sit inside a simplicial complex. The subsets in this family are balls, or neighborhoods, of a vertex or a higher dimensional face in a simplicial complex that is a HDX.
This construction resembles the set system underlying the gap-amplification based proof of the PCP theorem \cite{Din07}, in which an agreement theorem underlies the argument somewhat implicitly.
\begin{theorem}[Agreement on neighborhoods - short version of \pref{thm:agreement-on-links}]
There exists a constant $c > 0$ such that the following holds. Let $X$ be a $\frac{1}{d^3}$-two-sided high dimensional expander. For each vertex $z\in X(0)$ let $B_z$ be the set of neighbors of $z$, and let $S = \sett{B_z}{z\in X(0)}$. Then $S$ supports a $\frac{1}{d}$-approximately $c$-sound agreement test.
\end{theorem}
Finally, our last agreement theorem is for a family of subspaces of a vector space, also called the Grassmann. Such families were studied in PCP constructions for special ensembles whose local functions belong to some code. Such ensembles are guaranteed to have the following property. For all $s_1,t,s_2$, if $ \rest {f_{s_1}}t \neq \rest{f_{s_2}}t$ then $\dist(\rest {f_{s_1}}t , \rest{f_{s_2}}t)\ge \delta$. We call such ensembles $\delta$-ensembles and prove,

\begin{theorem}[Agreement on subspaces - informal, see \pref{thm:agreement-on-Grassmann-affine}]
There exists a constant $c > 0$ such that the following holds. Let $\F^n$ be a vector space and let $S$ have a set for every affine subspace of dimension $d$. Then $S$ supports a $1/q^{\Omega(d)}$-approximately $c$-sound agreement test for $\delta$-ensembles.
\end{theorem}%
For the benefit of the reader we added in \pref{sec:list-of-results} a list of theorems proven in this work.

\subsection*{Overview of the proof of our main theorem (\pref{thm:main-STAV-agreement-theorem})}
Our main agreement theorem on STAV structures has two parts, as in many previous works.
The first part of the proof uses the amplification given by the surprise parameter to construct a family of functions for each $a \in A$, that is $\mathbf{g} = \sett{g_a: \adj{a} \to \Sigma}{a \in A}$. The reach of $a$ is the set of all vertices $v$, so that $\set{v} \dunion a \subset s$ for some $s \in S$. The value $g_a(v)$ is defined by popularity of $f_s(v)$ for all $s\supset a$. This part is standard and occurs in many agreement test analyses.

The second part of the proof is our main new technical contribution.
In this step one constructs a global $G:V\to\Sigma$ from the pieces $g_a$. This is done by showing sufficient agreement between the different $g_a$'s. We consider a graph connecting a pair $a,a'$ when they sit together inside some $s$. In earlier works this graph is dense and has very low diameter ($2$ typically). This can only happen when the functions $g_a$ are defined on a pretty large part of the vertex set (as in \cite{DinurSteurer2014, BDL17, DinurFH19, RazS1997}) unlike our context where each $\adj a$ is quite tiny (its size can be a constant, far smaller than $\card V$). When the diameter is small and $\adj a$ is huge it is easy to stitch the different $g_a$'s together, even when the agreement between the $g_a$'s is rather crude, by taking a very short random walk from $a$ to $a'$ to $a''$.

In contrast, in our case the diameter is logarithmic and we cannot afford a random walk because the error would build up badly.
Instead, we construct the global function $G:V \to \Sigma$ by
\[ G(v) = \plu \sett{g_a(v)}{a \in \adj{v}}, \]
i.e. the most popular opinion of the $g_a$'s on $v$. We show that it has the desired properties. This argument relies on the fact that the VASA random walk (in particular, moving from $a$ to $s$ to $a'$) is a very strong expander. That such VASA distributions are available is proven through a new type of random walk which we call the complement random walk, and is discussed separately below.

The only previous work that analyzed an agreement test on a sparse set system (where this ``large diameter'' problem appears) was in \cite{DinurK2017}. Their solution circumvented this problem by reducing to the dense case in a certain way. That reduction is ad-hoc and required an additional external layer of sets above \(S\), which limited the generality of the theorem. Whereas the current proof is more direct and works without this technical caveat.
\subsubsection*{The complement random walk in high dimensional expanders}
Several previous works \cite{KaufmanM17,DinurK2017,KaufmanO2017} analyzed random walks on high dimensional expanders\footnote{In this section we assume familiarity with high dimensional link expansion, see \pref{sec:HDX}  for formal definitions.}. In this work we study a new type of random walk which we call the complement random walk. 

Interestingly, independent recent work of Alev, Jeronimo, and Tulsiani  \cite{AlevJT2019}, studies the same walk, where it is called ``swap walk''. The authors use this walk for analyzing an algorithm that solves constraint satisfaction problems (CSPs) on high dimensional expanders.

The complement walk goes from $i$-face to $i$-face via a shared $j$-face, just like the upper and lower random walks previously studied. However it has significantly better expansion, and is hence much more useful for us. We construct with it $\gamma$-good STAVs in many of our applications. The problem with many of the previously studied random walks is that they have an inherent ``laziness'' built in: starting from an $i$ face and walking down to a $j$ face, and then back up to another $i$ face, the $j+1$ common vertices are limiting the expansion of this walk (the family of all sets containing a fixed vertex will have not-so-good expansion). In contrast, the complement walk starts with an $i$-face $a$ moves up to a $j$-face $b\supset a$ and then moves down to another $i$-face $a'\subset b$ conditioned on $a,a'$ being \emph{disjoint} (of course we need $j\geq 2i+1$, note that any choice of such $j$ would give the exact same random walk). It turns out (see \pref{thm:complement-walk-is-a-good-expander}) that this walk has great expansion. This can be seen by examining for example the case of $i=0$ \rem{and $j=1$} and noting that this is just the non-lazy random walk on a graph.

We prove the properties of this (and other) walks in \pref{sec:complement-walk}. The proof goes through \emph{Garland's method}. This method, proves global properties of the simplicial complexes by properties on the links. This method, originally developed by Garland in \cite{Garland1973}, is used in many works such as ~\cite{EvraK2016,DinurK2017,Oppenheim2018}.

We believe these random walks are interesting on their own account. These walks generalize the non-lazy adjacency operator in a graph, and the bipartite adjacency operator in a bipartite graph to high dimensions. As a bonus we show an immediate application for these walks: a new high dimensional expander mixing lemma for sets in all dimensions (see \pref{lem:two-sided-HDEML} and \pref{lem:one-sided-HDEML}), extending the work of \cite{coloring,oppenheim3}.
\subsection*{More background and context}
As mentioned earlier the first agreement testing theorems are the line versus line or plane versus plane low degree agreement tests \cite{RuSu96,ArSu,RazS1997} that play an important part in various PCP constructions. Combinatorial analogs of these theorems were subsequently dubbed ``direct product tests'' and studied in a sequence of works \cite{GolSaf97,DR06,DG08,ImpagliazzoKW2012,DinurSteurer2014,DL17}.  For a long while there were only two prototypical set systems for which agreement tests were known:
\begin{itemize}
  \item All $k$-dimensional subspaces of some vector space
  \item All $k$-element subsets of an underlying ground set
\end{itemize}
Each of these has several variants (varying the field size and ambient dimension, deciding whether the sets are ordered or not, etc.). 

The study of agreement tests initially came as a part of a PCP construction, as in the case of the low degree agreement tests and later in works leading towards combinatorial proofs for the PCP theorem, as started in \cite{GolSaf97} and continued in \cite{DR06,Din07}.

Further works relied on agreement tests for hardness amplification: \cite{ImpagliazzoKW2012} showed hardness for label cover (called a two-query PCP) based on their direct product agreement test. A recent line of work \cite{KhotMS2017,DinurKKMS2018-2to1, BarakKS19,KhotMS18} concerning unique and $2:2$ games used agreement tests on the Grassmann as an inner verifier (see in particular \cite{DinurKKMS2018-2to1}).

In hope of getting more efficient PCPs and LTCs it seemed that understanding the power of agreement tests in a more general setting would give us a better handle on domains in which locally testable codes and PCP constructions can reside. However, despite some attempts, no derandomization techniques managed to find further (and hopefully sparser) constructions.

A couple of years ago \cite{DinurK2017} discovered a new and very sparse set system that supports an agreement test. This new system is based on group theoretic (and number theoretic) constructions of so-called high dimensional expanders. The number of sets in this set system is linear in the size of the ground set, a feature that seems key towards new and more efficient locally testable codes and PCPs.

This suggested that there is possibly a much richer collection of set systems that support agreement tests, and brought to the fore once more the question of understanding which set systems support agreement tests.

\section{Agreement Tests for STAV Structures}
\def\sts{\textsc{sts}}
\def\stsa{{G_{STSa}}}
\def\agr{{\textrm agree}}
\newcommand\disagr[1]{\textrm {rej}_{{#1}}}
\def\glob{{\mathcal G}}
\def\testd{\mathcal{D}}
\def\stavd{D_\textsc{stav}}
\def\surp{\xi}
\def\vasad{D_\textsc{vasa}}

%} [Commented this line since it made errrors - Yotam]
\subsection{Agreement tests and agreement expansion}
We begin with the definition of an agreement expander, similar to that of \cite{DinurK2017}.
Let $S$ be a family of subsets of a ground set $V$. An ensemble of local functions is a collection $\sett{f_s:s\to\Sigma}{s\in S}$ consisting, for each subset $s\in S$, of a function whose domain is $s$. A perfect ensemble is one that comes from a global function $g:V\to\Sigma$, namely $f_s= \rest g s$ for all $s\in S$. We denote the set of all perfect ensembles by \[ \glob(V;\Sigma) = \sett{ \set{\rest g s}_{s\in S}}{g:V\to\Sigma}.\]

An agreement test is given by a distribution $\testd$ over pairs of intersecting subsets,
\begin{itemize}
\item Input: An ensemble of local functions $\sett{f_s:s\to\Sigma}{s\in S}$
\item Test: Choose a random edge $\set{s_1,s_2}$ according to the distribution $\testd$, let $t=s_1\cap s_2$ and accept iff $\rest{f_{s_1}}t = \rest{f_{s_2}}t$.
\end{itemize}
We denote by $\disagr{\testd}(f)$ the probability that the agreement test rejects a given ensemble $f=\set{f_s}$. A perfect ensemble is clearly accepted with probability $1$. We say that the test is sound if it is a sound test for the property $\glob(V;\Sigma)$ in the standard property testing sense, namely,
\begin{definition}[Sound agreement test]
  An agreement test is $c$-sound if every ensemble $f=\set{f_s}$ satisfies
  \[ \dist(f,\glob) \le c \cdot \disagr \testd (f) \,.\]
\end{definition}
Finally we can define an agreement expander,
\begin{definition}[$c$-agreement expander]\label{def:aexpander}
  An agreement expander is a family $S$ of subsets of a ground set $V$ that supports a $c$-sound agreement test.
\end{definition}
The reason for the term ``agreement expander'' is the similarity to a Rayleigh quotient given by
\[ \frac 1 c  = \inf_{f\not\in\glob}\frac{\disagr \testd (f)}{\dist(f,\glob)}\,,
\]
where the numerator counts the number of rejecting edges and the denominator measures the distance from the property. See \cite{KaufmanL14} for a more detailed analogy between expansion and property testing.

\subsection*{Approximate versus exact agreement}
For some agreement tests one cannot expect a conclusion as strong as in \pref{def:aexpander}. For example, suppose that the testing distribution $\testd$ selects pairs $s_1,s_2$ that typically intersect on an $\eta\ll 1$ fraction of $s_1$ (and of $s_2$). In such a case consider the following ensemble,
\begin{example}\label{ex:noisyensemble} Construct an ensemble $f=\set{f_s}$ at random as follows. For all $s$ set $f_s = \rest 0 s$ and then for each $s$ with probability $\alpha$ do: change one bit of $f_s$ at random.
\end{example}
This ensemble passes the test with probability at least $1-2\alpha\eta$ while being roughly $\alpha$-far from $\glob$. Setting $\alpha=1$ rules out any kind of conclusion as in  \pref{def:aexpander}.
However, not all is lost, and a meaningful theorem can still be proven if we move to a softer notion of {\em approximate} agreement. Let us denote by $\dist_\gamma(f,f')$ the fraction of sets $s$ on which $f_s,f'_s$ differ on more than $\gamma$ fraction of $s$. Namely,
\[ \dist_\gamma(f,f') = \Pr_s [ \dist(f_s,f'_s)>\gamma ].
\]
\begin{definition}[$\gamma$-approximate soundness]
  An agreement test is $\gamma$-approximately $c$-sound if every ensemble $f=\set{f_s}$ satisfies
  \[ \dist_\gamma(f,\glob) \le c \cdot \disagr \testd (f) \,.\]
\end{definition}
When $\gamma < 1/\card s$ we recover the previous notion of soundness which we now call {\em exact} soundness. So a test is $c$-sound or exactly $c$-sound if it is $\gamma$-approximately $c$-sound for some $\gamma<1/\card s$.
\subsection{STAV structures}\label{sec:stav}
A STAV structure introduces two additional layers of subsets of $V$: layer $T$ and layer $A$. These come in addition to the top layer $S$ that we already have in the definition of an agreement expander. The layer $T$ represents the intersections of pairs of subsets $s_1,s_2\in S$, and is implicit in the definition of the agreement test distribution. The layer $A$ is new and sits below $T$. It provides a certain amplification needed for the analysis.
\rem{\todo{\begin{enumerate}
    \item move STAV- structure to poset notation.
    \item draw v-a-s etc distributions.
    \item Besides the formal definition I need to discuss this and give examples: Grassmann, HDX, etc.
\end{enumerate}
}}

\begin{figure}[h!]
    \centering
    \includegraphics[scale=0.5]{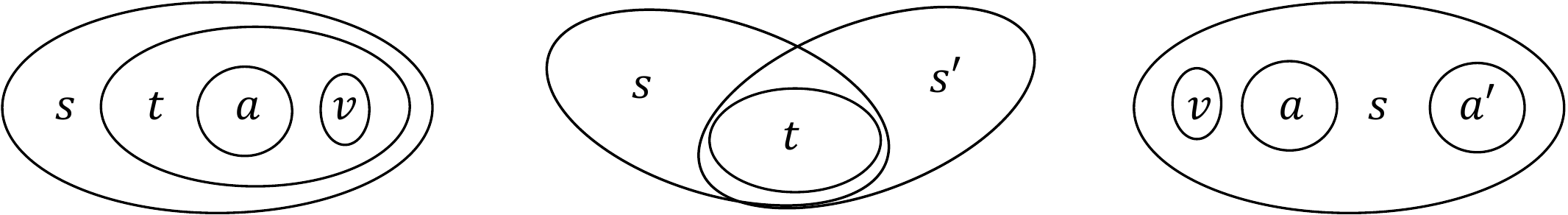}
    \caption{The STAV, STS, and VASA distributions}
    \label{fig:all3}
\end{figure}
\begin{definition}[STAV structure]\torestate{\label{def:STAV}
A STAV structure is a tuple $X=(S,T,A,V;\,\stavd)$ consisting of a ground set $V$ and three layers of subsets \(A,T,S \subset \mathbb{P}(V)\), together with a stochastic process $\stavd$ that samples $(s,t,a,v)$ as follows. %The process can be described by a $4$-layer graph with vertex sets $S,T,A,V$ and edges from $S$ to $T$ as well as edges between all three layers $T,A,V$.
\begin{itemize}
    \item Choose $s$
    \item Choose $t$ conditioned on $s$
    \item Choose $a,v$ conditioned on $t$ (but not dependent on $s$)
%The probability of choosing $(s,t)$ is the probability of choosing $(s,t)$ in the agreement test.
 \end{itemize}
The distributions in which the above are chosen are not restricted except for assuming that the marginal of this process is uniform over $v$ and that the probability to choose a vertex or a set is never zero.
The STAV comes with two distributions,\begin{itemize}
  \item {\em STS distribution:} A distribution over triples $(s_1,t,s_2)$ that is symmetric with respect to $s_1,s_2$ and satisfies that the marginal of $(s_1,t)$ (and therefore $(s_2,t)$) is identical to the marginal of $\stavd$.
  \item {\em VASA distribution:} A distribution $\vasad$ over tuples $(v,a_1,s,a_2)$ that is symmetric with respect to $a_1,a_2$ and satisfies that the marginal of $(v,a_1,s)$ (and therefore $(v,a_2,s)$) is identical to the marginal of $\stavd$.
  \end{itemize}}\end{definition}

Notation: Throughout this paper we use the letters $s,t,a,v$ to denote elements in $S,T,A$ and $V$ respectively without specifically mentioning this. So for example fixing $a_0$, $\ul{a_0}{s}$ stands for all elements of $S$ that contain $a_0\in A$. Unless specified otherwise, all random choices are with respect to the distributions $\stavd$ or the STS or VASA distributions.\\

Before we continue to define what a ``good'' STAV is, let us mention a couple of examples that might be useful to keep in mind.
\begin{example}[The direct product test STAV]\label{ex:DP}
Fix $k$ and let $\ell = k/3$. We construct the following family of STAVs for all $n \gg k$, $n\to\infty$. Let $V=[n]$, let $S =\binom{[n]}{k}, T=\binom{[n]}{\ell}$ and $A=\binom{[n]}{\ell-1}$. The STAV distribution is choosing a $k$-element set uniformly, then an $\ell$-element subset of it, and then splitting $t$ randomly into $a$ and $v$. A possible STS distribution is to choose a random $t$ and then two independent $s_1,s_2\supset t$. Another possibility is to choose $s_1,s_2\supset t$ so that their intersection is exactly \(t\). The VASA distribution is to choose $s$ uniformly and in it $a,a',v$ uniformly so that they are all disjoint.

An agreement test for this example appears in \cite{DinurSteurer2014} under the name direct product test.
\end{example}

\begin{example}[HDX simplicial complexes, generalizing \pref{ex:DP}]
\label{ex:hdx-sc}
Fix $k$ and let $\ell = k/3$. We construct the following family of STAVs for infinitely many $n \gg k$. Suppose $X$ is a high dimensional expander on $n$ vertices. Let $V=X(0)$, let $S =X(k), T=X(\ell)$ and $A=X(\ell-1)$. The STAV distribution is choosing a random $s$ from the distribution of $X$, then a uniform $t\subset s$, and then splitting $t$ randomly into $a$ and $v$. A possible STS distribution is to choose a random $t$ and then two independent $s_1,s_2\supset t$. Another possibility is to choose $s_1,s_2\supset t$ so that they must be disjoint. The VASA distribution is to choose $s$ according to the $X$ distribution and in it $a,a',v$ uniformly so that they are all disjoint.
\end{example}
Agreement tests for this example were analyzed in \cite{DinurK2017} for certain complexes $X$ and certain bounds on the dimension $k$.

\begin{example}[Subspaces STAV]
Fix $m>d>\ell$. We construct the following family of STAVs for all finite fields $\F=\F_q$, $q\to\infty$. Let $V=\F^m$, let $S$ be all $d$-dimensional spaces of $V$, let $T$ be all $\ell$-dimensional spaces of $V$ and let $A$ be all $(\ell-1)$-dimensional spaces of $V$. The STAV distribution is choosing $s$ uniformly, $t\subset s$ uniformly, then $a\subset t$ uniformly, then $v$ uniformly from $t\setminus a$. A possible STS distribution is to choose a random $t$ and then two uniform  $s_1,s_2\supset t$. The VASA distribution is to choose $s$ uniformly and in it $a,a',v$ uniformly so that they are all disjoint.

This example generalizes the plane vs. plane low degree agreement test. An agreement test for it is proved in \cite{RazS1997} for ensembles whose local functions are low degree functions, and in \cite{ImpagliazzoKW2012} for general ensembles (in both cases the focus was on a different parameter regime).
\end{example}
We now define several graphs that arise as local views of the STS and VASA distributions. The first of these is the bipartite graph obtained by the marginal of $\stavd$ on $A$ and $V$,
\begin{definition}[The \(\AV{}\)-Graph (reach graph)]\torestate{\label{def:reach-graph}
   The \(\AV{}\)-graph, or reach graph, is a bipartite graph $(V,A,E)$ where the probability of choosing an edge $(v,a)$ is given by the marginal of $\stavd$ on $V\times A$, namely,
   $Pr[(v,a)] = \sum_{s,t}\Prob[\stavd]{(s,t,a,v)}$.}
\end{definition}
We denote $\adj a\subset V$ the set of neighbors of $a$ in this graph, and by $\adj v\subset A$ the set of neighbors of $v$ in this graph.

    \begin{definition}[The local reach graphs]\torestate{\label{def:local-reach-graph}
    Let $X$ be a STAV-structue, and fix $s \in S$. The \emph{$s$-local} reach graph, or $\AV{s}$-graph, is a bipartite graph where:
    \[ L = \sett a{a \subset s}.\]
    \[ R = \sett v{v \in s}.\]
    \[ E = \sett{(a,v)}{v \in \adj{a}}. \]
    The probability of choosing an edge $(a,v)$ is the probability of choosing $(a,v)$ in the STAV-distribution given that we chose $s$.}
    \end{definition}

%\begin{definition}[link of a face]
% Let $X$ be a STAV structure, and let $x \in V$ or $x \in A$ or $x \in T$. The upper-link of $x$, denoted $\ul{x}{s}$ (respectively $\ul{x}{a},\ul{x}{t}$) are all the sets in $S$ ($A,T$ respectively) that contain $x$.
%
%\medskip
%
%Similarly, the lower link of $x$, denoted $\dl{x}{}$ are all the sets in $X$ that are contained inside $x$ (or vertices that belong to $x$). We denote by $\dl{\tau}{Y} = \dl{\tau}{} \cap X$.
%\end{definition}
%\inote{I think this def is not needed and should be removed}

\subsubsection*{The STS graph and its local views}
The STS distribution is conveniently viewed as a graph whose vertex set is $S$ and whose edges are labeled by elements of $T$, with the weight of the edge from $s_1$ to $s_2$ labeled by $t$ given by the probability of $(s_1,t,s_2)$. The graph is undirected since the STS distribution is symmetric wrt $s_1,s_2$.

We consider ``local views'' of the $\sts$ graph - obtained by inducing it on a smaller set of vertices.
\begin{definition}[$\sts_{a}$-Graph]\torestate{\label{def:sts-a-graph}
For a fixed $a$, an $\sts_{a}$-Graph is has vertex set $\sett s{s\supset a}$ and the probability of choosing an edge $\set{s_1,s_2}_t$ is given by $2\cProb{\sts}{(s_1,t,s_2)}{t\supset a}$.}
\end{definition}

\begin{definition}[$\sts_{a,v}$-Graph]\torestate{\label{def:sts-a-v-graph}
For a fixed $a,v$, an $\sts_{a,v}$-Graph is has vertex set $\sett s{s\supset a\cup\set v}$ and the probability of choosing an edge $\set{s_1,s_2}_t$ is given by $2\cProb{\sts}{(s_1,t,s_2)}{t\supset a\cup v}$.}
\end{definition}

\subsubsection*{Local views of the VASA distribution}
When fixing one of the four terms in $(v,a,s,a')$, we can define the following two graphs by the marginal:
\begin{definition}[$\vASA{v}$-Graph]\torestate{\label{def:v-asa-graph}
For a fixed $v$, an $\vASA{v}$-Graph is the graph whose vertex set is $\adj{v}$, and labeled edges are \[E = \sett{\set{a_1,a_2}_{s}}{a_1,a_2 \in A, s\in S, v,a_1,a_2 \subset s}.\]
The probability to choose an edge $\set{a_1,a_2}_s$ is given by \[\cProb{\vasad}{(v',a_1,s,a_2)}{v' = v}.\]}
\end{definition}

\begin{definition}[Bipartite $\VAS{a}$-Graph]\torestate{\label{def:bip-vas-a-graph}
For a fixed $a$, an $\VAS{a}$-Graph is the bipartite graph $(L,R,E)$ where
\begin{align*}
  L &= \adj{a}, \\
  R &= \sett{(a',s)}{\exists v \in L \; (v,a,s,a') \in Supp(\vasad) },\\
 E &= \sett{(v,(a',s)) }{ (v,a,s,a') \in Supp(\vasad) }.
\end{align*}
The probability of choosing an edge $(v,(a',s))$ is given by $\cProb{\vasad}{(v,a_0,s,a')}{a_0 = a}$.}
\end{definition}

%Admittedly this is the least ``natural'' graph we have defined, but it will be useful nevertheless.
%

\subsubsection*{Good STAV-Structures}
Having defined all the relevant graphs, we come to the requirements for a good STAV:
\inote{there exist vasa distr and sts distr such that}
\begin{definition}[A good STAV-Structure]
\label{def:good-stav-structure}
Let $X$ be STAV structure and $\gamma < 1$ be some constant. We say $X$ is a $\gamma$-good if assumptions (A1)-(A3) and one of (A4($r$)) or (A4) below hold for $X$:
\rem{\Ynote{I added the first assumption of what I wanted to require.
Our original proof asked $\adj{a} \cap s$ to be a constant size of $s$.

Another (I think weaker) requirement is the $\delta$-samplability - which essentially say that if $B \subset \set{v < s}$ is larger than $\delta$, then a constant fraction of the $\set{a < s}$ think $\cProb{}{v \in B}{v \sim a} \frac{1}{3}\delta$.

Thus if a STAV is $\gamma$-good, and the A vs V graph under $s$ is an $r \gamma$-sampler, then we can prove that
\[ \prob{f_s \ane{r\gamma} \rest{G}{s}} = \bigO{(1+1/r)\varepsilon}. \]
by saying that
This allows us to say that:
\[ \Prob[s]{f_s \ane{r\gamma} \rest{G}{s}} \geq 3\Prob[(s,a)]{\rest{f_s}{s\cap \adj{a}} \ane{\frac{1}{3}r\gamma} \rest{G}{s \cap \adj{a}}}.\]

I added below the assumption the definitions we need to understand it: $\delta$-samplability, and the local $A-V$ graphs (see below). I didn't put them in other places in this section because I'm not sure what you're editing and I don't want to make a mess.

Unfortunately, in order to talk about the parameter $r$ in the theorem, we need to add it here.

The ways I see to get $r\gamma$-samplability in the $A-V$ graphs is by using the sampler lemma on a $r\gamma$-bipartite expander.}}
\begin{enumerate}[start=1,label={\textbf{(A\arabic*)}}]
    \item \label{ass:global-A-V-graph} The reach graph is a $\sqrt{\gamma}$-bipartite expander.
    \item \begin{enumerate}

    \item \label{ass:edge-expander} For all $a \in A$, the $STS_a$-Graph is a $\frac{1}{3}$-edge expander.

    \item \label{ass:v-a-graph}
    For all $a \in A$ and $v \in \adj{a}$, the $STS_{(a,v)}$-graph is an $\gamma$-two-sided spectral expander. \rem{It is also a $\frac{1}{3}$-edge-expander.}
\end{enumerate}
    \item \label{ass:vasa-graphs}
    \begin{enumerate}
        \item \label{ass:vasa-graphs-v} For all $v \in V$, the $\vASA{v}$-graph is a either a $\gamma$-bipartite expander or a \(\gamma\)-two-sided spectral expander.
        \item \label{ass:vasa-graphs-a} For all $a \in A$, the $\VAS{a}$-graph is a $\sqrt{\gamma}$-bipartite expander.
    \end{enumerate}

    \labeleditem{\textbf{(A4($r$))}} \label{ass:AVS-graph-good-sampler} For all $s \in S$, the $\AV{s}$-graph is a $r\gamma$-sampler graph. Here $r>0$ is a parameter. A $r\gamma$-sampler graph is defined in \pref{def:sampling-graph}.

    \item \label{ass:a-not-large} For every pair $a,s$ so that $a \subset s$, the size of $\adj{a}$ inside $s$ is relatively large, that is
    \[ \cProb{v \sim D}{v \in \adj{a}}{v \in s} \geq \frac{1}{2}.\]

\end{enumerate}

\end{definition}
%Definition recopied to glossary, every change needsto be reflected there!

\begin{remark}
The constants  $\frac 1 2,\frac{1}{3}$ are arbitrary. In addition, in the proof of the main theorem, we will use the fact that the graphs in \pref{ass:vasa-graphs}, \pref{ass:v-a-graph} are $\frac{1}{3}$-edge expanders. By the famous Cheeger's inequality, for a small enough $\gamma$, if the graphs above are $\gamma$-spectral expanders, then they are also $\frac{1}{3}$-edge-expanders.
\end{remark}

\subsection{The surprise parameter}
Let $f=\set{f_s}_{s\in S}$ be an ensemble.
In this section we discuss an additional parameter of $f$ and the underlying STAV structure $X$ that influences the agreement theorem. This is the so-called {\em surprise} parameter. This parameter measures how surprised we are when $f_s$ and $f_{s'}$ agree on $a$ given that we already know that they disagree on $t$, where $t\supset a$. If this probability is small, we get strong amplification. This idea played an important role in several previous works and it seems useful to consider this parameter explicitly.
\begin{definition}[Surprise of an ensemble]\torestate{\label{def:surprise}
Let $X$ be a STAV structure. The surprise of a given ensemble $f=\set{f_s}$ with respect to $X$ is
\[ \surp(X,f) = \cProb{s_1,s_2,t,a,v}{ \rest{f_{s_1}}a = \rest{f_{s_2}}a\ve f_{s_1}(v)\neq f_{s_2}(v)}{\rest{f_{s_1}} t\neq \rest{ f_{s_2}}t}
\] where the probability is over choosing $s_1,t,s_2$ from the $\sts$ distribution and then choosing $(a,v)$ conditioned on $t$. Note that both $s_1,t,a,v$ and $s_2,t,a,v$ are distributed as in $\stavd$.}
\end{definition}
It is sometimes natural to restrict attention to a sub-family of ensembles which we call $\delta$-ensembles.
\begin{definition}[$\delta$ ensemble]\torestate{\label{def:delta-ensemble}
  An ensemble $f$ is a $\delta$-ensemble if for every labeled edge $(s_1,t,s_2)$ in the $\sts$ graph,
  \[\rest{f_{s_1}}t \neq \rest{ f_{s_2}}t\qquad \Longrightarrow\qquad\dist(\rest{f_{s_1}}t, \rest{f_{s_2}}t)> \delta \] (where $\dist(\cdot,\cdot)$ stands for relative hamming distance).}
\end{definition}
\begin{remark}
Note that every ensemble is a $\frac 1 {\card t}$ ensemble.
\end{remark}
\begin{remark}
Agreement theorems are often considered for special ensembles where each $f_s$ belongs to an error correcting code, such as the Reed-Muller code in the case of low degree tests. Furthermore, in the low degree test examples, for all $t\subset s$, $\rest{f_s} t$ itself belongs to an error correcting code with some distance $\delta$. Clearly, such ensembles are automatically $\delta$-ensembles.
\end{remark}
In some important cases the STAV structure itself implies a non-trivial surprise parameter for all possible ensembles. We are thus led to define the surprise of the STAV as the supremum over all possible ensembles,
\begin{definition}[Global surprise]\torestate{\label{def:global-surprise}
Let $X$ be a STAV structure. The surprise of $X$ is
\[ \surp(X) = \sup_{f}\surp(f)\,.
\]}
\end{definition}

While the agreement of $\FF$ is a property of the ensemble $\FF$, the surprise is influenced by the STAV-structure itself.
For this, the following graphs play a role:
\begin{definition}[T-Lower Graph]\torestate{\label{def:t-lower-graph}
Fix $t \in T$. The T-lower graph of $t$ is a bipartite graph where
\[L = \sett v{v \in t}, \; R = \sett a{a \subset t}, \; E = \sett{(a,v)}{v \in a}. \]
Notice that here, we require $v \in a$ and \emph{not} $v \in \adj{a}$ as we required in the STAV-structure.
The probability to choose an edge $(a,v) \in E$ is the probability of choosing $a$ given that $a \subset t$ and then choosing $v$ at random inside $a$.}
\end{definition}
A priori, the T-lower graphs need not be good expanders, as in the STAV-structures defined for \pref{thm:main-agreement-theorem-one-sided}. However, when they are, we can use their expansion properties to establish the ``surprise''.
We can give the following easy bound on the surprise parameter,
\begin{lemma}\label{lem:delta-surprise}
Let $X$ be a STAV-structure so that for every $t \in T$, the T-lower graph is a $\eta$-bipartite expander. For any $\delta$ ensemble $f$,  $ \surp(X,f) \le O(\frac{\eta^2}{\delta})$.
\end{lemma}
Before proving the lemma let us give a couple of examples demonstrating its usefulness.
\begin{example}[HDX simplicial complexes, continued]\label{example:hdx-sc2}
Consider the STAV from \pref{ex:hdx-sc}. For any $t \in X(\ell)$, the $T$-lower graph of $t$ is the graph where $R$ is the vertices of $t$, and $L$ are subsets of $t$ of size $|t| - 1$, where the edges denote containment.
The reader may calculate that this graph is a $\eta$-bipartite expander with $\eta=\frac{1}{\ell}$. Plugging in $\delta = 1/\ell$ we get $\surp(X)\le \eta^2/\delta = 1/\ell$.
\end{example}
\begin{example}[The Grassmann Poset]
Let $\F$ be a finite field, let $X$ is a STAV-structure where $V= \F^n$, $T$ is the set of $\ell$-dimensional linear subspaces of $\F^n$, $A$ is the set of $(\ell-1)$-dimensional spaces. For any $t \in X(\ell)$, the $T$-lower graph of $t$ is the graph where $R$ are the $1$-dimensional subspaces of $t$, and $L$ are the $\ell$-dimensional subspaces of $t$, where the edges denote containment.
The reader may calculate that this graph is an $\bigO{\sqrt{\frac{1}{q^{t-1}}}}$-bipartite expander. One is often interested in agreement theorems on the Grassmann poset where the local functions are promised to come from some error correcting code. In this case the ensemble $f$ will be a $\delta$-ensemble for constant $\delta$, and therefore we bound the surprise by $\surp(X)\le O(1/q^{t-1})$.
\end{example}
\begin{proof}[Proof of \pref{lem:delta-surprise}]
It suffices to show that
\[ \cProb{}{\rest{f_{s_1}}{\dl{a}{V}} = \rest{f_{s_2}}{\dl{a}{V}}}{\rest{f_{s_1}}{\dl{t}{V}} \ne \rest{f_{s_2}}{\dl{t}{V}}} = \bigO{\frac{\eta^2}{\delta}}.\]

Denote by $B = \sett{v \in \dl{t}{V}}{f_{s_1}(v) \ne f_{s_2}(v)}$. By our assumption on the distance, we are promised that $\prob{B} \geq \delta$.
And indeed, we can invoke the sampler lemma, \pref{lem:sampler-lemma}, and get that the probability of $a$ to see no vertices in $B$ is $O(\frac{\eta^2}{\delta})$.
\end{proof}
\subsection{Main theorem: agreement on STAV structures}
We are now ready to state our main technical theorem. Recall that for a given distribution $\testd$ over pairs $s_1,s_2$ we denoted by $\disagr{\testd}(f)$ the probability that $\rest{f_{s_1}}t \neq \rest{f_{s_2}}t$ when choosing $s_1,s_2 \sim \testd$ and setting $t=s_1\cap s_2$. For a given STAV $X$ we extend this notation to $\disagr{X}(f)$ understanding that the sets $s_1,t,s_2$ are now chosen via the STS distribution that comes with $X$.
\begin{theorem}[STAV Agreement Theorem] \torestate{ \label{thm:main-STAV-agreement-theorem}
Let $\Sigma$ be some finite alphabet (for example $\Sigma = \set{0,1}$).
Let \( X = (S,T,A,V)\) be a $\gamma$-good STAV-structure for some \(\gamma < \frac{1}{3}\). Let \( \FF = \sett{f_s :\dl{s}{V} \to \Sigma }{s \in S}\) be an ensemble such that
\begin{enumerate}
    \item Agreement:
    \begin{equation} \label{eq:agreement}
        \disagr{X}(f) \leq \varepsilon,
    \end{equation}
    \item Surprise:
    \begin{align} \label{eq:surprise}
        \surp(X,f)\le O(\gamma)
    \end{align}
\end{enumerate}
Then assuming either \pref{ass:AVS-graph-good-sampler} for $r=1$ or \pref{ass:a-not-large}, 
\[ \dist_\gamma(f,\glob)\le O(\varepsilon).
\]
More explicitly, there exists a global function $G: V \to \Sigma$ s.t.
\[ \Prob[s \in S]{f_s \ane{\gamma} \rest{G}{\dl{s}{V}}} \defeq \Prob[s \in S]{ \cProb{v \in V}{f_s(v) \ne \rest{G}{\dl{s}{V}}}{v \in s} \geq \gamma } = \bigO{\varepsilon}. \]
Moreover, for any $r > 0$, if either \pref{ass:AVS-graph-good-sampler} or \pref{ass:a-not-large} holds then
\begin{equation} \label{eq:thm-then}
    \Prob[s \in S]{f_s \ane{r \gamma} \rest{G}{\dl{s}{V}}} = \bigO{\left (1+\frac{1}{r} \right )\varepsilon}.
\end{equation}

The O notation \emph{does not} depend on any parameter including $\gamma,\varepsilon$, the size of the alphabet, the size of $|S|,|T|,|A|,|V|$ and, size of any $s \in S$.
}
\end{theorem} % Is restated manually on list-of-results!!!

\section{Proof of Main Theorem}\label{sec:main-thoerems-proof}
%\subsection{Theorem Statement}
In this section we prove our main theorem, \pref{thm:main-STAV-agreement-theorem}. 

We first give a direct proof for the case of two-sided high dimensional expanders, that follows the same line of general proof. Afterwards we prove the theorem in full generality.

\subsection{Proof for a Representative Case: Two-Sided High Dimensional Expanders}
\label{sec:direct-proof-hdx}
In this section we give a direct proof to a special case of our main theorem. We give a sound agreement test on set systems coming from a two-sided high dimensional expander.

We recall that a simplicial complex $X$ is a family of subsets that is downwards closed to containment, i.e. if $s \in X$ and $t \subset s$ the then $t \in X$. We denote by $X(\ell)$ all subsets (also called faces) of size $\ell+1$. We identify $X(0)$ with the set of vertices. A complex is $d$-dimensional if the largest faces have size $d+1$.
Our test is the following:
\begin{definition}[$d,\ell$-agreement distribution]\torestate{ \label{def:d-l-agreement-test}
Let $X$ be a $d$-dimensional simplicial complex and $\ell < d$ be a positive integer. We define the distribution $D_{d,\ell}$ by the following random process
\begin{enumerate}
    \item Sample $t  \in X(\ell)$.
    \item Sample $s_1,s_2 \in X(d)$ independently, given that $t \subset s_1,s_2$.
\end{enumerate}}
\end{definition}
The $d,\ell$-agreement test is the test associated with the $d,\ell$-agreement distribution on this family.

\begin{theorem}[Agreement for High Dimensional Expanders]
\torestate{\label{thm:main-agreement-theorem-two-sided-restricted-version}
There exists a constant $c>0$ such that for every $d>0$ such that the following holds. Suppose that $X$ is a $\frac{1}{d^3}$-two-sided $d$-dimensional HDX, and $\ell = \lfloor \frac{d}{3} \rfloor$.
Then the $d,\ell$-agreement test is exactly $c$-sound.
}
\end{theorem}
This theorem holds for a wider range of parameters. Also, in this section we will assume that the alphabet is binary, namely that the local functions are \(f_s: s \to \set{0,1}\). The full theorem, \pref{thm:main-agreement-theorem-two-sided}, is discussed and proven in \pref{sec:agreement-on-hdx}.

\subsubsection{Proof of \pref{thm:main-agreement-theorem-two-sided-restricted-version}}

The proof of the theorem goes through some auxiliary functions:
\begin{definition}[local popularity function]
\torestate{\label{def:h-sigmas-two-sided}
For every $a \in X(\ell-1)$ define $h_a: a \to \Sigma$ by popularity, i.e. $h_a = \plu_{s\supset a} \set{\rest{f_s}{a}}$. The notation $\plu$ refers to the value $\rest {f_s}{a}$ with highest probability over $s\supset a$, ties are broken arbitrarily. }
\end{definition}

\begin{definition}[the reach function]\torestate{
\label{def:g-sigmas-two-sided}
For every $a \in X(\ell-1)$ define $g_a: X_a(0) \to \Sigma$ by the popularity conditioned on $\rest{f_s}{a} = h_a$, i.e.
\[g_a(v) = \plu \set{f_s(v) : s \supset a, \rest{f_s}{a} = h_a}.\]
Ties are broken arbitrarily.}
\end{definition}

First, we will prove the following lemma on the local popularity functions:

\begin{lemma} \torestate{
\label{lem:h-lemma-two-sided}
    For any $a \in X(\ell-1)$, let $h_a$ be as in \pref{def:h-sigmas-two-sided}. Denote by $\varepsilon_a$ the disagreement probability given that the intersection $t \in X(\ell)$ contains \(a\). That is,
    \[\varepsilon_{a} = \cProb{}{\rest{f_{s_1}}{t} \ne \rest{f_{s_2}}{t}}{a \subset t}. \]

    Then for every $a \in X(\ell-1)$:
    \[ \cProb{s \in X(d)}{\rest{f_s}{a} \ne h_a }{s \supset a} = \bigO{\varepsilon_a}.\]
    }
\end{lemma}

Next, we move towards showing that when \(\rest{f_s}{a} = h_a\), then for a typical \(a\),  \(f_s(v) = g_a(v)\) occurs with probability \(1-\bigO{\frac{\varepsilon}{d}}\).

Consider the distribution \((a,s,a') \sim D_{comp}\), where we choose \(s \in X(d)\) and then two \(a,a' \subset s\) uniformly at random given that they are disjoint.

We say that a triple $(a,s,a')$ is \emph{bad} if $\rest{f_s}{a} \ne h_a$ or $\rest{f_s}{a'} \ne h_{a'}$. It is easy to see from \pref{lem:h-lemma} that there are $\bigO{\varepsilon}$ bad triples at most.

We use the bad triples to define the set of globally bad elements in $X(\ell-1)$. These are all $a \in X(\ell-1)$ with many bad triples touching them
\[ A^* = \sett{a \in X(\ell-1)}{\Prob[(s,a_2)]{(a,s,a_2) \text{ is a bad triple} } \geq \myconst}.\]
We shall use this set $A^*$ to filter and disregard certain $a \in X(\ell-1)$, that ruin the probability to agree with the $\set{g_a}_{a \in X(\ell-1)}$, and later on with the global function. The constant $\myconst$ is arbitrary, and once it is fixed, we can say that $\prob{A^*} = \bigO{\varepsilon}$ by Markov's inequality.

\begin{lemma}[agreement with link function] \torestate{
    \label{lem:g-lemma-two-sided}
    Let $(a,s,v) \sim D$ be the distribution where we choose \(s \in X(d)\) and from it \(a,v\) uniformly at random so that \(v \notin a\).
    Then
    \begin{equation} \label{eq:g-lemma-to-bound-two-sided}
    \Prob[(a,v,s) \sim D]{f_s(v) \ne g_a(v) \ve \rest{f_s}{a} = h_a \ve a \notin A^* } = \bigO{\frac{\epsilon}{d}}.
    \end{equation}
    }
\end{lemma}% Is restated manually later on!!!
Finally our goal is to stitch the \(g_a\)'s functions together to one global function.

\pref{lem:g-lemma} motivates us to define the global function as the popularity vote on $g_a(v)$ for all $a \in X_{v}(\ell-1)$ that see few bad triples when conditioned on \(v\). However, in order to properly define the global function, we need to define another process that takes into account the agreement of two functions $g_a, g_{a'}$. For this we need to look at each vertex \(v \in X(0)\) separately.

To do so, we define the following graph:
\begin{definition}[Local Complement Graph]
\label{def:local-complement-graph-two-sided}
Fix any \(v_0 \in X(0)\). The local complement graph \(H_{v_0}\) is the graph whose vertices are \(V = X_{v_0}(\ell-1)\). Our labeled edges are chosen as follows: Given that we are at element \(a\) we traverse to \(a'\) via edge \(s\), by choosing some \(s \supset a \dunion \set{v_0}\) and then choosing some \(a' \subset s\) given that \(a \cap a' = \emptyset\).
\end{definition}

For $v \in V$, we say $a \in X_v(\ell-1)$ is \emph{locally bad for $v$}, if
\[ \cProb{(a_1,s,a_2)\in E(H_{v})}{(a_1,s,a_2) \text{ is bad}}{a_1 = a}> \twomyconst. \]
The constant here is also arbitrary.

Finally, for every $v \in V$, we define $A^*_v$ to be the set of all $a \in X(\ell-1)$ that are either globally bad, or  locally bad for $v$.

We show using the sampler lemma, \pref{lem:sampler-lemma}, that if $a \in X(\ell-1)$ is not globally bad, then the probability over $v \in V$, that it will be locally bad for $v$ is small, i.e.
\begin{claim}[Not Globally Bad implies Not Locally Bad]\torestate{
\label{claim:not-global-bad-implies-not-local-bad-two-sided}
\[ \Prob[a \in X(\ell-1), v \in X_a(0)]{a \in A^*_v \ve a \notin A^*} = \bigO{\frac{\varepsilon}{d}}.\]
}
\end{claim}

Now we can define our global function $G:V \to \Sigma$ as follows:
\[G(v) = \plu \sett{g_a(v) }{a \in X_v(\ell-1),\; a \notin A^*_v },\]
as usual, ties are broken arbitrarily. In words, we remove a small amount of bad $a \in X(\ell-1)$, where many functions $f_s$'s don't agree with the $g_a$'s, and take the popular vote of the remainder.

Using the local complement graph and \pref{claim:not-global-bad-implies-not-local-bad-two-sided}, we can now prove:
\begin{lemma}[agreement with global function]
\torestate{
\label{lem:global-agreement-links-two-sided}
\[\Prob[a \in X(\ell-1), v \in X_0(a)]{g_a(v) \ne G(v) \ve a \notin A^*_v } = \bigO{\frac{\varepsilon}{d}}.\]
}
\end{lemma}

Given the lemmata above, we prove the theorem.
\begin{proof}[Proof of \pref{thm:main-agreement-theorem-two-sided-restricted-version}]
We note that it is enough to show 
\begin{equation} \label{eq:real-thm-then-two-sided}
    \Prob[s \in X(d), a \in X(\ell-1), a \subset s]{\rest{f_s}{s \setminus a} \ne \rest{G}{s \setminus a}} = \bigO{\varepsilon}.
\end{equation}
This is due to the fact that \(\abs{s \setminus a} \geq \frac{1}{2}\abs{s}\), thus if \(f_s \ne \rest{G}{s}\), then \(\rest{f_s}{s \setminus a} \ne \rest{G}{s \setminus a}\) for at least half of the possible \(a \subset s\).

Next, we prove \eqref{eq:real-thm-then-two-sided}.
We define the following events, when we choose \((a,s,v)\) in the simplicial complex:
\begin{enumerate}
    \item $E_1$ - the event that $\rest{f_s}{a} \ne h_a$.
    \item $E_2$ - the event that $a \in A^*$, i.e. the $a$ chosen has many bad edges.
\end{enumerate}
Define a random variable $Z$, that samples $s,a$ and outputs
\begin{equation}\label{eq:z-prob-two-sided}
Z(s,a) = \Prob[v \in s \setminus a]{f_s(v) \ne G(v)},
\end{equation}
i.e. the fraction of vertices in \(s \setminus a\) so that \(f_s(v) \ne G(v)\).

The probability for $E_1 \vee E_2$ is $\bigO{\varepsilon}$ by \pref{lem:h-lemma-two-sided} and Markov's inequality.

If $\neg (E_1 \vee E_2)$, yet a vertex $v$ contributes to the probability in \eqref{eq:z-prob-two-sided}, then one of the three must occur:
\begin{enumerate}
    \item $a \in A^*_v$.
    \item $f_s(v) \ne g_a(v)$ and $a \notin A^*_v$.
    \item $a \notin A^*_v$ but $f_s(v) = g_a(v) \ne G(v)$.
\end{enumerate}
The first event occurs with probability $\bigO{\frac{\varepsilon}{d}}$ by \pref{claim:not-global-bad-implies-not-local-bad-two-sided}. The second occurs with probability $\bigO{\frac{\varepsilon}{d}}$ by \pref{lem:g-lemma-two-sided}. The third occurs with probability $\bigO{\frac{\varepsilon}{d}}$ by \pref{lem:global-agreement-links-two-sided}.
Thus by the expectation of $Z$ given that $\neg (E_1 \vee E_2)$ is $\bigO{\frac{\varepsilon}{d}}$. By Markov's inequality
\[\cProb{s \in X(d), a \in X(\ell-1), a \subset s}{\rest{f_s}{s \setminus a} \ne \rest{G}{s \setminus a}}{\neg (E_1 \vee E_2)} = \cProb{}{Z \geq \frac{1}{d}}{\neg (E_1 \vee E_2)} \]
\[= \abs{s \setminus a} \bigO{\frac{\varepsilon}{d}} = \bigO{\varepsilon}.\]

In conclusion
\[\Prob[s \in X(d)]{f_s \ne \rest{G}{s}} \leq \prob{E_1 \vee E_2} + \cProb{s \in X(d), a \in X(\ell-1), a \subset s}{\rest{f_s}{s \setminus a} \ne \rest{G}{s \setminus a}}{\neg (E_1 \vee E_2)} = \bigO{\varepsilon}.\]
\end{proof}

\subsubsection{Proof of the Lemmata}

\restatelemma{lem:h-lemma-two-sided}
\begin{proof}[Proof of \pref{lem:h-lemma-two-sided}]
Fix $a \in X(\ell-1)$. If $\varepsilon_a \geq \frac{1}{6}$ we are trivially done, so assume otherwise. Consider the following graph:
\begin{enumerate}
    \item The elements in the graph are all \(s \supset a\).
    \item We connect two elements \(s_1,s_2\) whenever there exists some \(t \in X(\ell)\), \(t \supset a\) so that \(s_1 \cap s_2 \supset t\).
\end{enumerate}
The random walk in this graph, given \(s_1\) traverses to \(s_2\) by the \(d,\ell\)-agreement test's distribution, given that the intersection contains \(a\).

By \pref{thm:containment-graph-is-a-good-expander}, this graph is a very good spectral expander. In particular, it is a \(\frac{1}{3}\)-edge expander, when \(d\) is sufficiently large. 

We color the vertices of this graph according to their value at \(a\). Denote by \(S_1,S_2,...\) the colors, where \(S_1\) is the largest. Namely, \(S_1\) are all the \(s\) so that \(\rest{f_s}{a} = h_a\).

Denote by $S_i = \set{s :\rest{f_s}{a} = h^i_a}$. We need to show that the set of vertices $S_1 = \set{s : \rest{f_s}{a} = h_a}$ (the largest of all $S_i$) is $1 - \bigO{\varepsilon_a}$.

The quantity $\varepsilon_a$, i.e. the amount of edges between $S_i$'s, is by assumption less than $\frac{1}{6}$.

%It is a known fact that if we partition a vertex of an edge-expander graph, and there are few outgoing edges, then one of the parts in the partition is large. This fact is formulated in \pref{claim:edge-expander-partition-property}.

By \pref{claim:edge-expander-partition-property}, using the fact that the graph is a $\frac{1}{3}$-edge expander and the fact that the fraction of edges between the $S_i$'s is less that $\frac{1}{6}$. We get that $\prob{S_1} \geq \frac{1}{2}$.

Furthermore, by the edge-expander property $\prob{S_1^c} \leq 3E(S_1,S_1^c) \leq 3\varepsilon_a$.
\end{proof}

\begin{corollary}
\label{cor:a-star-is-small-two-sided}
$\prob{A^*} = \bigO{\varepsilon}$.
\end{corollary}

\begin{proof}[Proof of \pref{cor:a-star-is-small-two-sided}]
Each $a \in X(\ell-1)$ contributes to $A^*$ if the amount of bad triples that $a$ participates in is $\geq \frac{1}{40}$. The total amount of bad triples is $\bigO{\varepsilon}$ by \pref{lem:h-lemma-two-sided}. Thus by Markov's inequality $\prob{A} = \bigO{\varepsilon}$.
\end{proof}

\bigskip

We move to \pref{claim:not-global-bad-implies-not-local-bad-two-sided}.
\restateclaim{claim:not-global-bad-implies-not-local-bad-two-sided}
\begin{proof}[Proof of \pref{claim:not-global-bad-implies-not-local-bad-two-sided}]

Fix some $a \notin A^*$. Consider the following bipartite graph:
    \begin{itemize}
        \item $L = \set{(a', s) : a' \dunion a \subset s}.$
        \item $R = X_a(0).$
        \item $E = \set{(v,(a',s)) : \set{v} \dunion a' \dunion a \subset s},$
    \end{itemize}
 The probability to choose each edge is given by the following distribution in the link \(X_a)\):
 \begin{enumerate}
     \item Sample \(v \in X_a(0)\).
     \item Sample \(s \setminus a \in X_a(d - \ell)\) so that \(v \in s\).
     \item Sample \(a' \in X_a(\ell-1)\) so that \(a' \subset s \setminus \set{v}\).
 \end{enumerate}

Note that the probability of \((a',s)\) in the left side, is precisely the probability to choose the triple \((a,s,a') \sim D_{comp}\), given that the first element is \(a\). 

Denote by $B \subset L$ the that consists of all $(s,a')$ s.t. $(a,s,a')$ is a bad triple. If $a \notin A^*$ then $\prob{B} < \myconst$.

By \pref{prop:two-sided-graph}, this graph is a \(\bigO{\frac{1}{\sqrt{d}}}\)-bipartite expander.

Define the set
\[V^* = \sett{v \in \adj{a}}{\cProb{(s,a')}{B}{v \sim (s,a')} \geq \twomyconst}, \]
the set of $v \in \adj{a}$ so that the probability for a bad edge is larger than $\twomyconst$, namely, that $a$ is locally bad for $v$. 

In the sampler lemma, \pref{lem:sampler-lemma}, we see that bipartite-expanders are good samplers. We use \pref{lem:sampler-lemma} to get that $\prob{V^*} = \bigO{\frac{1}{d}}\prob{B}$. Taking expectation on all $a \in A$ we get that
\[ \Prob[a \in X(\ell-1), v \in X_a(0)]{a \in A^*_v \ve a \notin A^*} = \prob{a \notin A^*} \cdot \Ex[a \notin A^*]{\prob{V^*} } =\]
\[ \prob{a \notin A^*} \cdot \bigO{\frac{1}{d}}\Ex[a \notin A^*]{\prob{B}} \leq \bigO{\frac{1}{d}} \Ex[a \in A]{\prob{B}}= \bigO{\frac{\varepsilon}{d}},\]
The last inequality is due to the fact that taking expectation on this set conditioned on $a \notin A^*$, is less than the expectation on all $A$ (by definition when $a \in A^*$, then $\prob{B} \geq \frac{1}{40}$, and when $a \notin A^*$, $\prob{B} < \frac{1}{40}$).
The last equality is since $\prob{B} = \bigO{\varepsilon}$ by \pref{cor:a-star-is-small-two-sided}.
\end{proof}

\bigskip

We move towards proving \pref{lem:g-lemma-two-sided}.
We shall use the following ``surprise'' claim.
\begin{claim}[Surprise]
\label{claim:agreement-on-last-vertex-two-sided}
Let \(\hat{D}\) denote the distribution where we sample:
\begin{enumerate}
    \item \(a \in X(\ell-1)\).
    \item \(v \in X_a(0)\).
    \item \(s_1,s_2 \in X(d)\) independently, given that they contain \(t = a \dunion \set{v}\).
\end{enumerate}
Then
\[ \Prob[\hat{D}]{f_{s_1}(v) \ne f_{s_2}(v) \ve \rest{f_{s_1}}{a} = \rest{f_{s_2}}{a}} = \bigO{\frac{\varepsilon}{d}}.\]
\end{claim}
This claim is given in full generality in that is given in \pref{sec:HDX}. For this section to be self contained, we give it an elementary proof:
\begin{proof}[Proof of \pref{claim:agreement-on-last-vertex-two-sided}]
\(\hat{D}\) can be described as first choosing \(s_1,s_2,t\) and then partitioning \(t=a \dunion \set{v}\). So from the law of total probability we obtain:
\[\Prob[\hat{D}]{f_{s_1}(v) \ne f_{s_2}(v) \ve \rest{f_{s_1}}{a} = \rest{f_{s_2}}{a}} = \] 
\[ \Ex[(t,s_1,s_2) ]{\Prob[v \in t, a = t \setminus \set{v} ]{f_{s_1}(v) \ne f_{s_2}(v) \ve \rest{f_{s_1}}{a} = \rest{f_{s_2}}{a} }}.\]
Notice that for every $t \in X(\ell)$, the $\set{s_1,s_2}$ pairs that contribute to the probability above, are the ones that fail the test (but do so on exactly one vertex). By the agreement test, there are at most an \(\varepsilon\)-fraction of such pairs. Given that we choose such a pair, their contribution to the expectation is $\frac{1}{\ell} = \bigO{\frac{1}{d}}$ since that is the probability of choosing the $v \in t$ s.t. $f_{s_1}(v) \ne f_{s_2}(v)$.
\end{proof}

Now we are ready to prove \pref{lem:g-lemma-two-sided}.
\begin{lemma*}[Restatement of \pref{lem:g-lemma-two-sided}]
    Let $(a,s,v) \sim D$ be the distribution where we choose \(s \in X(d)\) and from it \(a,v\) uniformly at random so that \(v \notin a\).
    Then
    \begin{equation*} %\label{eq:g-lemma-to-bound-two-sided}
    \Prob[(a,v,s) \sim D]{f_s(v) \ne g_a(v) \ve \rest{f_s}{a} = h_a \ve a \notin A^* } = \bigO{\frac{\epsilon}{d}}.
    \end{equation*}
\end{lemma*}

The proof we give here relies on the fact that the alphabet is binary, or at least of size \(\bigO{1}\). It is possible to prove this for an alphabet of unbounded size, as we do in the main proof.

\begin{proof}[Proof of \pref{lem:g-lemma-two-sided}]
First, note that by \pref{claim:not-global-bad-implies-not-local-bad-two-sided}, \eqref{eq:g-lemma-to-bound-two-sided} is less or equal to
\[ \prob{a \notin A^* \ve a \in A^*_v} + \Prob[(a,v,s) \sim D]{f_s(v) \ne g_a(v) \ve \rest{f_s}{a} = h_a \ve a \notin A^*_v} = \]
\[ \bigO{\frac{\varepsilon}{d}} + \Prob[(a,v,s) \sim D]{f_s(v) \ne g_a(v) \ve \rest{f_s}{a} = h_a \ve a \notin A^*_v}.\]

Thus we focus on bounding
\begin{equation}
    \label{eq:g-lemma-to-bound-two-sided-local}
    \Prob[(a,v,s) \sim D]{f_s(v) \ne g_a(v) \ve \rest{f_s}{a} = h_a \ve a \notin A^*_v}.
\end{equation}

We write the expression we want to bound in \eqref{eq:g-lemma-to-bound-two-sided-local} as
\[\Ex[a,v]{\Prob[s]{f_s(v) \ne g_a(v) \ve \rest{f_s}{a} \ne h_a \ve a \notin A^*_v}}.\]
We denote the expression inside the expectation as \[p_{a,v} = \Prob[s]{f_s(v) \ne g_a(v) \ve \rest{f_s}{a} \ne h_a \ve a \notin A^*_v}.\]
Thus we want to show that
\[\Ex[a,v]{p_{a,v}} = \bigO{\frac{\varepsilon}{d}}.\]

By \pref{claim:agreement-on-last-vertex-two-sided}, we got that 
\[ \Prob[(a,v,s_1,s_2) \sim \hat{D}]{f_{s_1}(v) \ne f_{s_2}(v) \ve \rest{f_{s_1}}{a} = \rest{f_{s_2}}{a}} = \bigO{\frac{\varepsilon}{d}}.\]
We can write this also as an expectation over \(a,v\):
\[\Ex[a,v]{\Prob[(s_1,s_2)]{f_{s_1}(v) \ne f_{s_2}(v) \ve \rest{f_{s_1}}{a} = \rest{f_{s_2}}{a}} } = \bigO{\frac{\varepsilon}{d}}.\]
We denote the expression in the expectation by
\[q_{a,v} = \Prob[(s_1,s_2)]{f_{s_1}(v) \ne f_{s_2}(v) \ve \rest{f_{s_1}}{a} = \rest{f_{s_2}}{a}}.\]

Our goal is to relate the two quantities, namely, to show that \(p_{a,v} = \bigO{q_{a,v}}\). This will show that
\[\Ex[a,v]{p_{a,v}} = \bigO{\Ex[a,v]{q_{a,v}}} = \bigO{\frac{\varepsilon}{d}}.\]

\medskip

Fix some \(a \in X(\ell-1)\) and \(v \in X_a(0)\). If \(a \in A^*_v\) then \(p_{a,v} = 0\) and we are done. So assume \(a \notin A^*_v\).

Denote by \(H_0\) the set of all \(s \supset t=a \dunion \set{v}\). In the sampling process for \(p_{a,v}\) we choose some \(s \in H_0\), and in the sampling process for \(q_{a,v}\) we choose \(s_1,s_2 \in H_0\) independently. 

We can partition \(H_0\) to
\[H_0 = G \dunion B,\]
where \(G\) contains all \(s \in H_0\) so that \(\rest{f_s}{a} = h_a\). \(B\) is all \(s \in H_0\) so that \(\rest{f_s}{a} \ne h_a\).

\(a \notin A_v^*\), thus 
\[\Prob[s \in H_0]{B} < \twomyconst,\]
or 
\[\Prob[s \in H_0]{G} > \frac{19}{20}.\]
Thus, conditioning on \(G\) doesn't change the probability of \(q_{a,v}\) significantly, namely
\[\cProb{s_1,s_2}{\rest{f_{s_1}}{a} = \rest{f_{s_2}}{a} \ve f_{s_1}(v) \ne f_{s_2}(v)}{s_1,s_2 \in G} \leq 2q_{a,v}.\]
The first equality in the probability, \(\rest{f_{s_1}}{a} = \rest{f_{s_2}}{a}\), is immediately satisfied in this set, since if \(s_1,s_2 \in G\) then \(\rest{f_{s_1}}{a} = h_a = \rest{f_{s_2}}{a}\). So we get
\[ \cProb{s_1,s_2}{f_{s_1}(v) \ne f_{s_2}(v)}{s_1,s_2 \in G} \leq 2q_{a,v}.\]
Because \(s_1,s_2\) are chosen independently, we can say that
\[\cProb{s_1,s_2}{f_{s_1}(v) \ne f_{s_2}(v)}{s_1,s_2 \in G} = \]
\[2\cProb{s_1}{f_{s_1}(v) = g_a(v)}{s_1 \in G} \cdot \cProb{s_2}{f_{s_2}(v) \ne g_a(v)}{s_2 \in G}.\footnote{We are using the fact that the \(f_s\)'s are binary.}\]
The definition of \(g_a(v)\) is taking the majority of \(f_{s}(v)\) for all \(s \in G\). Thus \(\cProb{s_1}{f_{s_1}(v) = g_a(v)}{s_1 \in G} \geq \frac{1}{2}\).
\[ \cProb{s_1,s_2}{f_{s_1}(v) \ne f_{s_2}(v)}{s_1,s_2 \in G} \geq  \cProb{s_2}{f_{s_2}(v) \ne g_a(v)}{s_2 \in G} \geq p_{a,v}.\]
The last inequality is by the definition of \(p_{a,v}\). In conclusion, \(p_{a,v} \leq 2q_{a,v}\) and we are done.
\end{proof}

We state this immediate corollary:
\begin{corollary} \label{cor:key-cor-two-sided}
Consider the following distribution \((v,a,s,a') \sim D_{vasa}\), where \((a,s,a')\) are chosen by \(D_{comp}\) and \(v\) is sampled from \(s \setminus (a \dunion a')\) uniformly at random.
Then
\[\Prob[(v,a,s,a') \sim D_{vasa}]{\rest{f_s}{a_i} = h_{a_i} \ve g_{a_1} \ne g_{a_2} \ve a_i \notin A^*_v \text{ for } i=1,2} = \bigO{\frac{\varepsilon}{d}}.\]
\(\qed\)
\end{corollary}
The proof for this corollary is just applying \pref{lem:g-lemma-two-sided} for each \(a_i\) and using a union bound.

It remains to prove \pref{lem:global-agreement-links-two-sided}.
\restatelemma{lem:global-agreement-links-two-sided}

For the proof of the lemma, we'll need the following property of expander graphs. In an expander graph, the number of outgoing edges from some $A \subset V$, is an approximation to the size of $A$ or $V \setminus A$. The following claim generalizes this fact to the setting where we count only outgoing edges from $A$ to a (large) set $B \subset V \setminus A$.

\begin{claim*}[Restatement of \pref{claim:almost-cut}]
    Let $G = (V,E)$ be a $\lambda$-two sided spectral expander. Let $V = A \dunion B \dunion C$, s.t. $\prob{A} \leq \prob{B}$. Then
    \begin{equation} %\label{eq:eml-bound}
            \prob{A} \leq \frac{1}{(1-\lambda)\prob{B}} \left ( \prob{E(A,B)} + \lambda \prob{C} \right ).
    \end{equation}

    In particular, if  \( \prob{A}, 1-\lambda = \Omega(1)\) then
    \[ \prob{A} = \bigO{  \prob{E(A,B)} + \lambda \prob{C} }. \]
\end{claim*}%This is manually fixed. Need to fix macro

\begin{proof}[Proof of \pref{lem:global-agreement-links-two-sided}]
Fix some \(v_0 \in X(0)\). We view the local complement graph \(H_0\) from \pref{def:local-complement-graph-two-sided}.

The walk on this graph is the \(\ell-1,\ell-1\)-complement walk in the link of \(v\). By \pref{thm:complement-walk-is-a-good-expander}, that we will show in \pref{sec:complement-walk}, this graph is a \(\bigO{\frac{1}{d}}\)-two-sided spectral expander.

Consider the following sets in this graph:
\[ M_{v_0} = \sett{a \in X_{v_0}(\ell-1) \setminus A^*_{v_0}}{g_a(v_0) =   G(v)}, \text{ the popular vote, }\]
\[ N_{v_0} = \sett{a \in X_{v_0}(\ell-1) \setminus A^*_{v_0}}{g_a(v_0) \ne G(v)}, \text{ the other votes, }\]
\[ C_{v_0} = A^*_{v_0}, \]
The $a \in N_{v_0}$ are those where $g_a(v_0) \ne G(v_0)$ and $a \notin A^*_{v_0}$. Hence we need to bound
\[ \Ex[v_0]{\prob{N_{v_0}}}.\]

We invoke \pref{claim:almost-cut} for $N_{v_0},M_{v_0},C_{v_0}$ and get that
%\[ (1-\bigO{\frac{1}{d}})\prob{M_{v_0}} \prob{N_{v_0}} \leq \prob{E(N_{v_0},M_{v_0})} + \bigO{\frac{1}{d}} \prob{C_{v_0}}, \]
%or
\begin{equation} \label{eq:minority-bound-two-sided}
\prob{N_{v_0}} \leq \frac{1}{(1-\bigO{\frac{1}{d}})\prob{M_{v_0}}} \prob{E(N_{v_0},M_{v_0})} + \bigO{\frac{1}{d}} \prob{C_{v_0}}.
\end{equation}

The proof now has two steps:
\begin{enumerate}
    \item We show that $\prob{M_{v_0}} \geq \frac{9}{20}$ for all but $\bigO{\frac{\varepsilon}{d}}$ of the vertices $v_0$ (the constant is arbitrary). This will imply that the denominator in \eqref{eq:minority-bound-two-sided} is larger than some constant (say \(\frac{1}{2}\)).
    \item We show that the right hand side of \eqref{eq:minority-bound-two-sided} is $\bigO{\frac{\varepsilon}{d}}$ in expectation.
\end{enumerate}

\paragraph{To show step 1} We will need to show that for all but $\bigO{\frac{\varepsilon}{d}}$ of the $v_0$, the size of $C_{v_0}$ is smaller than $\frac{1}{20}$, namely

 \begin{equation}\label{eq:locally-bad-set-is-small-two-sided}
   \Prob[v]{\prob{A^*_v} > \twomyconst } = \bigO{\frac{\varepsilon}{d}}
\end{equation}

Assuming that for $\prob{C_{v_0}} \leq \frac{1}{20}$, it is obvious that $\prob{M_{v_0}} \geq \frac{9}{20}$, using the fact that the alphabet is binary in this special case, thus \(M_{v_0}\) is the larger set between \(M_{v_0}, N_{v_0}\). 

\bigskip

To show \eqref{eq:locally-bad-set-is-small-two-sided} consider the complement graph between $X(0)$ and $X(\ell-1)$, where the edges are all \((v,a)\) so that \(\set{v} \dunion a \in X(\ell)\). This is the \(0,(\ell-1)\)-complement walk. 

The set of vertices $v$ that we need to bound is the set of $v$'s with large $\prob{A_v^*}>\twomyconst$. 
There are two types of vertices $v$:
\begin{itemize}
    \item $\prob{A_v^* \cap A^*}\leq \myconst$
    \item $\prob{A_v^* \cap A^*}> \myconst$
    \end{itemize}
By \pref{claim:not-global-bad-implies-not-local-bad-two-sided}, $\Prob[(a,v)]{a \in A^*_v \ve a \notin A^*} = \bigO{\frac{\varepsilon}{d}}$. Thus by Markov's inequality, only   $\bigO{\frac{\varepsilon}{d}}$ of the vertices can see $\myconst$ fraction of neighbors $a\in A^*_v\setminus A^*$, thus bounding by $\bigO{\frac{\varepsilon}{d}}$ the fraction of $v$'s of the first type.

To bound the vertices of the second type, note that these are vertices that have a large $(>\myconst)$ fraction of neighbors in $A^*$. By \pref{cor:a-star-is-small-two-sided}, $\prob{A^*} = \bigO{\varepsilon}$. 
According to \pref{thm:complement-walk-is-a-good-expander}, our graph is a $\sqrt{\frac{1}{d}}$-bipartite expander. Thus by the sampler lemma \pref{lem:sampler-lemma}, the set of vertices $v_0 \in X(0)$ who have more than $\myconst$-fraction neighbours in $A^*$, is $\bigO{\frac{\varepsilon}{d}}$.

\paragraph{As for step 2} Taking expectation on \eqref{eq:minority-bound-two-sided} we have that

\[ \ex{\prob{N_{v_0}}} \leq \ex{\frac{1}{(1-\bigO{\frac{1}{d}})\prob{M_{v_0}}} \prob{E(N_{v_0},M_{v_0})}} + \bigO{\frac{1}{d}} \ex{\prob{C_{v_0}}}\]

\begin{equation} \label{eq:expectation-minority-bound-two-sided}
\leq \Prob[v]{\prob{A^*_v} > \twomyconst } + \ex{4\prob{E(N_{v_0},M_{v_0})}}  + \bigO{\frac{1}{d}} \ex{\prob{C_{v_0}}},
\end{equation}

where the second inequality is due to the fact that when  \(\Prob[v]{\prob{A^*_v} \leq \twomyconst }\) then
\[\frac{1}{(1-\bigO{\frac{1}{d}})\prob{M_{v_0}}} \leq 4.\]

We bound each of the terms on the right hand side of \eqref{eq:expectation-minority-bound-two-sided} separately.

By \eqref{eq:locally-bad-set-is-small-two-sided},
\[\Prob[v]{\prob{A^*_v} > \twomyconst } = \bigO{\frac{\varepsilon}{d}}.\]

By \pref{cor:a-star-is-small-two-sided} and \pref{claim:not-global-bad-implies-not-local-bad-two-sided} \[\bigO{\frac{1}{d}}\Ex[v_0]{\prob{C_{v_0}}} = \bigO{\frac{1}{d}}\Ex[v]{\prob{A^*_v}} = \bigO{\frac{\varepsilon}{d}}.\]

We continue to bound $\prob{E(N_{v_0},M_{v_0})}$ in expectation. Every edge counted in $E(N_{v_0},M_{v_0})$ is either a bad triple (i.e. and edge $(a_1,s,a_2)$ s.t. $\rest{f_s}{a_i} \ne h_{a_i}$ for $i=1$ or $2$), or a non-bad edge (an edge who is not bad) for which we see a disagreement. By \pref{cor:key-cor-two-sided} there are $\bigO{\frac{\varepsilon}{d}}$ non-bad edges in the cut.

As for the bad edges, notice that $a \in N_{v_0}$ is not a member of $A^*_{v_0}$, thus the amount of bad edges that are connected to $a$ is at most $\twomyconst$-fraction of the edges connected to $a$ (by definition). Thus the amount of bad edges is bounded by $\fourmyconst\prob{N_{v_0}}$, and
\[\prob{E(N_{v_0},M_{v_0})} \leq \bigO{\frac{\varepsilon}{d}} + \fourmyconst\prob{N_{v_0}}.\]

By summing up the bounds we get that
\[ \ex{\prob{N_{v_0}}} \leq \bigO{\frac{\varepsilon}{d}} + \frac{4}{20}\ex{\prob{N_{v_0}}}\]
hence
\[ \ex{\prob{N_{v_0}}} = \bigO{\frac{\varepsilon}{d}}.\]
\end{proof}

\subsection{Proof for the General Case}
Next we prove \pref{thm:main-STAV-agreement-theorem} in full generality.

The proof of the theorem goes through these auxiliary functions:
\begin{definition}[local popularity function]
\torestate{\label{def:h-sigmas}
For every $a \in A$ define $h_a: \dl{a}{V} \to \Sigma$ by popularity, i.e. $h_a = \plu_{s\supset a} \rest{f_s}{a}$. The notation $\plu$ refers to the value $\rest {f_s}a$ with highest probability over $s\supset a$, ties are broken arbitrarily. }
\end{definition}

\begin{definition}[the reach function]\torestate{
\label{def:g-sigmas}
For every $a \in A$ define $g_a: \adj{a} \to \Sigma$ by the popularity conditioned on $\rest{f_s}{a} = h_a$, i.e.
\[g_a(v) = \plu \set{f_s(v) : a \subset s, \rest{f_s}{a} = h_a}.\]
Ties are broken arbitrarily.}
\end{definition}

First, We will prove the following lemma on the local popularity functions:

\begin{lemma} \torestate{
\label{lem:h-lemma}
    For any $a \in A$, let $h_a$ be as in \pref{def:h-sigmas}. Denote by $\varepsilon_a$ the disagreement probability given that $t' \supset a$, i.e.
    \[\varepsilon_{a} = \cProb{}{\rest{f_{s_1}}{\dl{t'}{V}} \ne \rest{f_{s_2}}{\dl{t'}{V}}}{t' \in \ul{a}{t}}. \]

    Then for every $a \in A$:
    \[ \Prob[s \in \ul{a}{s}]{\rest{f_s}{\dl{a}{V}} \ne h_a } = \bigO{\varepsilon_a}.\]
    }
\end{lemma}

Next, we move towards showing that when \(\rest{f_s}{a} = h_a\), then for a typical \(a\),  \(f_s(v) = g_a(v)\) occurs with probability \(1-\bigO{\gamma \varepsilon}\).

We consider the $\VASA$-distribution.
We say that a triple $(a,s,a')$ is \emph{bad} if $\rest{f_s}{\dl{a}{V}} \ne h_a$ or $\rest{f_s}{\dl{a'}{V}} \ne h_{a'}$, in the context of the $\vASA{v}$-graphs defined in \pref{sec:stav}, we call these triples bad edges, since the edges of the $\vASA{v}$-graphs correspond to triples $(a,s,a')$. It is easy to see from \pref{lem:h-lemma} that there are $\bigO{\varepsilon}$ bad edges at most.

We use the bad triples to define the set of globally bad elements in $A$, to be all $a \in A$ with many bad triples touching them
\[ A^* = \sett{a \in A}{\Prob[(s,a_2)]{(a,s,a_2) \text{ is a bad edge} } \geq \myconst}\]
namely, all the $a \in A$ so that the probability in \pref{lem:h-lemma} given that we chose a fixed $a \in A$, is larger than a constant. We shall use this set $A^*$ to filter and disregard certain $a \in A$, that ruin the probability to agree with the $\set{g_a}_{a \in A}$, and later on with the global function. The constant $\myconst$ is arbitrary, and once it is fixed, we can say that $\prob{A^*} = \bigO{\varepsilon}$ by Markov's inequality.

\begin{lemma}[agreement with link function] \torestate{
    \label{lem:g-lemma}
    Let $D$ be a distribution over $(a,s,v) \in A \times S \times V$ defined by the STAV-structure, that is:
    \begin{enumerate}
        \item Choose some $(a,v)$ where $v \in \adj{a}$.
        \item Choose some $(a,v) \subset s$ (where we mean $\set{v},a \subset s$).
    \end{enumerate}
    Then
    \begin{equation} \label{eq:g-lemma-to-bound}
    \Prob[(a,v,s) \sim D]{f_s(v) \ne g_a(v) \ve \rest{f_s}{\dl{a}{V}} = h_a \ve a \notin A^*_v } = \bigO{\gamma \epsilon}.
    \end{equation}
    }
\end{lemma}
Finally our goal is to stitch\(g_a\)'s functions together to one global function.

\pref{lem:g-lemma} motivates us to define the global function as the popularity vote on $g_a(v)$ for all $a \notin A^*_v$ such that $v\in \adj a$. However, in order to properly define the global function, we need to define another process that takes into account the agreement of two functions $g_a, g_{a'}$. For this we use the $\vASA{v}$-graphs described in \pref{ass:vasa-graphs-v}.

For $v \in V$, we say $a \in \adj{v}$ is \emph{locally bad for $v$}, if
\[ \cProb{(a_1,s,a_2)\in E(\vASA{v})}{(a_1,s,a_2) \text{ is bad}}{a_1 = a}> \twomyconst. \]
The constant here is also arbitrary.

Finally, for every $v \in V$, we define $A^*_v$ to be the set of all $a \in \adj{v}$ that are either globally bad, or  locally bad for $v$.

We show using the sampler lemma, \pref{lem:sampler-lemma}, that if $a \in A$ is not globally bad, then the probability over $v \in V$, that it will be locally bad for $v$ is small, i.e.
\begin{claim}[Not Globally Bad implies Not Locally Bad]\torestate{
\label{claim:not-global-bad-implies-not-local-bad}
\[ \Prob[a \in A, v \in \adj{a}]{a \in A^*_v \ve a \notin A^*} = \bigO{\gamma \varepsilon}.\]
}
\end{claim}

Now we can define our global function $G:V\to \Sigma$ as follows:
\[G(v) = \plu \sett{g_a(v) }{a \in \adj{v},\; a \notin A^*_v },\]
as usual, ties are broken arbitrarily. In words, we remove a small amount of bad $a \in A$, where many functions $f_s$'s don't agree with the $g_a$'s, and take the popular vote of the remainder.

We can now prove:
\begin{lemma}[agreement with global function]
\torestate{
\label{lem:global-agreement-links}
\[\Prob[a \in A, v \in \adj{a}]{g_a(v) \ne G(v) \ve a \notin A^*_v } = \bigO{\gamma \varepsilon}.\]
}
\end{lemma}

Given the lemmata above, we prove the theorem for STAV-structures.
\begin{proof}[Proof of \pref{thm:main-STAV-agreement-theorem}]
We first show that based on \pref{ass:a-not-large} or \pref{ass:AVS-graph-good-sampler}, it is enough to prove
\begin{equation} \label{eq:real-thm-then}
    \Prob[s \in S, a \in \dl{s}{A}]{\rest{f_s}{s \cap \adj{a}} \ane{\frac{1}{2} r \gamma} \rest{G}{s \cap \adj{a}}} = \bigO{\left (1+\frac{1}{r} \right )\varepsilon}.
\end{equation}
Indeed for any $r > 0$,
\begin{itemize}
    \item If \pref{ass:a-not-large} holds, and 
    \[{f_s \ane{r \gamma} \rest{G}{s}}.\]
    it implies that
\[\rest{f_s}{s \cap \adj{a}} \ane{\frac{1}{2}r \gamma} \rest{G}{s \cap \adj{a}}\]
for all $a \subset s$. Thus there cannot be more than $\bigO{\left (1+\frac{1}{r} \right )\varepsilon}$ $s \in S$ as above.

\item If \pref{ass:AVS-graph-good-sampler} holds for $r \gamma$, then for any \[\rest{f_s}{s \cap \adj{a}} \ane{ r \gamma} \rest{G}{s \cap \adj{a}},\]
it is true by the assumption that a $\frac{2}{3}$-fraction of the $a \subset s$ have the property that
\[\rest{f_s}{s \cap \adj{a}} \ane{\frac{1}{3}r \gamma} \rest{G}{s \cap \adj{a}}.\]
Hence
    \begin{equation*}
    \Prob[s]{f_s \ane{r \gamma} \rest{G}{s}}
    \leq
    \frac{3}{2}\Prob[s \in S, a \in \dl{s}{A}]{\rest{f_s}{s \cap \adj{a}} \ane{\frac{1}{3} r \gamma} \rest{G}{s \cap \adj{a}}} = \bigO{\left (1+\frac{1}{r} \right )\varepsilon}
    \end{equation*}
    and we are done.
\end{itemize}

Next, we prove \eqref{eq:real-thm-then}.
We define the following events:
\begin{enumerate}
    \item $E_1$ - the event that $\rest{f_s}{a} \ne h_a$.
    \item $E_2$ - the event that $a \in A^*$, i.e. the $a$ chosen has many bad edges.
\end{enumerate}
Define a random variable $Z$, that samples $s,a$ and outputs
\begin{equation}\label{eq:z-prob}
\Prob[v \in s \cap \adj{a}]{f_s(v) \ne G(v)}.
\end{equation}
The probability for $E_1 \vee E_2$ is $\bigO{\varepsilon}$ by \pref{lem:h-lemma} and Markov's inequality.

If $\neg (E_1 \vee E_2)$, yet a vertex $v$ contributes to the probability in \eqref{eq:z-prob}, then one of the three must occur:
\begin{enumerate}
    \item $a \in A^*_v$.
    \item $f_s(v) \ne g_a(v)$ and $a \notin A^*$.
    \item $a \notin A^*_v$ but $f_s(v) = g_a(v) \ne G(v)$.
\end{enumerate}
The first event occurs with probability $\bigO{\gamma \varepsilon}$ by \pref{claim:not-global-bad-implies-not-local-bad}. The second occurs with probability $\bigO{\gamma \varepsilon}$ by \pref{lem:g-lemma}. The third occurs with probability $\bigO{\gamma \varepsilon}$ by \pref{lem:global-agreement-links}.
Thus by the expectation of $Z$ given that $\neg (E_1 \vee E_2)$ is $O(\gamma \varepsilon)$. By Markov's inequality for any $r > 0$,
\[\cProb{s \in S, a \in \dl{s}{A}}{\rest{f_s}{s \cap \adj{a}} \ane{r\gamma} \rest{G}{s \cap \adj{a}}}{\neg (E_1 \vee E_2)} = \bigO{\frac{\varepsilon}{r}}.\]

In conclusion
\begin{align*}
    \Prob[s \in S, a \in \dl{s}{A}]{\rest{f_s}{s \cap \adj{a}} \ane{r \gamma} \rest{G}{s \cap \adj{a}}} \leq  & \\
     \prob{E_1 \vee E_2} +  \cProb{s \in S, a \in \dl{s}{A}}{\rest{f_s}{s \cap \adj{a}} \ane{r \gamma} \rest{G}{s \cap \adj{a}}}{\neg (E_1 \vee E_2)} & = \\
     \bigO{\left (1+\frac{1}{r} \right ) \varepsilon} & .
\end{align*}
\end{proof}

In a special case, we can say something even stronger.
\begin{theorem}[Extension of \pref{thm:main-STAV-agreement-theorem}]
\label{thm:stronger-equation-in-main-theorem}
Let $X, \FF$ be as in \pref{thm:main-STAV-agreement-theorem}. Suppose that we have a distribution $(v,b,a,s)$ of sets $b$ where $v \in b \subset \dl{s}{V} \cap \adj{a}$. Suppose that the marginal $(v,a,s)$ is the marginal of $V \times A \times S$ in $D_{stav}$, then the following holds:
\begin{equation} \label{eq:stronger-equation-in-main-theorem}
    \Prob[s \in S, a \in \dl{s}{A}, b \subset \dl{s}{V} \setminus \dl{a}{V}]{\rest{f_s}{b} \ane{r \gamma} \rest{G}{b}} = \bigO{\left (1+\frac{1}{r} \right )\varepsilon}.
\end{equation}
\end{theorem}
We will need \pref{thm:stronger-equation-in-main-theorem} for some of our applications.

\begin{proof}[Proof of \pref{thm:stronger-equation-in-main-theorem}]
The case discussed in \eqref{eq:stronger-equation-in-main-theorem} has a similar proof to \pref{thm:main-STAV-agreement-theorem}. We define the random variable $Z'$ that samples $(s,a,b)$ and outputs $\Prob[v \in b]{f_s(v) \ne G(v)}$. Consider the same events as in the proof of \pref{thm:main-STAV-agreement-theorem}. Since
\begin{enumerate}
    \item $E_1 \vee E_2$ occur with the same probability.
    \item The expectation of $Z'$ given that $\neg(E_1 \vee E_2)$ is still $\bigO{\gamma \varepsilon}$.
\end{enumerate}
Then by Markov's inequality for any $r > 0$, \eqref{eq:stronger-equation-in-main-theorem} holds.
\end{proof}
%%%%%%%%%%%%%%%%%%%%%%%%%%%%%%%%%%%%%%%%%%%%%%%%%%%%%%%%%%%%%%%%%%%%

\subsection{Proof of the Lemmata}
\restatelemma{lem:h-lemma}

\begin{proof}[Proof of \pref{lem:h-lemma}]
Fix $a \in A$, and denote by $\varepsilon_a$ the probability to succeed in the STS-test given that $s_1,s_2,t \supset a$. If $\varepsilon_a \geq \frac{1}{6}$ we are trivially done, so assume otherwise. Denote by $\set{h^i_a}_i$ all possible functions from $\dl{a}{V}$ to $\Sigma$, where $h^1_a = h_a$. Consider the $STS_a$-graph. According to \pref{ass:edge-expander}, this is a $\frac{1}{3}$-edge expander.

Denote by $S_i = \set{s :\rest{f_s}{\dl{a}{V}} = h^i_a}$. We need to show that the set of vertices $S_1 = \set{s : \rest{f_s}{\dl{a}{V}} = h_a}$ (the largest of all $S_i$) is $1 - \bigO{\varepsilon_a}$.

The quantity $\varepsilon_a$, i.e. the amount of edges between $S_i$'s, is by assumption less than $\frac{1}{6}$. The $STS_a$-graph is a $\frac{1}{3}$-edge expander.

It is a known fact that if we partition a vertex of an edge-expander graph, and there are few outgoing edges, then one of the parts in the partition is large. This fact is formulated in \pref{claim:edge-expander-partition-property}.

We invoke \pref{claim:edge-expander-partition-property}, using the fact that the graph is a $\frac{1}{3}$-edge expander and the fact that the fraction of edges between the $S_i$'s is less that $\frac{1}{6}$. We get that $\prob{S_1} \geq \frac{1}{2}$.

By the edge-expander property $\prob{S_1^c} \leq 3E(S_1,S_1^c) \leq 3\varepsilon_a$.

\end{proof}

\begin{corollary}
\label{cor:a-star-is-small}
$\prob{A^*} = \bigO{\varepsilon}$.
\end{corollary}

\begin{proof}[Proof of \pref{cor:a-star-is-small}]
$a \in A$ contributes to $A^*$ if the amount of bad edges that $a$ participates in is $\geq \frac{1}{40}$. The total amount of bad edges is $\bigO{\varepsilon}$ by \pref{lem:h-lemma}. Thus by Markov's inequality $\prob{A} = \bigO{\varepsilon}$.
\end{proof}

\bigskip

We move to \pref{claim:not-global-bad-implies-not-local-bad}.
\restateclaim{claim:not-global-bad-implies-not-local-bad}
\begin{proof}[Proof of \pref{claim:not-global-bad-implies-not-local-bad}]

Fix some $a \notin A^*$. Consider the $\VAS{a}$-graph for this $a$.
This is the bipartite graph, where
\[ L = \adj{a_0}, \; R = \sett{(a,s)}{\exists v \in L \; (v,a_0,s,a) \in Supp(D) },\]
\[ E = \set{(v,(a,s)) : (v,a_0,s,a) \in Supp(D) }. \]
The probability of choosing an edge $(v,(a',s))$ is given by $\cProb{D}{(v,a_0,s,a')}{a_0 = a}$.

Denote by $B \subset L$ the that consists of all $(s,a')$ s.t. $(a,s,a')$ is bad. If $a \notin A^*$ then $\prob{B} < \myconst$. From \pref{ass:vasa-graphs-a}, this graph is a $\sqrt{\gamma}$-bipartite expander.
Define the set
\[V^* = \sett{v \in \adj{a}}{\cProb{(s,a')}{B}{v \sim (s,a')} \geq \twomyconst}, \]
the set of $v \in \adj{a}$ so that the probability for a bad edge is larger than $\twomyconst$, namely, that $a$ is locally bad for $v$. In the sampler lemma, \pref{lem:sampler-lemma}, we see that bipartite-expanders are good samplers.

We use \pref{lem:sampler-lemma} to get that $\prob{V^*} = \bigO{\gamma}\prob{B}$. Taking expectation on all $a \in A$ we get that
\[ \Prob[a \in A, v \in \adj{a}]{a \in A^*_v \ve a \notin A^*} = \prob{a \notin A^*} \cdot \Ex[a \notin A^*]{\prob{V^*} } =\]
\[ \prob{a \notin A^*} \cdot \bigO{\gamma}\Ex[a \notin A^*]{\prob{B}} \leq \bigO{\gamma} \Ex[a \in A]{\prob{B}}= \bigO{\gamma \varepsilon},\]
The last inequality is due to the fact that taking expectation on this set conditioned on $a \notin A^*$, is less than the expectation on all $A$ (by definition when $a \in A^*$, then $\prob{B} \geq \frac{1}{40}$, and when $a \notin A^*$, $\prob{B} < \frac{1}{40}$).
The last equality is since $\prob{B} = \bigO{\varepsilon}$ by \pref{lem:h-lemma}.
\end{proof}

\bigskip

Moving on to \pref{lem:g-lemma}:

\begin{lemma*}[Restatement of \pref{lem:g-lemma}] 
    Let $D$ be a distribution over $(a,s,v) \in A \times S \times V$ defined by the STAV-structure, that is:
    \begin{enumerate}
        \item Choose some $(a,v)$ where $v \in \adj{a}$.
        \item Choose some $(a,v) \subset s$ (where we mean $\set{v},a \subset s$).
    \end{enumerate}
    Then
    \begin{equation*} %\label{eq:g-lemma-to-bound}
    \Prob[(a,v,s) \sim D]{f_s(v) \ne g_a(v) \ve \rest{f_s}{\dl{a}{V}} = h_a \ve a \notin A^*_v } = \bigO{\gamma \epsilon}.
    \end{equation*}
\end{lemma*}%This is manually fixed. Need to fix macro

For the proof of the lemma, we'll need the following property of expander graphs. In an expander graph, the number of outgoing edges from some $A \subset V$, is an approximation to the size of $A$ or $V \setminus A$. The following claim generalizes this fact to the setting where we count only outgoing edges from $A$ to a (large) set $B \subset V \setminus A$.

\begin{claim*}[Restatement of \pref{claim:almost-cut}]
    Let $G = (V,E)$ be a $\lambda$-two sided spectral expander. Let $V = A \dunion B \dunion C$, s.t. $\prob{A} \leq \prob{B}$. Then
    \begin{equation} %\label{eq:eml-bound}
            \prob{A} \leq \frac{1}{(1-\lambda)\prob{B}} \left ( \prob{E(A,B)} + \lambda \prob{C} \right ).
    \end{equation}

    In particular, if  \( \prob{A}, 1-\lambda = \Omega(1)\) then
    \[ \prob{A} = \bigO{  \prob{E(A,B)} + \lambda \prob{C} }. \]
\end{claim*}%This is manually fixed. Need to fix macro

\begin{proof}[Proof of \pref{lem:g-lemma}]

For a fixed $(a_0,v_0)$ we consider he conditioned $STS_{a_0,v_0}$-graph. Recall that the vertices in this graph are all the $s \supset (a,v)$ and the edges are $(s,t,s')$ so that $t \supset (a,v)$.

We partition this graph to three sets:
    \[ M_{a_0,v_0} = \sett{s \in V}{\rest{f_s}{\dl{a_0}{V}} = h_{a_0}, f_s(v_0) = g_a(v_0)}, \]
    \[ N_{a_0,v_0} = \sett{s \in V}{\rest{f_s}{\dl{a_0}{V}} = h_{a_0}, f_s(v_0) \ne g_a(v_0)}, \]
    \[ C_{a_0,v_0} = \sett{s \in V}{\rest{f_s}{\dl{a_0}{V}} \ne h_{a_0}}. \]

For $(a_0,v_0)$ so that $a_0 \notin A^*_{v_0}$, the elements $s \in N_{a_0,v_0}$ are exactly those who contribute to the probability in \eqref{eq:g-lemma-to-bound}.
Thus the probability in \eqref{eq:g-lemma-to-bound} is
\[\Prob[(a_0,v_0)]{a_0 \notin A^*_{v_0}}\Ex[(a_0,v_0): \; a_0 \notin A^*_{v_0}]{\prob{N_{a_0,v_0}}}.\]

We also denote by $H_{a_0,v_0}$ the set of edges $(s_1,t,s_2)$ in the $STS_{a_0,v_0}$-graph, so that \[f_{s_1}(v_0) \ne f_{s_2}(v_0) \ve \rest{f_{s_1}}{\dl{a_0}{V}}  = \rest{f_{s_2}}{\dl{a_0}{V}}. \]
Note that any edge between $N_{a_0,v_0}$ and $M_{a_0,v_0}$ is an edge of $H_{a_0,v_0}$.
%\[ H_{a_0,v_0} = \sett{(s_1,t,s_2)}{(a_0,v_0) \subset t \subset s_1,s_2, \; f_{s_1}(v_0) \ne f_{s_2}(v_0) \ve \rest{f_{s_1}}{\dl{a_0}{V}}  = \rest{f_{s_2}}{\dl{a_0}{V}}}, \]
By \eqref{eq:surprise}, $\surp(\FF) = \gamma$. Thus in particular
\[ \Prob[(s_1,s_2,t,a,v)]{f_{s_1}(v) \ne f_{s_2}(v) \ve \rest{f_{s_1}}{a} = \rest{f_{s_2}}{a} }  \leq \]
\[ \Prob[(s_1,s_2,t,a,v)]{\rest{f_{s_1}}{t} \ne \rest{f_{s_2}}{t} \ve \rest{f_{s_1}}{a} = \rest{f_{s_2}}{a} } = \prob{\rest{f_{s_1}}{t} \ne \rest{f_{s_2}}{t}}\surp(\FF) = \gamma \varepsilon.
\]

Thus
\[ \Ex[(a_0,v_0)]{\prob{H_{a_0,v_0}}} = \prob{f_{s_1}(v) \ne f_{s_2}(v) \ve \rest{f_{s_1}}{\dl{a}{v}} = \rest{f_{s_2}}{\dl{a}{v}}} = \bigO{\gamma \varepsilon}.\]

According to \pref{ass:v-a-graph}, the $STS_{a_0,v_0}$-graph is a $\gamma$-two-sided spectral expander, thus we can use the almost cut approximation property \pref{claim:almost-cut} to show that
\begin{equation}\label{eq:g-lemma-almost-cut}
    (1-\gamma)\prob{M_{a_0,v_0}}\prob{N_{a_0,v_0}} = \bigO{\prob{E(N_{a_0,v_0},M_{a_0,v_0})} + \gamma \prob{C_{a_0,v_0}}}.
\end{equation}

To conclude the proof we show first that the right hand side of \eqref{eq:g-lemma-almost-cut} is bound by $\bigO{\gamma \varepsilon}$ in expectation over $(a_0,v_0)$. Then we show that for all but $\bigO{\gamma \varepsilon}$ of the $(a_0,v_0)$,
\begin{equation}         \label{eq:b-larger-than-half-1}
    \prob{M_{a_0,v_0}}  \geq \frac{1}{2}.
\end{equation}

Indeed, as
\[E(N_{a_0,v_0},M_{a_0,v_0}) \subset H_{a_0,v_0},\]
we get that
    \[\Prob[(a_0,v_0)]{a_0 \notin A^*_{v_0}}\Ex[(a_0,v_0): \; a_0 \notin A^*_{v_0}]{\prob{E(N_{a_0,v_0},M_{a_0,v_0})}} \]
    \[ \leq \Ex[(a_0,v_0)]{\prob{H_{a_0,v_0}}} = \bigO{\gamma \varepsilon}.\]
Furthermore, Note that
    \[\Prob[(a_0,v_0)]{a_0 \notin A^*_{v_0}}\Ex[(a_0,v_0): \; a_0 \notin A^*_{v_0}]{\gamma \prob{C_{a_0,v_0}}}  \leq \gamma \Ex[(a_0,v_0)]{\prob{C_{a_0,v_0}}} = \]
    \[ \gamma \prob{\rest{f_s}{\dl{a}{v}} \ne h_a}. \]
This is bounded by $\bigO{\gamma \varepsilon}$ by \pref{lem:h-lemma}.

Hence the right hand side of \eqref{eq:g-lemma-almost-cut} is bound by $\bigO{\gamma \varepsilon}$ in expectation over $(a_0,v_0)$.

\medskip

To complete the proof, we turn to showing \eqref{eq:b-larger-than-half-1} for all but $\bigO{\gamma \varepsilon}$ of the $(a_0,v_0)$. For this, we use the edge expander partition property, \pref{claim:edge-expander-partition-property}.

Partition the vertices of the $STS_{a_0,v_0}$-graph to $B_1,B_2,...B_{n+1}$ where $B_1 = M_{a_0,v_0}, B_2 = C_{a_0,v_0}$ and $N_{a_0,v_0} = B_3 \dunion ... \dunion B_n$, where each $B_j$ is the set of $s$ so that $f_s(v) = \sigma_j$  for all $\sigma_j \in \Sigma$.

We assumed that $\prob{C_{v_0}} \leq \twomyconst$ hence $E(B_2,B_2^c) = E(C,C^c) \leq \fourmyconst$.

From \eqref{eq:surprise}, $\Ex[(a_0,v_0]{\prob{H_{a_0,v_0}}} = \bigO{\gamma\varepsilon}$. From Markov's inequality, $\prob{H_{a_0,v_0}} < \twomyconst$, for all but $\bigO{\gamma \varepsilon}$ of the $(a_0,v_0)$. When this occurs, the amount of edges between the partition parts is $\frac{3}{20}< \frac{1}{6}$.

From the edge expander partition property \pref{claim:edge-expander-partition-property} we get that one of the partition sets has probability $\geq \frac{1}{2}$. This is not $B_2 = C$, as its probability is $\leq \twomyconst$. Thus $\prob{M_{a_0,v_0}} \geq \frac{1}{2}$.

\bigskip

Thus by using the fact that $\gamma < \frac{1}{3}$, for all but $\bigO{\gamma \varepsilon}$ of the $(a_0,v_0)$,
\[ \prob{N_{a_0,v_0}} = \bigO{\prob{E(N_{a_0,v_0},M_{a_0,v_0})} + \gamma \prob{C_{a_0,v_0}}}. \]

Hence
\[\Ex[(a_0,v_0)]{\prob{N_{a_0,v_0}}} = \bigO{\gamma\varepsilon}.\]
\end{proof}

\begin{corollary} \label{cor:key-cor}
Consider the $VASA$-distribution promised for us in \pref{ass:vasa-graphs}.

\[ \Prob[(v,a_1,s,a_2)]{\rest{f_s}{a_i} = h_{a_i} \wedge g_{a_1}(v) \ne g_{a_2}(v) \ve a_i \notin A^*_v \text{ for }i=1,2} = \bigO{\gamma \varepsilon}.\]
\end{corollary}

\begin{proof}[Proof of \pref{cor:key-cor}]
The probability is bounded by twice the probability we bound in \pref{lem:g-lemma}, and the probability we bound in \pref{claim:not-global-bad-implies-not-local-bad}.
\end{proof}

\subsubsection{Proof of \pref{lem:global-agreement-links}}
We restate \pref{lem:global-agreement-links}:
\restatelemma{lem:global-agreement-links}

%For the proof of \pref{lem:g-lemma} we need the following claim:

\begin{proof} [Proof of \pref{lem:global-agreement-links}]
Fix some $v_0 \in V$ and consider its $\vASA{v}$-graph defined in \pref{sec:stav}, namely the graph whose vertices $\adj{v_0}$ and edges are all the $(a_1,s,a_2)$ so that $(v_0,a_1,s,a_2)$ is in the support of our $\VASA$-distribution.

Consider the following sets in this graph:
\[ M_{v_0} = \sett{a \in \adj{v_0} \setminus A^*_{v_0}}{g_a(v_0) =   G(v)}, \text{ the popular vote, }\]
\[ N_{v_0} = \sett{a \in \adj{v_0} \setminus A^*_{v_0}}{g_a(v_0) \ne G(v)}, \text{ the other votes, }\]
\[ C_{v_0} = A^*_{v_0}, \]
The $a \in N$ are those where $g_a(v_0) \ne G(v_0)$ and $a \notin A^*_{v_0}$. Hence we need to bound
\[ \Ex[v_0]{\prob{N_{v_0}}}.\]

By \pref{ass:vasa-graphs-v} it is a either a $\gamma$-bipartite expander or a $\gamma$-two-sided spectral expander. \pref{claim:almost-cut-bipartite}, the bipartite almost cut approximation property, is the analogue claim to \pref{claim:almost-cut} for bipartite graphs. We invoke either \pref{claim:almost-cut-bipartite} or \pref{claim:almost-cut} for $N_{v_0},M_{v_0},C_{v_0}$ and get that
\[ (1-2\gamma)\prob{M_{v_0}} \prob{N_{v_0}} \leq \prob{E(N_{v_0},M_{v_0})} + 4\gamma \prob{C_{v_0}}, \]
or
\begin{equation} \label{eq:minority-bound}
\prob{N_{v_0}} \leq \frac{1}{(1-2\gamma)\prob{M_{v_0}}} \prob{E(N_{v_0},M_{v_0})} + 4\gamma \prob{C_{v_0}}.
\end{equation}

The proof now has two steps:
\begin{enumerate}
    \item We show that $\prob{M_{v_0}} \geq \frac{1}{2}$ for all but $\bigO{\gamma \varepsilon}$ of the vertices $v_0$.
    \item We show that the right hand side of \eqref{eq:minority-bound} is $\bigO{\gamma \varepsilon}$.
\end{enumerate}

\paragraph{To show step 1} we will need to show that for all but $\bigO{\gamma \varepsilon}$ of the $v_0$, the size of $C_{v_0}$ is smaller than $\frac{1}{20}$.

\begin{equation}\label{eq:locally-bad-set-is-small}
   \Prob[v]{\prob{A^*_v} > \twomyconst } = \bigO{\gamma \varepsilon}
\end{equation}

Assuming that for $\prob{C_{v_0}} \leq \frac{1}{20}$, we show that $\prob{M_{v_0}} \geq \frac{1}{2}$ using the edge expander partition property, \pref{claim:edge-expander-partition-property}.

By \pref{ass:vasa-graphs-v}, the $\vASA{v_0}$-graph is a either $\gamma$-bipartite expander or a $\gamma$-two-sided spectral expander for $\gamma < \frac{1}{3}$, thus it is also a $\frac{1}{3}$-edge expander. We indend to invoke \pref{claim:edge-expander-partition-property}. Partition $V$ to:
\begin{itemize}
    \item $B_0 = A^*_{v_0} = C_{v_0}$.
    \item $B_1 = M_{v_0}$.
    \item $B_2,...,B_n$ - elements $a \in A$ s.t. $g_a(v) = \sigma_i$ for all $\sigma_i \in \Sigma$ that are not the majority assumption. Note that $A_{v_0} = B_2 \dunion ... \dunion B_n$.
\end{itemize}
By \pref{eq:locally-bad-set-is-small}, the set $B_0 = A^*_{v_0}$ is $\leq \twomyconst$ for all but $\bigO{\gamma \varepsilon}$ of the $v$'s. When this occurs, then $E(C,C^c) \leq \frac{1}{10}$.

We bound the amount of edges between the $B_i$'s that are not $B_0$. We can divide the edges to bad edges, and edges that are not bad.

The ``bad edges'' between the $B_i$'s account for at most $\twomyconst$ as for every $i$ and every $a \in B_i$, the amount of bad edges connected to it is $\leq \frac{1}{20}$ (since $a \notin A^*_{v_0}$).

Finally, from \pref{cor:key-cor} and Markov's inequality, there are at most $\bigO{\gamma \varepsilon}$ of the $v$'s where the amount of edges between different $B_i$'s that are not bad is greater than $\twomyconst$.

Thus for all but $\bigO{\gamma \varepsilon}$ of the $v$'s, the amount of edges between parts of this partition is $\leq \frac{2}{20} < \frac{1}{6}$. Invoke \pref{claim:edge-expander-partition-property}, to get that one set above must be of size at least $\frac{1}{2}$. This must be $B_1 = M_{v_0}$, as it is larger than the other $B_i$'s where $i \geq 1$, and since $B_0 = C_{v_0}$ is of size $\leq \twomyconst$.

\bigskip

We move to show that \eqref{eq:locally-bad-set-is-small} is true for all but $\bigO{\gamma \varepsilon}$ of the vertices $v_0 \in V$. Consider the graph between STAV-parts $V$ and $A$ where we choose a pair $(a,v)$ according to the probability to chose them in the $STAV$-structure.

The set of vertices $v$ that we need to bound is the set of $v$'s with large $\prob{A_v^*}>\twomyconst$. 
There are two types of vertices $v$:
\begin{itemize}
    \item $\prob{A_v^* \cap A^*}\leq \myconst$
    \item $\prob{A_v^* \cap A^*}> \myconst$
    \end{itemize}
By \pref{claim:not-global-bad-implies-not-local-bad}, $\Prob[(a,v)]{a \in A^*_v \ve a \notin A^*} = \bigO{\gamma\varepsilon}$. Thus by Markov's inequality, only   $\bigO{\gamma\varepsilon}$ of the vertices can see $\myconst$-fraction of neighbors $a \in A^*_v\setminus A^*$, thus bounding by $\bigO{\gamma\varepsilon}$ the fraction of $v$'s of the first type.

To bound the vertices of the second type, note that these are vertices that have a large $(>\myconst)$ fraction of neighbors in $A^*$. By \pref{cor:a-star-is-small}, $\prob{A^*} = \bigO{\varepsilon}$. 
According to \pref{ass:global-A-V-graph}, our graph is a $\sqrt{\gamma}$-bipartite expander. Thus by the sampler lemma \pref{lem:sampler-lemma}, the set of vertices $v_0 \in X(0)$ who have more than $\myconst$-fraction neighbours in $A^*$, is $\bigO{\gamma\varepsilon}$.

\paragraph{As for step 2} Taking expectation on \eqref{eq:minority-bound} we get that

\[ \ex{\prob{N_{v_0}}} \leq \ex{\frac{1}{(1-2\gamma)\prob{M_{v_0}}} \prob{E(N_{v_0},M_{v_0})}} + 4\gamma \ex{\prob{C_{v_0}}}\]

\begin{equation} \label{eq:expectation-minority-bound}
\leq \Prob[v]{\prob{A^*_v} > \twomyconst } + \ex{6\prob{E(N_{v_0},M_{v_0})}}  + 4\gamma \ex{\prob{C_{v_0}}},
\end{equation}

where the second inequality is due to the fact that we assumed that \(\gamma < \frac{1}{3}\) and that when  \(\Prob[v]{\prob{A^*_v} \leq \twomyconst }\) then \(\prob{M_{v_0}} \geq \frac{1}{2}\), hence 
\[\frac{1}{(1-2\gamma)\prob{M_{v_0}}} \leq 6.\]

We bound each of the terms on the right hand side of \eqref{eq:expectation-minority-bound} separately.

By \eqref{eq:locally-bad-set-is-small},
\[\Prob[v]{\prob{A^*_v} > \twomyconst } = \bigO{\gamma \varepsilon}.\]

By \pref{cor:a-star-is-small} and \pref{claim:not-global-bad-implies-not-local-bad} \[4 \gamma\Ex[v_0]{\prob{C_{v_0}}} = 4\gamma \Ex[v]{\prob{A^*_v}} = \bigO{\gamma \varepsilon}.\]

We continue to bound $\prob{E(N_{v_0},M_{v_0})}$ in expectation. Every edge counted in $E(N_{v_0},M_{v_0})$ is either a bad triple (i.e. and edge $(a_1,s,a_2)$ s.t. $\rest{f_s}{a_i} \ne h_{a_i}$ for $i=1$ or $2$), or a non-bad edge (an edge who is not bad) for which we see a disagreement. By \pref{cor:key-cor-two-sided} there are $\bigO{\frac{\varepsilon}{d}}$ non-bad edges in the cut.

As for the bad edges, notice that $a \in N_{v_0}$ is not a member of $A^*_{v_0}$, thus the amount of bad edges that are connected to $a$ is at most $\twomyconst$-fraction of the edges connected to $a$ (by definition). Thus the amount of bad edges is bounded by $\fourmyconst\prob{N_{v_0}}$, and
\[\prob{E(N_{v_0},M_{v_0})} \leq \bigO{\gamma \varepsilon} + \fourmyconst\prob{N_{v_0}}.\]

By summing up the bounds we get that
\[ \ex{\prob{N_{v_0}}} \leq \bigO{\gamma \varepsilon} + \frac{6}{20}\ex{\prob{N_{v_0}}}\]
hence
\[ \ex{\prob{N_{v_0}}} = \bigO{\gamma \varepsilon}.\]

\end{proof}

\section{Agreement on High Dimensional Expanders}
\label{sec:agreement-on-hdx}
In the next three sections we derive several different agreement theorems from our STAV agreement theorem, \pref{thm:main-STAV-agreement-theorem}.

The first two agreement testing theorems, \pref{thm:main-agreement-theorem-two-sided} and \pref{thm:main-agreement-theorem-one-sided}, improve and extend the agreement tests from \cite{DinurK2017} and from \cite{DinurSteurer2014}. In both theorems the ground set are the vertices of a simplicial complex and the subsets are the faces. In the first theorem the complex is a two-sided high dimensional expander, and in the second theorem it is a one-sided high dimensional expander with a $d+1$-partite structure. These types of objects are generalizations of expander graphs, defined formally in \pref{sec:HDX}.
%
%We recall that a simplicial complex $X$ is a family of subsets that is downwards closed to containment, i.e. if $s \in X$ and $t \subset X$ the then $t \in X$. We denote by $X(\ell)$ all subsets (also called faces) of size $\ell+1$. We identify $X(0)$ with the set of vertices. A complex is $d$-dimensional if the largest faces have size $d+1$.
%
%We say $X$ is \emph{pure $d$-dimensional} if every face $t \in X$ is contained in a face $s \in X$ of size $d+1$. The elements in $X$ are called \emph{faces}.
%
%If $H = (V,E)$ is a $d+1$-uniform hypergraph, we can turn it to pure $d$-dimensional simplicial complex by adding all the subsets of the $s \in E$.
%

Our first application is for the family $S=X(d)$ whose ground set is $X(0)$. Our test is the \(d,\ell\)-agreement test:
\restatedefinition{def:d-l-agreement-test}
%\begin{definition}[$d,\ell$-agreement distribution, Restated] \torestate{\label{def:d-l-agreement-test}
%Let $X$ be a $d$-dimensional simplicial complex and $\ell < d$ be a positive integer. We define the distribution $D_{d,\ell}$ by the following random process
%\begin{enumerate}
%    \item Sample $t  \in X(\ell)$.
%    \item Sample $s_1,s_2 \in X(d)$ independently, given that $t \subset s_1,s_2$.
%\end{enumerate}}
%\end{definition}
%
The $d,\ell$-agreement test is the test associated with the $d,\ell$-agreement distribution on this family. We show that the ${d,\ell}$-agreement test is sound, as long as $X$ is a \emph{two-sided high dimensional expanders (HDX)}.

\begin{theorem}[Agreement for High Dimensional Expanders]
\torestate{\label{thm:main-agreement-theorem-two-sided}
There exists a constant $c>0$ such that for every two natural numbers $d>\ell$ such that $\frac{1}{2}d - \ell = \Omega(d)$ the following holds. Suppose that $X$ is a $\frac{1}{d^2 \ell}$-two-sided $d$-dimensional HDX.
Then for every $r > 0$ the $d,\ell$-agreement test is $\frac{r}{\ell}$-approximately $\left(c(1+ \frac{1}{r} ) \right )$-sound.
In particular, if $\ell = \Omega(d)$, then the test is {\em exactly} $c$-sound.
\rem{Let $\FF$ be family of \emph{local functions}:
\[ \FF = \sett{f_s : s \to \Sigma }{s \in X(d)}.\]
Assume that the agreement probability of the (d,$\ell$)-agreement test is $Agree(\FF) = 1 - \varepsilon$.
Then there exists a global function $G: X(0) \to \Sigma$, so that for every $r > 0$
\[ \Prob[s \in X(d)]{f_s \ane{\frac{r}{\ell}} \rest{G}{s}} = \bigO{\left (1+\frac{1}{r} \right ) \varepsilon}.\]
In particular, if $\ell = \Omega(d)$, then
\[ \Prob[s \in X(d)]{f_s \ne \rest{G}{s}} = O(\varepsilon).\]
Namely, the probability that $f_s$ is equal to $\rest{G}{s}$ \emph{on all vertices of $s$} is $1-O(\varepsilon)$.
}
}
\end{theorem}

The theorem in \cite{DinurSteurer2014} says that the $d$-dimensional {\em complete complex} supports a $c$-sound agreement test for some constant $c > 0$. The complete complex is the complex that contains all possible sets of size $\leq d+1$. This is a special case of \pref{thm:main-agreement-theorem-two-sided}, but even this case is not trivial.

Building on \cite{DinurSteurer2014}, the main theorem in \cite{DinurK2017} shows that the $\sqrt d$-dimensional skeleton $S=X(\sqrt{d})$, of a $d$-dimensional two-sided high-dimensional expander, supports a $c$-sound agreement test for some constant $c > 0$. This gave the first agreement test on a sparse system of sets, that is, such that every vertex is contained in a constant number of sets. Why go to a $\sqrt d$ dimensional skeleton? This was due to a technical step in the proof, and we show in 
\pref{thm:main-agreement-theorem-two-sided} that it is unnecessary. In fact \emph{all} levels of a two-sided high dimensional expander, give rise to a sound agreement test.

A subtle and not very important difference between our theorem and the theorem in \cite{DinurK2017} is the agreement distribution. The two distributions are slightly different (one is based on an upper walk and one is based a lower walk), but the difference is unimportant because one you've proven the result with one of these distributions, it implies the same for the other. For a further discussion on this matter, see \pref{sec:independent--to-expanding-distributions}.

This theorem has some implications for matroids. Let $X$ be a simplicial complex whose faces are the independent sets of a fixed matroid whose rank is $r$ (i.e. the largest independent set in this matroid has size $r$). In an exciting recent work \cite{ALGV18} it was proven that this complex is a $0$-one-sided HDX. Oppenheim \cite{oppenheim3} proved that if we truncate this complex by keeping only faces of dimensions $0\le i \le d$ then it becomes a $1/(r-d-2)$-two-sided HDX. We reach the following conclusion
\begin{corollary}[Truncated matroids]
  For any matroid of rank $r$, for any $d \le \sqrt[3] r$, the collection of independent sets in a matroid whose size is $d$ supports a sound agreement test.\qed
\end{corollary}
Furthermore, some matroids are themselves (without truncation) two-sided high dimensional expanders. For example the matroid of linear bases of a vector space $\F_q^n$ can easily be shown to be a $\frac 1 q$-two-sided HDX. When $n \leq q$ we can deduce that
\begin{corollary}[Linear bases matroid]
  Let $S$ be the collection of all linear bases of a vector space $\F_q^n$. If $n\leq q$ then this set system supports a sound agreement test. \qed
\end{corollary}

%\paragraph{Partite Simplicial Complexes}
If simplicial complexes are high dimensional analogues to graphs, then $d+1$-partite simplicial complexes are analogues to bipartite graphs: in these complexes we can partition the vertex set $V$ to $V_1,...,V_{d+1}$, so that every set of size $d$ contains exactly one vertex from each set $V_i$.

Our second theorem shows that the $d,\ell$-agreement test is sound when $X$ is a $d+1$-partite complex that is an \emph{one-sided high dimensional expander (HDX)}. One sided HDX are the high dimensional analogue to bipartite expanders. They are formally defined \pref{sec:agreement-in-hdx-one-sided}.

\begin{theorem}[Agreement for $(d+1)$-Partite High Dimensional Expanders]
\torestate{\label{thm:main-agreement-theorem-one-sided}
There exists a constant $c > 0$ such that for every two natural numbers $k, \ell$ so that $k \geq 4\ell+4$ the following holds. Suppose $X$ is a $k$-dimensional skeleton of a $(d+1)$-Partite $\frac{1}{k^2\ell}$-one sided HDX (including $k = d$)\footnote{a $k$-skeleton of a $d$-dimensional simplicial complex $Y$ is $X=\sett{s \in X}{|s| \leq k+1}$.}. Then for every $r > 0$ the $d,\ell$-agreement test is $\frac{r}{\ell}$-approximately $\left ( c \left ( 1 + \frac{1}{r} \right ) \right )$-sound. In particular, if $\ell = \Omega(k)$, then the test is exactly $c$-sound.

}
\end{theorem}
Interestingly, in the known one-sided $d+1$-partite simplicial complexes, the distribution on $X(d)$ is uniform. Thus this theorem gives us a sparse \emph{uniformly distributed} set system with a sound agreement test. This is unlike the known constructions for two-sided high-dimensional expanders that come from truncating one-sided high-dimensional expanders and for which the distribution of the test over $S=X(d)$ is not uniform.

\paragraph{Organization of this section} This section is a bit long so let us quickly explain its contents. In \pref{sec:rw-in-sc} we describe random walks on simplicial complexes, both the well-known ``containment'' random walks as well as the new complement random walks. In \pref{sec:agr2sided} we prove \pref{thm:main-agreement-theorem-two-sided} by showing that any two sided HDX supports a STAV structure. In \pref{sec:agreement-in-hdx-one-sided} we prove \pref{thm:main-agreement-theorem-one-sided}. The proof of this theorem is more intricate, as we don't only find one STAV structure but rather many different STAVs. We apply our main technical theorem on each and then combine the outcomes together.

\subsection{Random Walks on Simplicial complexes}\label{sec:rw-in-sc}
We refer to the definition of a weighted simplicial complex and High Dimensional Expanders in \pref{sec:HDX}.

\paragraph{The Containment Walk} On a $d$-dimensional simplicial complex we can define the $k,\ell$-lower random walk, for $\ell< k \leq d$:

\begin{definition}[The lower walk]\torestate{\label{def:lower-walk}
Given $s \in X(k)$ we choose $s' \in X(k)$ by:
\begin{itemize}
\item Choose $t \in X(\ell)$ given that $t \subset s$.
\item Choose $s' \in X(k)$ given that $t \subset s'$.
\end{itemize}}
\end{definition}

One can also define the \emph{$\ell,k$-upper walk} on $X(\ell)$, where we given $t \in X(\ell)$ we choose $s \supset t$ in $X(k)$, and then choose $t' \subset s$.

This random walk is in fact two independent steps in the $k,\ell$-containment graph:
\[ L = X(k), \; R = X(\ell), \; E = \sett{(s,t)}{t \subset s}. \]
We denote the bipartite operator of this graph by $D_{k,\ell}$. Note that
\[ \Down_{k,\ell} = \Down_{\ell+1,\ell} \Down_{\ell+2,\ell+1} ... \Down_{k,k-1}. \]

This random walk has been studied in ~\cite{KaufmanM17, KaufmanO2017,DinurK2017,DiksteinDFH2018} and more. In particular \cite{KaufmanO2017} proved the following theorem:

\begin{theorem}
\label{thm:containment-graph-is-a-good-expander}
Let $X$ be a $\lambda$-one sided link expander, then  $\lambda(\Down_{k+1,k})$, the second largest eigenvalue of the upper walk, is $\sqrt{\frac{k+1}{k+2}} + \bigO{k\lambda}$.

\end{theorem}

\pref{thm:containment-graph-is-a-good-expander} immediately implies the following useful corollary:
\begin{corollary}
\label{cor:containment-generalized-graph-is-a-good-expander}
Let $X$ be a $\lambda$-one sided link expander, then  $\lambda(\Down_{k,\ell})$ is $\sqrt{\frac{\ell+1}{k+1}} + O_{k+t}(\lambda)$.
$\qed$
\end{corollary}

\paragraph{The Complement Walk}
As we noted in the introduction, we needed a random walk for a $\VASA$-distribution on two-sided high dimensional expanders. The spectral gap of this walk needed to be $\bigO{\frac{1}{\ell}}$. Unfortunately, the lower walk, or its dual, the upper walk, had spectral gap of approximately $\frac{\ell+1}{k+1}$. This is a constant when $\ell = \Omega (k)$.

The complement walk, is a walk between $X(k)$ and $X(\ell)$, where we go from $s \in X(k)$ to $t \in X(\ell)$ by $t \dunion s \in X(k+\ell+1)$.

\begin{definition}[The Complement Walk]
\torestate{
\label{def:complement-walk}
Let $X$ be a $d$-dimensional simplicial complex. Let $k,\ell$ be integers s.t. $k+\ell+1 \leq d$. The $k,\ell$-complement walk is the bipartite graph with edges $(L,R,E)$:
\begin{itemize}
    \item The vertices are $L = X(k), \; R = X(\ell)$.
    \item The edges are $E = \sett{ (s,t)}{s \dunion t \in X(k+\ell+1)}$.
\end{itemize}
The probability of choosing an edge $(s,t)$ is the probability of choosing $s \dunion t \in X(k+\ell+1)$ and then choosing $s \in X(k)$, given that we chose $s \dunion t$.
}
\end{definition}

We will show that in a $\lambda$-two-sided spectral expander, this walk has spectral gap proportionate to $\ell,k$ and $\lambda$. More formally, we will prove the following claim (\pref{thm:complement-walk-is-a-good-expander}, item 1):
\begin{claim}
\label{claim:two-sided-complement-walk-is-a-good-expander}
Let $X$ be a $\lambda$-two-sided link expander. $\ell_1,\ell_2$ integers so that $\ell_1 + \ell_2 +1 \leq d$. Denote by $\comp{\ell_1,\ell_2}$, the bipartite operator of the $\ell_1,\ell_2$-complement walk. Then
            \[\lambda(\comp{\ell_1,\ell_2}) \leq (\ell_1+1)(\ell_2+1)\lambda.\]
\end{claim}

\paragraph{Colored Walks in $d+1$-Partite Simplicial Complexes} On one-sided high dimensional expanders, the complement walk may not be a good expander. However, in the $d+1$-partite case we can define an analogue to this walk, the \emph{colored walk}. For two colors $I,J$, this walk goes from $t \in X[ I]$ to $s \in X[ J]$ by a face in $X[I \dunion J]$.
\begin{definition}[The Colored Walk]\torestate{\label{def:colored-walk}
Let $X$ be a $d$-dimensional $d+1$-partite simplicial complex. Let $I,J \subset [d]$ be two disjoint sets of colors. The $I,J$-colored walk is the bipartite graph with edges $(L,R,E)$:
\begin{itemize}
    \item The vertices are $L = X[ I], \; R = X[ J]$.
    \item The edges are $E = \sett{ (s,t)}{ s \dunion t \in X[ I \dunion J]}$.
\end{itemize}
The probability of choosing an edge $(s,t)$ is the probability of choosing $s \dunion t \in X[ I \dunion J]$.}
\end{definition}

Denote the bipartite adjacency operator of this walk by $\comp{I,J}$. We show that if $X$ is a $d+1$-partite $\lambda$-one sided link expander, $\lambda(\comp{I,J})$ is proportionate to $\abs{I} \abs{J}$ and  $\lambda$. We state the following claim that bounds the spectral gap of the colored walks (\pref{thm:complement-walk-is-a-good-expander}, item 2):
\begin{claim}
\label{claim:colored-walk-is-a-good-expander}
Let $X$ be a $d+1$-partite $\frac{\lambda}{(d+1)\lambda + 1}$-one-sided link expander, where $\lambda < \frac{1}{2}$. Let $I,J \subset [d]$ be two disjoint colors. Denote by $\comp{I,J}$ the $I,J$-colored walk. Then
    \[\lambda(\comp{I,J}) \leq |I| |J| \lambda.\]
\end{claim}

\subsection{Agreement for Two-Sided High Dimensional Expanders}\label{sec:agr2sided}

\begin{proof}[Proof of \pref{thm:main-agreement-theorem-two-sided}]
First, note that when $d$ is small, the theorem is true by a simple union bound. Thus we may assume $d>>1$.

To show the theorem is true, we need to take some ensemble of functions $\FF$ and show that if $\disagr(\FF) = \varepsilon$ then there exists a global function $G: X(0) \to \Sigma$ such that
\[ \Prob[s \sim D_{d,\ell}]{f_s \ane{\frac{r}{\ell}} \rest{G}{s}} = c \left ( 1 + \frac{1}{r} \right ) \varepsilon, \]
for some constant $c$.

The STAV structure we examine for this agreement test is as follows:
\[S = X(d), \; T = X(t); A = X(t-1); V = X(0). \]
Our distribution is
\begin{enumerate}
    \item Choosing $s \in S$ according to the distribution of the simplicial complex.
    \item Choosing $t \subseteq s$ uniformly at random.
    \item Choosing $(v,a)$ by choosing $v \in t$ uniformly at random and setting $a = t \setminus \set{v}$.
\end{enumerate}
The $STS$-test of this structure is the $(d,\ell)$-agreement test. The $\VASA$-distribution is the following:
    \begin{enumerate}
        \item Choose $s \in X(d)$.
        \item Choose $a_1,a_2,v$ so that $a_1 \dunion a_2 \dunion \set{v} \subset s$.
    \end{enumerate}
This distribution is obviously symmetric in $a_1$ and $a_2$. Furthermore, the choice of $(v,a_i,s)$ is identical to the marginal in the STAV-structure.

First, we claim that for any simplicial complex, the STAV-structure above has $\surp(X) = \bigO{\frac{1}{\ell}}$.
\begin{claim}\torestate{
\label{claim:agreement-on-last-vertex}
Let $X$ be any simplicial complex, then the surprise of a STAV-structure with $T = X(\ell), \; A = X(\ell-1); \; V = X(0)$ has $\surp(X) = \bigO{\frac{1}{\ell}}$.}
\end{claim}

\begin{proof}[Proof of \pref{claim:agreement-on-last-vertex}]
Let $\FF$ be any ensemble of local functions. We need to show that
\[ \cProb{(s_1,s_2,t,a)}{\rest{f_{s}}{a} = \rest{f_{s}}{a}}{\rest{f_{s}}{t} = \rest{f_{s}}{t}}.\]
To do so, we want to invoke \pref{lem:delta-surprise}. For every $t \in T$, the $T$-lower graph is the containment graph where on one side we have
\[ L = \binom{t}{\ell-1},\]
and on the other we have
\[ R = \set{v \in t}.\]
%The next paragraph is not implied by \cite{KaufmanO2018} that only promines 1/\sqrt{\ell}%
It is a well known fact that this graph is a $\frac{1}{\ell}$-bipartite expander. Trivially, if $\rest{f_{s_1}}{t} \ne \rest{f_{s_2}}{t}$, they differ on at least $\frac{1}{\ell+1}$-fraction of the vertices (namely, one vertex). By \pref{lem:delta-surprise}, we get a surprise of $\frac{\ell^{-2}}{(\ell+1)^{-1}} = \bigO{\frac{1}{\ell}}$.
\end{proof}

If we show STAV-structure defined above theorem is $\bigO{\gamma}$-good as in \pref{def:good-stav-structure}, for $\gamma = \frac{1}{\ell}$, we could invoke \pref{thm:main-STAV-agreement-theorem} and conclude. Hence we need, that it fulfils the assumptions in \pref{def:good-stav-structure}:
\begin{enumerate}

    \item We begin with the proof of \pref{ass:a-not-large}. We require that the probability of choosing some $v \in s$ so that $v \in \adj{a}$ is greater or equal to $\frac{1}{2}$ for any $a \subset s$. as $\ell < \frac{1}{2}d$, and $s \cap \adj{a} = s \setminus a$ for any $a \subset s$,
    \[ \cProb{}{v \in \adj{a}}{v \in s} = \frac{|s \setminus a|}{|s|} \geq \frac{1}{2}.\]
    %%%   {This is at the beginning on purpose. I think it is better to put it here than to put it in the end, where it hides after spectral claims.}   %%%

    \item \pref{ass:global-A-V-graph}: The graph described in this assumption is the $0,\ell-1$ complement walk graph in $X$. By our assumption $X$ is a $\frac{1}{d^2 \ell}$-two-sided HDX. Thus from \pref{claim:two-sided-complement-walk-is-a-good-expander}, this graph is a $\bigO{\frac{1}{\ell}}$-bipartite expander.

    \item \pref{ass:edge-expander}: Fix $a \in A$. The conditioned $STS$-graph, $STS_{a}$ is the graph whose vertices are all $\ul{a}{s}$. Traversing from $s_1$ to $s_2$ is going by $t = a \dunion \set{v}$. We are to show that this graph is a $\bigO{\frac{1}{\ell}}$-two-sided spectral expander. Indeed this graph is (isomorphic to) the graph obtained whose vertices are $X_a(d - \ell)$, and $s_1 \setminus a, s_2 \setminus a$ are connected by an edge if their intersection contains a vertex in $X_a(0)$. $d-\ell = \Omega(\ell)$ thus by \pref{thm:containment-graph-is-a-good-expander}, and the fact that $X$ is an $\frac{1}{\ell d^2}$-HDX, this graph is a $\eta$-two-sided spectral expander, for
    \[ \eta = \frac{1}{d- \ell} + \bigO{\frac{d}{d^2 \ell}} = \bigO{\frac{1}{\ell}}. \]
     In particular, it is an $\frac{1}{3}$-edge-expander (for a large enough $d$).

    \item \pref{ass:v-a-graph}:
    We are to bound the spectral gap in the conditioned $STS$-graph, namely $STS_{a,v}$, whose vertices are $s \supset (a,v)$, and edges are $(s_1,t,s_2)$ where $t \supset (a,v)$. (When we say for instance $s \supset (a,v)$ we mean of course that $s \supset a,v$.)

    In the context of the STAV-structure above, conditioning on $a \in A, v \in X(0)$ is the same as conditioning on $t = a \dunion \set{v} \in T$. In this case, the choices of $s_1,s_2$ are independent - i.e. the graph we get is a clicque with self loops. This graph is a $0$-two-sided spectral expander.

    \item \pref{ass:vasa-graphs-v}: Fix some $v \in V$. The $\vASA{v}$-graph is the graph whose vertices are all $a,a'$ so that $(v,a,s,a')$ are in the support of the $\VASA$-distribution. In this case these are exactly $X_v(\ell-1)$. We go from $a$ to $a'$ by choosing $(v,a,s,a')$ in the $\VASA$-distribution. Thus in this case we go from $a$ to $a'$ if they are disjoint and share a face $s \setminus \set{v} \in X_v(d-1)$.

    The graph we just described is the graph whose double cover is the $(\ell-1),(\ell-1)$-complement walk graph in $X_v$, the link of $v$. $X$ is a $\bigO{\frac{1}{\ell d^2}}$-two-sided HDX, thus by \pref{claim:two-sided-complement-walk-is-a-good-expander}, this graph is a $\ell^2\bigO{\frac{1}{\ell d^2}} = \bigO{\frac{1}{\ell}}$-two sided spectral expander expander.

    \item \pref{ass:vasa-graphs-a}: This is the only part of the proof that is not immediate. Fix some $a \in A$. The graph in this assumption is the $\VAS{a}$-graph, the bipartite graph where
    \[ L = \adj{a_0}, \; R = \sett{(a,s)}{\exists v \in L \; (v,a_0,s,a) \in Supp(D) },\]
    \[ E = \set{(v,(a,s)) : (v,a_0,s,a) \in Supp(D) }. \]
    The probability of choosing an edge $(v,(a',s))$ is given by $\cProb{D}{(v,a_0,s,a')}{a_0 = a}$.

    We describe the graph in this case explicitly in this following proposition, that says this graph is a $\sqrt{\frac{1}{\ell}}$-bipartite expander.

    \begin{proposition}
        \label{prop:two-sided-graph}
    Fix some $a \in A$, and consider the following graph
    \begin{itemize}
        \item $L = \set{(a', s) : a' \dunion a \subset s}.$
        \item $R = \adj{a} = X_a(0).$
        \item $E = \set{(v,(a',s)) : \set{v} \dunion a' \dunion a \subset s},$ and the probability to choose each edge is given by the distribution that chooses $(s, a', v)$ in the link of $a$.
    \end{itemize}

    The graph described above is an $\bigO{\frac{1}{\sqrt{\ell}}}$-bipartite expander.
\end{proposition}

\begin{figure}
    \centering
    \includegraphics[scale=0.6]{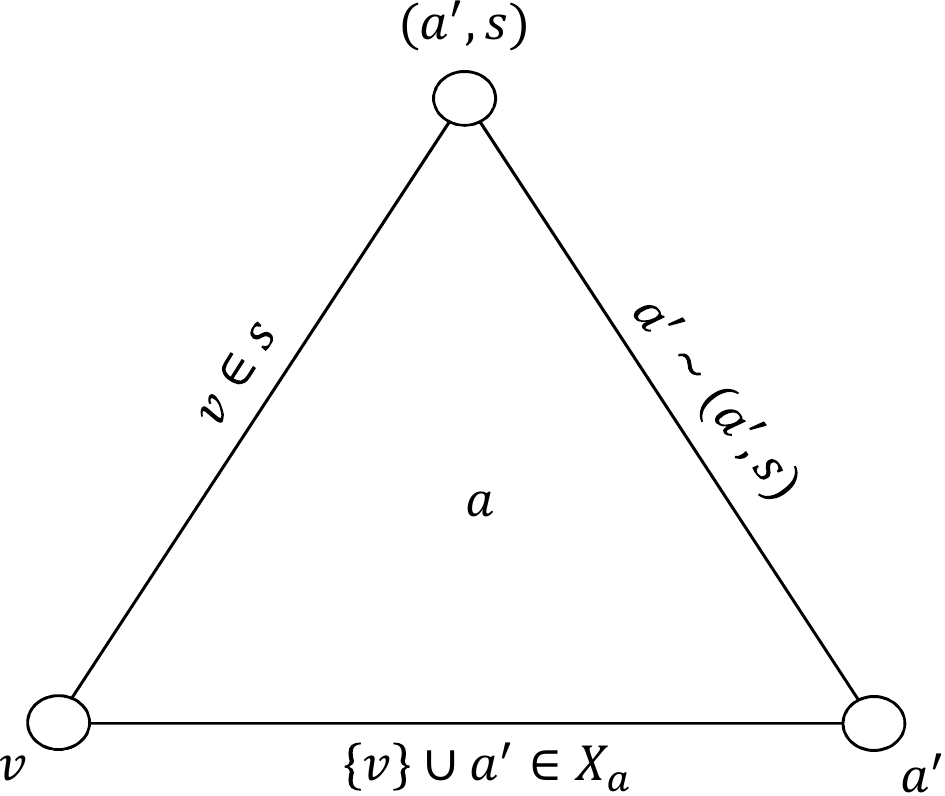}
    \caption{A triangle in $Y$. }
    \label{fig:my_label}
\end{figure}

To prove this proposition, we state \pref{lem:structure-trickling-lemma}. The proof of this lemma uses Garland's method, so we postpone its proof to \pref{sec:complement-walk}.

\begin{lemma}
    \torestate{
    \label{lem:structure-trickling-lemma}
    Let Y be a $2$-dimensional $3$-partite complex, and denote its parts by $X(0) = X[1] \dunion X[2] \dunion X[3]$. Suppose that for every $v \in X[1]$, $X_v$ is a $\eta$-bipartite expander. Denote by $A^{1,2}$, $A^{1,3}$  and $A^{2,3}$ the bipartite walks between $(V_1,V_2)$ $(V_2,V_3)$ and $(V_2,V_3)$ respectively. Then
    \[ \lambda (A^{2,3}) \leq \eta + \lambda(A^{1,2}) \lambda(A^{1,3}) . \]
    }
\end{lemma}

    \begin{proof}[Proof of \pref{prop:two-sided-graph}]
    Consider the following $2$-dimensional $3$-partite simplicial complex:
    \begin{itemize}
        \item The parts of the complex are $Y[1] = X_a(t-1), Y[2] = X_a(0), Y[3] = \sett{(a', s) \in X(t-1) \times X(d)}{a \dunion a' \subset s}$.
        \item We connect $(a',v,(a'',s)) \in Y(2)$ if $a' = a''$ and $\set{v} \dunion a' \dunion a \subset s$. The probability of choosing some triangle $(a',v,(a'',s))$ is the probability of choosing $s$ given $a$, and then choosing $a', v$ (given that they are disjoint from $a$): \[\cProb{X}{s}{a}\cProb{X}{a'}{s \setminus a}\cProb{X}{v}{s \setminus (a \dunion a')}. \]
    \end{itemize}
    We notice the following:
    \begin{enumerate}
        \item $A^{2,3}$ is the bipartite operator of the bipartite walk between $L,R$ in the we defined in the proposition.
        \item $A^{1,2}$ is the bipartite operator of the complement walk in the link of $a$, and from \pref{claim:two-sided-complement-walk-is-a-good-expander} $\lambda(A^{1,2}) \leq d^2 \lambda(X_a) = \bigO{\frac{1}{\ell}}$.
        \item for every $a' \in Y\sqbrack{1}$, the bipartite operator of the link of $a'$ is the containment walk between $X_{a \dunion a'}(0)$ and $X_{a \dunion a'}(d-2\ell)$. Recall that $\frac{1}{2}d - \ell = \Omega(d)$, thus $d-2\ell = \Omega(d)$. Hence this walk is also an $\bigO{\frac{1}{\sqrt{\ell}}}$ expander.
    \end{enumerate}
    Hence we can apply \pref{lem:structure-trickling-lemma} and conclude that
    \[ \lambda(A^{2,3}) \leq \bigO{\frac{1}{\sqrt{\ell}}} + \bigO{\frac{1}{\ell}} = \bigO{\frac{1}{\sqrt{\ell}}}.\]
    \end{proof}
\end{enumerate}
\end{proof}
\subsection{Agreement for One-Sided Partite High Dimensional Expanders}
\label{sec:agreement-in-hdx-one-sided}
We continue and prove an agreement theorem on one-sided partite high dimensional expanders. For a definition of partite simplicial complexes, and other terminology, see \pref{sec:partite-simplicial-complexes}.

In the proof of the two-sided case \pref{thm:main-agreement-theorem-two-sided}, we used a single STAV-structure derived from the sets $X(k),X(\ell),X(\ell-1),X(0)$. In the one-sided case the STAV defined above is not $\gamma$-good, so we need to work a little harder. As it turns out when the one-sided spectral expander is also $(d+1)$-partite, we can use the colored walks to substitute for the complement walk. Details follow.

\subsubsection*{Proof of \pref{thm:main-agreement-theorem-one-sided}}
As in \pref{thm:main-agreement-theorem-two-sided}, we are given an ensemble $\FF$ with $\disagr{D_{d,\ell}}(\FF) = \varepsilon$. We need to find a global function $\glob:X(0) \to \Sigma$ so that
\[ \Prob[s]{f_s \ane{\gamma} \rest{G}{s}} = \bigO{\varepsilon}.\]

Without loss of generality, $\ell > 1$. For any two disjoint colors $I,J$ of size $\ell$, we define the $I,J$-STAV-structure as follows:
\begin{enumerate}
    \item $S_{(I,J)} = \sett{s \in X(k)}{col(s) \supset I \dunion J}$.
    \item $T_{(I,J)} = \sett{t \in X(\ell)}{col(t) \cap (I \dunion J) \in \set{I,J}}$, i.e. $t$ so that it's color contains $I$ and is disjoint from $J$, or vice versa.
    \item $A_{(I,J)} = \sett{a \in X(\ell-1)}{col(a) = I \text{ or } col(a)=J}$.
    \item $V_{(I,J)} = \sett{v \in X(v)}{col(v) \notin I \dunion J}$.
\end{enumerate}

The $\sts$-test associated with the STAV-structure is:
\begin{enumerate}
    \item Choose $t \in X(\ell)$ given that $col(t)$ either contains $I$ and is disjoint from $J$, or contains $J$ and is disjoint from $I$.
    \item Choose $s_1,s_2 \supset t$ independently given that $col(s_i) \supset I\dunion J$ for $i=1,2$.
\end{enumerate}
We denote the test associated with this STAV-structure as the \emph{$I,J$-STAV-test}.

The $I,J$-STAV distribution is choosing $s,t$ as above, and then setting $a \subset t$ so that $col(a) = I$ or $col(a) = J$, and $\set{v} = t \setminus a$.  We denote the $I,J$-STAV distribution by $D_{I,J}$.
These STAV-structures come with $\VASA$-distributions that are choosing $v,s$ as in the $I,J$-STAV distribution, and taking $a_1,a_2$ to be the subsets of $s$ of colors $I,J$ respectively. 

We denote the surprise of $\FF$ in the $I,J$-STAV structure by $\surp_{(I,J)}(\FF)$, and the rejection probability by $\disagr{I,J}(\FF)$.

\begin{lemma}
\label{lem:one-sided-stav-are-1-over-t-good}
For any two disjoint colors $I,J$, each of size $\ell$, the STAV-structure above is $\gamma = \bigO{\frac{1}{\ell}}$-good.
\end{lemma}

For a pair of disjoint $(I,J)$-we would like to define a global functions $G_{I,J}$, that will be defined on all vertices so that $col(v) \notin I,J$ (using  \pref{thm:main-STAV-agreement-theorem}). After that, we would like to stitch the $G_{I,J}$'s together. In fact, we only need two such global functions, to cover vertices of all colors.

However, in order to invoke \pref{thm:main-STAV-agreement-theorem}, we need that both $\disagr{I,J}(\FF)=O(\varepsilon)$ and $\surp_{(I,J)}(\FF) = \bigO{\frac{1}{\ell}}$. Furthermore, we will need to use the $(d,\ell)$-agreement test to stitch the two global functions together.

We define an additional agreement test. This test will be used to stitch the $G_{I,J}$'s together. We call it the \emph{$(I,J)$-in-one-set test}:
\begin{enumerate}
    \item Choose $t \in X(\ell)$ (with no conditioning on the color).
    \item Choose $s_1, s_2 \supset t$ independently given that $\col(s_1) \supset I,J$.
\end{enumerate}
We denote the rejection probability of this test by $\disagr{I,J}^{1-set}(\FF)$.

The following lemma states formally what we require from the $I,J$-STAV distributions:
\begin{lemma}
\label{lem:four-good-colors}
There exists four disjoint colors $I_1,J_1,I_2,J_2$ where for $i = 1,2$:
\begin{enumerate}
    \item $\disagr{I_i,J_i}(\FF) = \bigO{\varepsilon}$.
    \item $\surp_{(I_i,J_i)}(\FF) = \bigO{\frac{1}{\ell}}$.
\end{enumerate}

Furthermore, we can require from $I_i,J_i$ that
\begin{enumerate}[start=3]
    \item $\disagr{I_i,J_i}^{1-set}(\FF) = \bigO{\varepsilon}$.
\end{enumerate}
\end{lemma}

Given the first two items in the lemma above, we can invoke \pref{thm:main-STAV-agreement-theorem} to get a global function
$G_{I_i,J_i}:V_{(I,J)} \to \Sigma$ so that for $i=1,2$
\[ \Prob[s \sim D_{I_i,J_i}]{f_s \ane{\gamma} \rest{G_{I_i,J_i}}{s}} = \bigO{\varepsilon}. \]

We glue these two functions to one global function $G: X(0) \to \Sigma$:
\[ G(v) = \begin{cases}
    G_{I_1,J_1}(v) & col(v) \notin I_1 \dunion I_2.\\
    G_{I_2,J_2}(v) & otherwise.
    \end{cases}
\]

Here's a short and informal overview the proof of the theorem given \pref{lem:four-good-colors}: We will choose $(s,t)$ as in the $d,\ell$-agreement test. Then we will choose an additional $t \subset s_1,s_2$, so that $s_1\supset I_1\dunion J_1$ and $s_2\supset I_2\dunion J_2$.

On the one hand, for $i=1,2$ the choice of $(s,t,s_i)$ is done as in the $(I,J)$-in-one-set distribution. By the third item of \pref{lem:four-good-colors}, $\rest{f_s}{t} \ne \rest{f_{s_i}}{t}$ with probability $\bigO{\varepsilon}$.

 On the other hand, by the first two items in \pref{lem:four-good-colors}, $\prob{\rest{f_{s_i}}{t} \ane{\gamma} \rest{G}{t}} = \bigO{\varepsilon}$. By union bound, we will get our theorem. Details follow.

\begin{proof}[Proof of \pref{thm:main-agreement-theorem-one-sided}]
First, we show that in order to prove
\[\Prob[s]{f_s \ane{\gamma} \rest{G}{s}} = \bigO{\varepsilon}\]
it is enough to prove that
\begin{equation} \label{eq:one-sided-key-probability}
\Prob[s \in x(d),t \subset s, t \in X(\ell)]{\rest{f_s}{t} \ane{\frac{1}{3}\gamma} \rest{G}{t}} = \bigO{ \varepsilon}.
\end{equation}

Denote by
\[H = \left\{s: \Prob[v \in s]{f_s(v) \ane{\gamma} G(v)} \geq \frac{1}{\ell}\right\}.\]
We need to show that given \eqref{eq:one-sided-key-probability}, $\prob{H} = \bigO{\varepsilon}$. Fix some $s \in H$, i.e. $f_s \ane{\gamma} \rest{G}{s}$.

Consider the following containment graph for $s$:
\[ L = s - \text{ the vertices in $s$,}\]
\[ R = \binom{s}{\ell}, \text { the subsets of $s$ of size $\ell$.}\]
This graph is a $\frac{2}{\sqrt{\ell}}$-bipartite expander by \pref{thm:containment-graph-is-a-good-expander}.
By \pref{lem:sampler-lemma}, this graph is a \emph{$\bigO{\frac{1}{\ell}}$-sampler graph}. Hence if
\[ \Prob[v \in s]{f_s(v) \ane{\gamma} G(v)} \geq \frac{1}{\ell}\]
then the set
\[T^*_s  = \sett{t \in \binom{s}{\ell}}{\Prob[v \in t]{f_s(v) \ne G(v)} < \frac{1}{3}\gamma }\]
has probability of at most $\frac{1}{3}$. In other words, at least $\frac{2}{3}$ of the $t \in \binom{s}{\ell}$ have that property that $\rest{f_s}{t} \ane{\frac{1}{3}\gamma} \rest{G}{t}$.

Hence $\bigO{\varepsilon} \geq \frac{2}{3}\prob{H}$, and we conclude that there may be on $\bigO{\varepsilon}$ of $s$'s so that $f_s \ane{\gamma} \rest{G}{s}$.

\bigskip

We move to showing \eqref{eq:one-sided-key-probability}. Observe the following distribution:
\begin{enumerate}
    \item Choose $s \in X(k)$ and $t \subset s$ according to the probability of the simplicial complex.
    \item Choose $\Delta \in X(d)$ given that $t \subset \Delta$.
    \item Choose two $s_1,s_2 \subset \Delta$ given that they also contain $t$, and so that $I_1 \dunion J_1 \subset \col(s_1)$ and $I_2 \dunion J_2 \subset \col(s_2)$.
\end{enumerate}

In a simplicial complex, a $k$-face $s_1$ (respectively $s_2$) is chosen by choosing a $d$-face $\Delta \in X(d)$ and choosing $s \subset \Delta$. Thus in this distribution $(s,t,s_1,s_2)$ are chosen so that
the marginals $(s,t,s_i)$ are chosen according to the \emph{$(I_i,J_i)$-in-one-set test}.

By the last item of \pref{lem:four-good-colors}, $f_s$ disagrees  on $t$ with $s_1$ or $s_2$, with probability $\bigO{\varepsilon}$.

\bigskip

Denote  $t_1 = \sett{v\in t}{col(v)\not\in I_1\dunion J_1}$ and $t_2 = \sett{v\in t}{col(v)\not\in I_2\dunion J_2}$, clearly $t = t_1 \cup t_2$ and some vertices might appear in both sets.

We would like to invoke \pref{thm:stronger-equation-in-main-theorem}, the extension to \pref{thm:main-STAV-agreement-theorem} to get that
\[ \Prob{\rest{f_{s_i}}{t_i} \ane{\frac{1}{2}\gamma} \rest{G}{t_i}} = \bigO{\varepsilon}, \]
Since if $|t_i| \leq \ell+1$, this implies that $\rest{f_{s_i}}{t_i} = \rest{G}{t_i}$.
Indeed by \pref{lem:four-good-colors}, we know that $\disagr{I_i,J_i}(\FF)=\bigO{\varepsilon}$ in the $I_i,J_i$-STAV-structure, and that $\surp(X_{(I,J)},F) = \bigO{\frac{1}{\ell}}$. Consider the sampling of $(v,a,s,t_i)$ where $a \subset s$ is of color $I_i$ or $J_i$, and $v \in t_i$ is chosen uniformly at random. It holds that $(v,a,s)$ is chosen as in the $I_i,J_i$-STAV-structure, hence by \pref{thm:stronger-equation-in-main-theorem},
\[ \Prob{\rest{f_{s_i}}{t_i} \ane{\gamma} \rest{G}{t_i}} = \bigO{\varepsilon}. \]

by the statement \eqref{eq:stronger-equation-in-main-theorem} of \pref{thm:main-STAV-agreement-theorem}.

\medskip

Hence we bound the probability \eqref{eq:one-sided-key-probability} by
\[ \leq \sum_{i=1}^2 \left ( \disagr{I_i,J_i}(\FF) + \prob{\rest{f_{s_i}}{t_i} \ane{\frac{1}{2}\gamma} \rest{G}{t_i}} \right ).\]
Which is $\bigO{\varepsilon}$ by \pref{lem:four-good-colors}.
\end{proof}

\begin{proof}[Proof of \pref{lem:four-good-colors}]
For this lemma, we consider the uniform distribution on the $4$-tuples of disjoint colors $I_1,J_1,I_2,J_2 \sim C_{\ell}^{\dunion}$.

To show there exists four disjoint colors $I_1,J_1,I_2,J_2$ with the properties in the lemma statement, we show that each property is satisfied separately with large probability, thus their intersection has non-zero probability as well. We do this by an expectation argument, and then use Markov's inequality.

\paragraph{Step 1: more than $0.8$ of the colors $4$-tuples $I_1,J_1,I_2,J_2$ satisfy the second item in \pref{lem:four-good-colors}}That is, we show the ``surprise'' $\surp_{(I_i,J_i)}(\FF) = \bigO{\frac{1}{\ell}}$.

For this, note that for \emph{any} color $J$,
\[ \cProb{(s_1,s_2,a,v)}{\rest{f_{s_1}}{t} \ne \rest{f_{s_2}}{t}}{\rest{f_{s_1}}{a} = \rest{f_{s_2}}{a} \ve B_J} = \bigO{\frac{1}{\ell}}, \]
where $B_J$ is the event that $\col(s_i \setminus (a\dunion \set{v})) \supset J$ for $i=1,2$. This is due to the same argument in \pref{claim:agreement-on-last-vertex}. Hence

\[\Ex[I_1,J_1,I_2,J_2 \sim C_{\ell}^{\dunion}]{\sum_{i=1}^2 \cProb{(s_1,s_2,a,v)}{f_{s_1}(v) \ne f_{s_2}(v)}{\rest{f_{s_1}}{a} = \rest{f_{s_2}}{a} \ve B_{J_i} \ve \col(a) = I_i } }\]
\[ = \bigO{\frac{1}{\ell}}. \]

By Markov's inequality $0.8$ of the $4$-tuples $I_1,J_1,I_2,J_2$ satisfy that
\[ \cProb{(s_1,s_2,a,v)}{f_{s_1}(v) \ne f_{s_2}(v)}{\rest{f_{s_1}}{a} = \rest{f_{s_2}}{a} \ve B_{J_i} \ve \col(a) = I_i } = \bigO{\frac{1}{\ell}}. \]

\paragraph{Step 2: more than $0.8$ of the $4$-tuples of disjoint colors $I_1,J_1,I_2,J_2$ satisfy the first item in \pref{lem:four-good-colors}} That is, that when we choose $(s_1,t,s_2)$ according to the $I_i,J_i$-STAV distribution, then the rejection probability is $\bigO{\varepsilon}$.

First recall that by our assumption
\[ \disagr{D_{d,\ell}}(\FF) = \Prob[(s_1,t,s_2)\sim D_{d,\ell}]{f_{s_1} \ne f_{s_2}} = \varepsilon.\]

We can condition this test on $\col(t) = I \cup \set{p} \text{ or } J \cup \set{p}$  for $p \notin I,J$ and on $\col(s_1) \supset I,J$ (but no conditioning on $s_2$).

This conditioning is different from the $I,J$-STAV-structure $STS$-test since we don't condition on $\col(s_2) \supset I,J$. It is also different from the $I,J$-in-one-set test since we \emph{do} condition on $\col(t) = I \cup \set{p} \text{ or } J \cup \set{p}$.

Denote the probability for this conditioned agreement test by $\disagr{I,J}^*(\FF)$. We know that
\[ \Ex[I,J]{\disagr{I,J}^*(\FF)} = \disagr{D_{d,\ell}}(\FF) = \varepsilon. \]
Hence by Markov's inequality, $0.8$ of the disjoint $4$-tuples $I_1,J_1,I_2,J_2$ satisfy $\disagr{I_i,J_i}^*(\FF) \leq \bigO{\varepsilon}$.

For a pair $I_i,J_i$, we think about the following experiment $(s_1,t,s_2,s')$:
\begin{enumerate}
    \item Choose $(s_1,t,s_2)$ as in the $Agree(\FF)_{I,J}$ test, i.e.  $\col(t) = I_i \cup \set{p} \text{ or } J_i \cup \set{p}$  for $p \notin I_i,J_i$ and $\col(s_1) \supset I_i,J_i$.
    \item Choose $\tilde{s_1}$, given $s_1 \subset t$ \emph{and} conditioning on $\col(\tilde{s_1}) \supset I_i, J_i$.
\end{enumerate}
Observe the following:
\begin{enumerate}
    \item The marginal $(s_1,t,\tilde{s_1})$ is according to the agreement test in the $I,J$-STAV-structure $STS$-test.
    \item The marginals $(s_1,t,s_2)$ and $(s_1,t,\tilde{s_1})$ is according to the  $\disagr{I,J}^*(\FF)$ test.
\end{enumerate}

If $\disagr{I_i,J_i}^*(\FF) = \bigO{\varepsilon}$, then by a union bound we get that
\[ \prob{\rest{f_{s_1}}{t} \ne \rest{f_{\tilde{s_1}}}{t}} \leq 1-2\prob{\rest{f_{s_1}}{t} \ne \rest{f_{s_2}}{t}} \rem{= 2\bigO{\varepsilon}} = \bigO{\varepsilon}.\]

\paragraph{Step 3: more than $0.8$ of the colors $4$-tuples $I_1,J_1,I_2,J_2$ satisfy the third item in \pref{lem:four-good-colors},} That is, that when we choose $(s_1,t,s_2)$ by the $I,J$-in-one-set distribution is $\disagr{I,J}^{1-set}(\FF) = \bigO{\varepsilon}$.

This step follows the same reasoning as in step 1 or 2. the agreement in the $d,\ell$-agreement test is $\varepsilon$. Hence by Markov's inequality, $0.8$ of pairs $I_1,J_1,I_2,J_2$ have the property that $\disagr{I,J}^{1-set}(\FF) = \bigO{\varepsilon}$ in the $I,J$-in-one-set distribution test.

\bigskip

From the three steps above, the size of the intersection of the three properties is lower bounded by $0.4$, by a union bound. In particular it is not empty.

\end{proof}

We move towards proving \pref{lem:one-sided-stav-are-1-over-t-good}. We need the following proposition, that containment walks in one-sided high dimensional expanders have a spectral gap even when conditioning on color:
    \begin{proposition}
        \label{prop:one-sided-weird-graph}
        Let $X$ be a $\gamma$-one sided $(d+1)$-partite high dimensional expander. Let $J$ be a color of size $\ell$. Consider the following graph:
        \begin{itemize}
            \item $L = \sett{v \in X(0) }{ col(v) \notin J}$.
            \item $R = \sett{ s \in X(k)}{ J \subset col(s)}$.
            \item $E = \set{(v,s) : v \subset s}$, where the probability of an edge is $\Pr{(v,s)}$ is to choose $s \in X(k)$ given that $J \subset col(s)$, and then choose $v \in s$ uniformly at random given that $col(v) \notin J$.
        \end{itemize}
        The this graph is a $\bigO{\frac{1}{\sqrt{k-\ell}}} + k\gamma$-bipartite expander.
    \end{proposition}
We prove this proposition after the proof of \pref{lem:one-sided-stav-are-1-over-t-good}.

\begin{proof}[Proof of \pref{lem:one-sided-stav-are-1-over-t-good}]
Again, we may assume that $\ell >> 1$. Fix some disjoint colors $I,J$, and consider the $I,J$-STAV structure. We show the five assumptions in \pref{def:good-stav-structure} hold for $\gamma = \bigO{\frac{1}{\ell}}$:
\begin{enumerate}
    \item \pref{ass:a-not-large}: We need to show that $\cProb{v}{v \in \adj{a}}{v \in s} \geq \frac{1}{2}$. This assumption holds trivially since because in these STAV-structures $col(a) \in \set{I,J}$ and the vertices in $v$ don't have colors $I,J$, and given $s$, all we can choose all possible pairs $(a,v)$ with these colors. Hence when choosing $v \in s$ it is always in $\adj{a}$.

    \item \pref{ass:global-A-V-graph}: We need to show that the global graph between $A$ and $V$, where choosing an edge is choosing a pair $(a,v)$ in the STAV-distribution is a $\bigO{\frac{1}{\ell}}$-bipartite expander. In this case, this graph is the graph of all $(v,a)$ where $col(a) \in \set{I,J}$ and $v \notin I,J$. Note that we can decompose this random walk to a convex combination of colored walks $\comp{I,k},\comp{J,k}$ for colors $k \notin I \dunion J$. For each $k$, this colored walk is $\bigO{\frac{\ell}{d^2 \ell}} = \bigO{\frac{1}{\ell}}$-bipartite expander by \pref{claim:colored-walk-is-a-good-expander}. Hence, the combination of walks is also a $\bigO{\frac{1}{\ell}}$-bipartite expander.

    \item \pref{ass:edge-expander}: Fix some $a \in A$. the $STS_a$-graph is the graph where we choose $(s_1,t,s_2)$ given that they all contain $a$. This graph is (isomorphic to) the graph whose vertices are $s \in X_a(d-\ell)$, and we connect $s_1,s_2$ if they share a vertex in $X_a(0)$ whose color is not in $J$. Taking a step in this graph is like taking two steps in the graph described in \pref{prop:one-sided-weird-graph} if we begin with some $s$. Hence by \pref{prop:one-sided-weird-graph}, this is a $\left ( \bigO{\frac{1}{\sqrt{k-2\ell}} + \frac{k}{\ell k^2}} \right )^2 = \bigO{\frac{1}{\ell}}$-two-sided spectral expander. In particular, for $\ell$ large enough, this is a $\frac{1}{3}$-edge expander.

    \item \pref{ass:v-a-graph}:
    We are to bound the spectral gap in the conditioned $STS$-graph, namely $STS_{a,v}$, whose vertices are $s \supset (a,v)$, and edges are $(s_1,t,s_2)$ where $t \supset (a,v)$.

    In the context of the $I,J$-STAV-structures, conditioning on $a \in A, v \in X(0)$ is the same as conditioning on $t = a \dunion \set{v} \in T$. In this case, the choices of $s_1,s_2$ are independent - i.e. the graph we get is a clicque with self loops. This graph is a $0$-two-sided spectral expander.

    \item \pref{ass:vasa-graphs}: We define the following $\VASA$-distribution:
    \begin{enumerate}
        \item Choose $s \in S$ (i.e. $s \in X(k)$ so that $\col(s) \supset I,J$).
        \item Set $a_1,a_2 \subset s$ so that $col(a_1)= I, \; col(a_2) = J$.
        \item Choose some $v \in s$ so that $col(v) \notin I,J$.
        \item Output either $(v,a_1,s,a_2)$ or $(v,a_2,s,a_1)$ with equal probability.
    \end{enumerate}
    This distribution is symmetric with respect to $a_1,a_2$. Furthermore when we restrict to one of the marginals $(v,a,s)$ or $(v,s,a')$, this is precisely the distribution in the STAV-structure. Hence this is indeed a $\VASA$-distribution.

    \item \pref{ass:vasa-graphs-v}: Fix some $v \in V$ and consider the $\vASA{v}$-graph. In the case of the $I,J$-STAV structure, this graph is all the bipartite graph where $L = X_v[I ],\; R = X_v[J ]$ and we connect $a,a'$ if they share some $s \in X_v(k - 1)$. This is the $(I,J)$-colored walk in the link of $v$. By \pref{claim:colored-walk-is-a-good-expander}, this is a $\bigO{\frac{\ell^2}{\ell d^2}} = \bigO{\frac{1}{\ell}}$-bipartite expander.

    \item \pref{ass:vasa-graphs-a}: Fix some $a \in A$, and without loss of generality its color is $I$. The graph in this assumption is the $\VAS{a}$-graph, the bipartite graph where
    \[ L = \adj{a_0}, \; R = \sett{(a,s)}{\exists v \in L \; (v,a_0,s,a) \in Supp(D) },\]
    \[ E = \sett{(v,(a,s)) }{ (v,a_0,s,a) \in Supp(D) }. \]
    The probability of choosing an edge $(v,(a',s))$ is given by $\cProb{D}{(v,a_0,s,a')}{a_0 = a}$.

    In this case our graph is the graph where $L = \sett{v \in X_a(0)}{\col(v) \notin J}$ and $R = \sett{(a',s)}{col(a') = J, a \subset s}$. Notice that $s$ has exactly one subset of color $J$ hence $R$ is (isomorphic to) the set $\sett{s \in X_a(d - \ell)}{J \subset \col(s)}$. This is the graph we described in \pref{prop:one-sided-weird-graph}, hence by that proposition it is a $\bigO{\frac{1}{\sqrt{d-2\ell}}} =\bigO{\frac{1}{\sqrt{\ell}}}$-bipartite expander.
\end{enumerate}
\end{proof}

\begin{proof}[Proof of \pref{prop:one-sided-weird-graph}]
This proof is similar to the proof of \pref{prop:two-sided-graph}. We build a $3$-partite complex where the bipartite graph is a walk between two of its sides and use \pref{lem:structure-trickling-lemma}.

Consider the following $2$-dimensional $3$-partite simplicial complex:
\begin{itemize}
    \item The parts of the complex are $Y[1] = X[J], Y[2] = \sett{v \in X(0)}{col(v) \notin J}, Y[3] = \sett{s \in X(k)}{J \subset col(s)}$.
    \item We connect $(a,v,s) \in Y(2)$ if $a \dunion \set{v} \subset s$. The probability of choosing some triangle $(a,v,s)$ is the probability of choosing $a \in X[J]$, and then choosing $v \subset s \setminus a$ from the link of $a$: \[\Prob[X\sqbrack{J} ]{a}\cProb{X}{s}{a}\cProb{X}{v}{a, s}. \]
\end{itemize}
We notice the following:
\begin{enumerate}
    \item $A^{2,3}$ is the bipartite operator of the bipartite walk between $L,R$ in the graph we defined in the proposition.
    \item $A^{1,2}$ is the convex combination of the bipartite operators of the colored walks $\comp{J,i}$ for all $i \notin J$. From \pref{claim:colored-walk-is-a-good-expander} $\lambda(A^{1,2}) \leq k \gamma$.
    \item for every $a \in Y\sqbrack{1}$, the bipartite operator of the link of $a$ is the containment walk between $X_{a}(0)$ and $X_{a}(k-\ell)$. Hence this walk is also an $\bigO{\frac{1}{\sqrt{k-\ell}}}$ expander.
\end{enumerate}
Hence we can apply \pref{lem:structure-trickling-lemma} and conclude that
\[ \lambda(A^{2,3}) \leq \bigO{\frac{1}{\sqrt{k-\ell}}} + k\gamma.\]
\end{proof}

\section{Agreement on Vertex Neighborhoods}
\label{sec:agreement-on-nbrhoods}

In this section we consider a number of new set systems. The sets in this set system consist of neighbors of a given vertex (or higher dimensional face). This resembles the set system underlying the gap-amplification based proof of the PCP theorem \cite{Din07}, in which an agreement theorem underlies the soundness proof somewhat implicitly.

Given a simplicial complex $X$, for a vertex $z\in X(0)$ we denote by $Ball_z$ the set of vertices adjacent to $z$ (recall that even if $X$ has high dimensional faces, it must also have edges). More generally, for a face $z\in X(k)$ we let $Ball_z = \sett{ v\in X(0)\setminus z}{v\cup z\in X}$ ($Ball_z$ is just the set of vertices in the link of $z$).

Our next agreement testing theorem is for the family $S=\sett{B_z}{z\in X(k)}$ whose ground set is $V=X(0)$.
In this section we abuse notation and refer to $f_{Ball_z}$ by $f_z$.%
%A link of a face $t \in X$ is a generalization of a neighbourhood of a vertex in a graph. The link $X_t$ is all the faces $t \setminus r$ for $t \subset r \in X$. In this context, an ensemble of local functions is\[ \FF = \sett{f_z:X_z(0) \to \Sigma}{z \in X(k)}.\]
%It is useful to think of the case where $k=0$. Here every $f_z$ is defined on $B_z$the neighbourhood of the vertex $z$, in the underlying graph of $X$, which is just the graph $(V = X(0), E = X(1))$.

We describe a couple of test distributions on such an ensemble:
\begin{definition}[Neighborhood independent agreement distribution]\torestate{\label{def:neighbourhood-ind-agreement-dist}
Let $X$ be a $d$-dimensional simplicial complex, and let $\ell,k$ be non-negative integers so that $\ell + k + 1 \leq d$. We define the distribution $NID_{\ell,k}$ by the following random process:
    \begin{enumerate}
        \item Sample $t \in X(\ell)$.
        \item Sample $z_1,z_2 \in X_t(k)$ independently.
    \end{enumerate}}
\end{definition}

\begin{definition}[Neighborhood complement agreement distribution]\torestate{\label{def:neighbourhood-comp-agreement-dist}
Let $X$ be a $d$-dimensional simplicial complex, so that $\ell + 2k+2 \leq d$. We define the distribution $NCD_{\ell,k}$ by the following random process.
\begin{enumerate}
    \item Sample $t \in X(\ell)$.
    \item Sample $z_1,z_2 \in X_t(k)$ by the $k,k$-complement walk in $X_t$.
\end{enumerate}}
\end{definition}
Note that $z_1,z_2$ are distributed as in the $k,k$-complement walk distribution.

Whereas usually an agreement test selects two subsets $s_1,s_2$ and checks if $f_{s_1}$ agrees with $f_{s_2}$ on their entire intersection, it sometimes makes sense to choose a random $t \subset s_1\cap s_2$ and check that $f_{s_1}$ and $f_{s_2}$ agree only on $t$.
For this section we call such tests {\em weak} and define two agreement test of this form.
\begin{enumerate}
    \item In the \emph{weak independent agreement test} we sample $(t,z_1,z_2) \sim NID_{\ell,k}$ and accept if $\rest{f_{z_1}}{t} = \rest{f_{z_2}}{t}$.
    \item In the \emph{weak complement agreement test} we sample $(t,z_1,z_2) \sim NCD_{\ell,k}$ and accept if $\rest{f_{z_1}}{t} = \rest{f_{z_2}}{t}$.
\end{enumerate}
If our simplicial complex is a two-sided high dimensional expander, then we can show that these agreement tests have some soundness, even in their weak variant:
\begin{theorem}[Agreement on neighborhoods]
\torestate{\label{thm:agreement-on-links}
There exists a constant $c > 0$ such that for every non-negative integers $\ell,k,d$ such that $4 \leq \ell \leq \frac{d-2}{2}$ and $\ell+2k+2 \leq d$, the following holds.
Let $X$ be a $d$-dimensional $\frac{1}{\ell \, (k+\ell)^2}$-two-sided high dimensional expander. Then
the $\ell,k$-weak independent agreement test and the $\ell,k$-weak complement agreement test are both $\frac{1}{\ell}$-approximately $c$-sound.
}
\end{theorem}
Clearly if $\rest{f_{z_1}}{I} = \rest{f_{z_2}}{I}$ for $I = B_{z_1}\cap B_{z_2}$ then  $\rest{f_{z_1}}{t} = \rest{f_{z_2}}{t}$ since $t\subset I$. Therefore, the theorem also holds if we make a stronger agreement test that checks agreement on the entire intersection. The current statement is stronger because it begins from a weaker assumption. However, it could be that if we make the stronger test, we could reach an even stronger conclusion in terms of the closeness of the ensemble to a perfect ensemble. This is an interesting direction for further study.
\begin{proof}[ Proof of \pref{thm:agreement-on-links}]

As in the proof of \pref{thm:main-agreement-theorem-two-sided}, we have an ensemble of functions $\FF$ that has $\disagr{}(\FF) = \varepsilon$ by either the independent agreement distribution, or the complement agreement distribution. We need to find a global function $G$ so that
\[ \Prob[s]{f_s \ane{\frac{1}{\ell}} \rest{G}{s}} = \bigO{\varepsilon}. \]
We do so using \pref{thm:main-STAV-agreement-theorem}.
For both distributions our STAV-structure is the following:
\begin{enumerate}
    \item $S = \sett{Ball_z}{z \in X(k)}$.
    \item $T = X(\ell)$.
    \item $A = X(\ell-1)$.
    \item $V = X(0)$.
\end{enumerate}
As noted before, whenever we choose $z$, we will always mean that we choose $Ball_z \in S$.
The STAV-structure's distribution will be $(z,t,a,v)$ where $z \in X(k)$, $t \in X_z(\ell)$, and $t = a \dunion \set{v}$ for a partition chosen uniformly at random.
Note that $(z,t)$ are chosen as the marginal of both the independent agreement test and the weak complement agreement test.

Given any fixed $t$, the independent agreement distribution samples $z_1,z_2 \supset t$ independently. The complement agreement distribution does not sample $z_1,z_2 \supset t$ independently, but according to an expanding random walk. In \pref{claim:independent-vs-expander}, we prove that in this case $\disagr{NID_{\ell,k}}(\FF) = \Theta(\disagr{NCD_{\ell,k}}(\FF))$. Thus it is enough to prove the theorem on the independent agreement distribution.

By \pref{claim:agreement-on-last-vertex} we know that $\surp(\FF) = \frac{1}{\ell}$. If we show that this STAV-structure is $\bigO{\frac{1}{\ell}}$-good, we can directly obtain the theorem by invoking \pref{thm:main-STAV-agreement-theorem}. We check that this STAV-structure fulfils the assumptions:
\begin{enumerate}
    \item \pref{ass:global-A-V-graph}: The graph between $A$ and $V$ whose edges are $(a,v)$ so that $a \dunion v \in X(\ell)$ is the $0,\ell$-complement walk. This graph is a $\bigO{\frac{\ell}{\ell(\ell+k)^2}} = \bigO{\frac{1}{\ell}}$-bipartite expander, by \pref{claim:two-sided-complement-walk-is-a-good-expander}.

    \item \pref{ass:edge-expander}: The $STS_a$-graph here is the graph where we choose $v \in X_a(0)$ and then choose independently two edges $v \dunion z_1$ and $v \dunion z_2$, and output $z_1,z_2$. This is just taking two steps in the $0,k$-complement walk in $X_a$, thus by \pref{claim:two-sided-complement-walk-is-a-good-expander}, this is a $\frac{1}{\ell}$-two-sided spectral expander. As $\ell \geq 4$, this is also a $\frac{1}{3}$-edge expander.

    \item \pref{ass:v-a-graph}: The $STS_{a,v}$-graph here is the graph obtained after choosing two $k$-faces in the link of $X_{a\dunion \set{v}}$  independently. Similar to the previous items in this section, this is a $0$-two-sided spectral expander.

\rem{
\paragraph{For the weak complement agreement test}
\begin{enumerate}[start=3]
    \item \pref{ass:edge-expander}: The $STS_a$-graph here is the graph where we choose $z_1,z_2 \in X_a(k)$, so that they are in the complement walk of some $v \in X_a(0)$. This is the same as choosing $z_1,z_2$ by the $k,k$-complement walk of $X_a$. Thus by \pref{claim:two-sided-complement-walk-is-a-good-expander} we obtain that this is a $\bigO{\frac{1}{\ell}}$-two sided spectral expander.

    \item \pref{ass:v-a-graph}: The $STS_{a,v}$-graph here is the graph obtained after choosing two vertices in the complement walk of $X_{a\dunion \set{v}}$. From our assumption, this is a $\bigO{\frac{1}{\ell}}$-two-sided spectral expander.
\end{enumerate}
}
    \item \pref{ass:vasa-graphs}: Consider the following $\VASA$-distribution.
    \begin{enumerate}
        \item Choose $z \in X(0)$ and $v \in X_z(0)$.
        \item Choose $a_1,a_2$ in the complement walk in the link of $\set{z,v}$.
        \item Output either $(v,a_1,z,a_2)$ or $(v,a_2,z,a_1)$ with probability $\frac{1}{2}$.
    \end{enumerate}
    This is symmetric in $a_1,a_2$. It is easy to verify that the marginals $(v,z,a_1)$ and $(v,z,a_2)$ are just the same as choosing according to the STAV-distribution.

    \item \pref{ass:vasa-graphs-v}: For each $v \in X(0)$, the $\vASA{v}$-graph here is just the $0,\ell-1$-complement walk in $X_v$. Hence by \pref{claim:two-sided-complement-walk-is-a-good-expander}, this is a $\bigO{\frac{1}{\ell}}$-two sided spectral expander.

    \item \pref{ass:vasa-graphs-a}: Finally, given $a$, the $\VAS{a}$-graph is the graph where
    \[ L = X_a(0).\]
    \[ R = \sett{(z,a')}{a \dunion z \dunion a' \in X}.\]
    We connect $(v,(z,a'))$ if $a' \in X_{\set{v,z}\dunion a}$. We can decompose this graph to two independent steps in two bipartite graphs. Denote $M = X_a(\ell)$.
    If we consider the complement walk between Between $L$ and $M$, and the graph between $M$ and $R$ where every $t$ is connected to $(z,a')$ so that $t = \set{z} \dunion a'$. It is easy to see that a step from $L$ to $R$ is two independent steps between $L$ and $M$, and then $M$ and $R$. By \pref{claim:two-sided-complement-walk-is-a-good-expander}, the step between $L$ and $M$ is a $\bigO{\frac{1}{\sqrt{\ell}}}$-expander, and thus the $\VAS{a}$-graph is a $\bigO{\frac{1}{\sqrt{\ell}}}$-expander.

    \item \pref{ass:AVS-graph-good-sampler}: We show that for every $z$, the $\AV{z}$-graph is a $\frac{1}{\ell}$-sampler graph. In this case the $\AV{z}-$-graph is a bipartite graph where
    \[L = X_z(\ell-1), \; R = X_z(0). \]
    and the edges are $(a,v)$ so that $a \dunion \set{v} \in X_z$, i.e. the $0,\ell-1$-complement walk in $X_z$. We need to show that this graph is a $\frac{1}{\ell}$-sampler, namely that if $C \subset R$ is of size $\prob{C} \geq \frac{1}{\ell}$ then the set
    \[ T = \sett{a \in L}{\cProb{v \in R}{v \in C}{v \sim a} \geq \frac{1}{3\ell}} \]
    is at least of size $\prob{T} \geq \frac{1}{3}$.
    Indeed, the complement set $L \setminus T$ is contained in the set of all $a \in R$ so that $\abs{\cProb[v \in R]{v \in C}{v \sim a} - \prob{C}} \geq \frac{2}{3}\prob{C}$.

    This walk is a $\frac{\ell}{\ell(k+\ell)^2}$-bipartite, expander. By the sampler lemma, \pref{lem:sampler-lemma},
    \[ 1 - \prob{T} = \prob{L \setminus T} \leq \frac{1}{\ell^2 (\frac{2}{3}\prob{C})^2}\prob{C} \leq \frac{9}{4\ell}.\]
    The statement follows for $\ell \geq 4$.
\end{enumerate}
\end{proof}

\section{The Grassmann Poset}
\label{sec:agreement-in-the-grassmann}
Finally, the fourth agreement testing theorem, \pref{thm:agreement-on-Grassmann} gives new agreement tests on the Grassmann poset. Such agreement tests are well studied in the PCP literature but other than \cite{ImpagliazzoKW2012} we are not aware of works that study the general question outside the context of Reed Muller codes. This part can be viewed as extending \cite{ImpagliazzoKW2012} to a broader parameter regime (our focus here is on the 99\% soundness whereas in \cite{ImpagliazzoKW2012} the focus was on 1\% soundness).

Let $\F$ be the finite field of size $q$, the Affine Grassmann Poset $X = \Graff{}{\F^n}{d}$ is the set of all \emph{affine} subspaces of dimension $\leq d$. We order the subspaces by containment, and denote by $X(k)$ all subspaces of dimension $k$.

Similarly, we can restrict ourselves to linear subspaces. We denote by $Y=\Grassmann{}{\F^n}{d}$, the set of all \emph{linear} subspaces $s \subset \F^n$ of dimension $\leq d+1$. We order the subspaces by containment, and the convention here is denoting by $Y(k)$ all subspaces of dimension exactly $k+1$.

\begin{definition}[The Grassmann $d,\ell$-distribution]\torestate{\label{def:grassmann-d-l-dist}
Let $\ell < d$.
We define the distribution $AGD_{d,\ell}$ on the Affine Grassmann Poset and a distribution $LGD_{d,\ell}$ on the Linear Grassman Poset, by the following random process:
\begin{enumerate}
    \item Sample $t \in X(\ell)$ (respectively in $Y(\ell)$).
    \item Sample $s_1,s_2 \in X(d)$ (respectively in $Y(d)$) given that $t \subset s_1,s_2$.
\end{enumerate}}
\end{definition}
The ground set in the Affine Grassmann agreement test is $V_{aff} = X(0)$, the set of points in $\mathbb{F}^n$. Our sets $S_{aff} = X(d)$ are the $d$-dimensional affine spaces.

In the Linear Grassmann agreement test, our ground set $V_{lin} = X(0)$ is the one-dimensional spaces. Our sets are 
\[ S_{lin} = \sett{[s] = \set{v \subset s}}{s \in X(d)}. \]
Namely, for each $d+1$-dimensional space $s \in X(d)$ the set $[s] \in S_{lin}$ is the collection of all the one-dimensional vectors paces that are contained in $s$.

We are ready to state our main theorem for Grassmann Posets:
\begin{theorem}[Agreement on the Affine Grassmann Poset]\torestate{
\label{thm:agreement-on-Grassmann-affine}
There exists a constant $c > 0$ such that for every prime power $q$, $r, \delta > 0$, and integers $\ell,d,n$ such that $3\ell+2 < d \leq n$ the following holds. The $d,\ell$-Grassmann agreement test on $X = \Graff{}{\F^n}{d}$ is $q^{-\ell} r \delta$-approximately $c \left (1 + \frac{1}{r} \right )$-sound for $\delta$-ensembles.}
\end{theorem}

\begin{theorem}[Agreement on the Linear Grassmann Poset]\torestate{
\label{thm:agreement-on-Grassmann}
There exists a constant $c > 0$ such that for every prime power $q$, $r, \delta > 0$, and integers $\ell,d,n$ such that $3\ell+2 < d \leq n$ the following holds. The $d,\ell$-Grassmann agreement test on $X = \Grassmann{}{\F^n}{d}$ is $q^{-\ell+1} r \delta$-approximately $c \left (1 + \frac{1}{r} \right )$-sound for $\delta$-ensembles.}
\end{theorem}

For the proofs of these theorems, we use the spectral gaps in the containment walk in the Grassmann, and the complement walk in the Grassmann. In particular

\begin{claim}[Folklore]
\label{claim:one-dim-containment-walk-is-a-good-expander}

\begin{enumerate}
    \item The following $0,k$-containment walk is a $\frac{1}{\sqrt{q^{k}}}$-bipartite expander in any the \emph{Affine Grassmann Poset} where $k \leq d$:
        \[ L = X(0), R = X(k),\]
    and $(v,a) \in E$ if $v \subset a$.
    \item The following $0,k$-containment walk is a $\frac{1}{\sqrt{q^{k}}}$-bipartite expander in any the \emph{Linear Grassmann Poset} or Linear Grassmann Poset where $k \leq d$:
        \[ L = X(0), R = X(k),\]
    and $(v,a) \in E$ if $v \subset a$.
\end{enumerate}
\end{claim}

We can define the complement walk for the Grassmann Posets as well. In the Affine Grassmann complement walk, we traverse from $w_1$ to $w_2$ if $dim(span(w_1,w_2)) = dim(w_1) + dim(w_2) + 1$. Here $span(w_1,w_2)$ is the smallest affine space that contains $w_1 \cup w_2$. 

In the Linear Grassmann complement walk, we we traverse from $w_1$ to $w_2$ if $dim(span(w_1,w_2)) = dim(w_1) + dim(w_2)$. Equivalently, if the intersection between $w_1$ and $w_2$ is trivial. 

It will be useful to examine these walks when we also condition on being independent with respect to a fixed subspace $u_0$. 

\begin{definition}[Conditioned Complement Walk in the Affine Grassmann Poset]\torestate{\label{def:cond-comp-walk-affine-grassmann}
Let $X = \Graff{}{F^n}{d}$, and let $\ell_1,\ell_2,\ell_3 \leq d$ so that $\ell_1+\ell_2+\ell_3+2 \leq n$. Fix some $u_0 \in X(\ell_3)$. The $u_0$-conditioned $\ell_1,\ell_2$-complement walk in $X$ is the walk where
\[L = \sett{v \in X(\ell_1)}{dim(span(u_0,v)) = \ell_1 + \ell_3 + 1},\]
\[R = \sett{w \in X(\ell_2)}{dim(span(u_0,w)) = \ell_2 + \ell_3 + 1},\]
\[ E = \sett{(v,w)}{v \in X(\ell_1), \; w \in X(\ell_2),\; dim(span(v,w,u_0)) = \ell_1+\ell_2+\ell_3 + 2}.\]
We choose an edge $(v,w)$ uniformly at random.}
\end{definition}

\begin{definition}[Conditioned Complement Walk in the Linear Grassmann Poset]\torestate{\label{def:cond-comp-walk-lin-grassmann}
Let $Y = \Grassmann{}{F^n}{d}$, and let $\ell_1,\ell_2,\ell_3 \leq d$ so that $\ell_1+\ell_2+\ell_3+3 \leq n$. Fix some $u_0 \in X(\ell_3)$. The $u_0$-conditioned $\ell_1,\ell_2$-complement walk in $X$ is the walk where
\[L = \sett{v \in X(\ell_1)}{u_0 \cap v = \set{0}},\]
\[R = \sett{w \in X(t_2)}{u_0 \cap w = \set{0}}\]
\[ E = \sett{(v,w)}{w \in X(\ell_1), \; v \in X(t_2),\; v \oplus w \oplus u_0 \in X(\ell_1+\ell_2+\ell_3)}.\]
Here $\oplus$ means direct sum. Requiring that the sum is direct, is equivalent to requiring the dimension of the sum, to be the sum of the dimensions of $v,w$ and $u_0$. We choose an edge $(v,w)$ uniformly at random.}
\end{definition}

\begin{claim}[Grassmann Complement Walk]
\torestate{\label{claim:grassmann-comp-walk}
\begin{enumerate}
    \item Let $X = \Grassmann{}{\mathbb{F}^n}{d}$ be an Affine Grassmann Poset. Let $\ell_1,\ell_2,\ell_3 \leq d$ so that $\ell_1 + \ell_2 + \ell_3 + 3 \leq n$. Fix some $u \in X(\ell_3)$. Then the $u$-conditioned $\ell_1,\ell_2$-complement walk in the Grassmann Poset is a $\frac{4}{q^{n-\ell_1-\ell_2-\ell_3 - 1}}$-bipartite expander.
    \item Let $Y = \Grassmann{}{\mathbb{F}^n}{d}$ be a Linear Grassmann Poset. Let $\ell_1,\ell_2,\ell_3 \leq d$ so that $\ell_1 + \ell_2 + \ell_3 + 3 \leq n$. Fix some $u \in X(\ell_3)$. Then the $u$-conditioned $\ell_1,\ell_2$-complement walk in the Grassmann Poset is a $\frac{4}{q^{n-\ell_1-\ell_2-\ell_3 - 2}}$-bipartite expander.
\end{enumerate}
}
\end{claim}
We prove this claim in \pref{sec:grassman-comp-walk}.

\begin{proof}[Proof of \pref{thm:agreement-on-Grassmann-affine}]
As in the proof of \pref{thm:main-agreement-theorem-two-sided}, we have an ensemble of functions $\FF$ that has $\disagr{AGR_{d,\ell}}(\FF) = \varepsilon$. We need to find a global function $G$ so that
\[ \Prob[s]{f_{[s]} \ane{r \delta q^{-\ell}} \rest{G}{s}} = \bigO{\left (1 + \frac{1}{r} \right )\varepsilon}. \]
We do so using \pref{thm:main-STAV-agreement-theorem}.

Consider the following STAV-structure
\[S = S_{aff} = X(d), \; T = X(\ell), \; A = X(\ell-1)\; V = X(0),\]
The distribution is choosing:
\begin{enumerate}
    \item $s \in X(d)$ uniformly at random.
    \item $t \in X(\ell)$ given that $t \subset s$.
    \item A pair $(a,v)$ given that $span(a,v) = t$.
\end{enumerate}

By \pref{lem:delta-surprise}, and the fact that our $T$-lower graph is the containment graph in the Grassmann, which is a $\bigO{\frac{1}{\sqrt{q^{-\ell}}}}$-bipartite expander, any $(\ell, \delta)$-distance ensemble has $\surp(\FF) = \bigO{\delta^{-1} q^{-\ell}}$.

Next, we are to show that the STAV-structure defined in the theorem is $\bigO{\gamma}$-good, for $\gamma = \frac{1}{q^{\ell}}$. Namely, that it fulfils the assumptions in \pref{def:good-stav-structure}:
\begin{enumerate}

    \item \pref{ass:global-A-V-graph}: The bipartite graph whose edges are the pairs $(a,v)$ in the Affine Grassmann Poset is the complement walk graph in the Affine Grassmann poset. By \pref{claim:grassmann-comp-walk}, this graph is a $\bigO{\frac{4}{q^{n-\ell-2}}} = \bigO{\frac{1}{q^{-\ell}}}$-bipartite expander.

    \item \pref{ass:edge-expander}: Note that for any $a \in X$, the collection of subspaces that contain $a$ are isomorphic to the Grassmann Poset of the quotient of $\mathbb{F}^n / a'$ where $a'$ is a linear subspace of dimension $\ell$. 
    Thus, the $STS_a$-graph is the two step version of the $0,d-\ell+1$-containment walk in the Linear Grassmann Poset. By \pref{claim:one-dim-containment-walk-is-a-good-expander}, this graph is a $\frac{1}{q^{d-\ell}}$-two-sided spectral expander. This is in particular a $\frac{1}{3}$-edge expander.

    \item \pref{ass:v-a-graph}: As in the simplicial complex case, once we condition on $(a,v)$, there is only one space in $t = span(a, v) \in T$ that contains both $a$ and $v$. Thus the graph in the assumption is a clique with self loops, and in particular has $\frac{1}{q^{\ell}}$-spectral expansion.

    \item \pref{ass:vasa-graphs} Consider the following $\VASA$-distribution:
    \begin{enumerate}
        \item Choose $s \in S$.
        \item Choose two $a_1,a_2 \subset s$ so that $dim(span(a_1,a_2)) = 2\ell+1$.
        \item Choose $v \in s$ so that $dim(span(v,a_1,a_2)) = 2\ell+2$.
        \item out put $(v,a_1,s,a_2)$ or $(v,a_2,s,a_1)$ with probability $\frac{1}{2}$.
    \end{enumerate}
    This distribution is symmetric in $a_1,a_2$, and its marginal is exactly the choice of $(v,a,s)$ in the STAV-structure above.

    \item \pref{ass:vasa-graphs-v}: Fix some $v \in V$. The $\vASA{v}$-graph in the Affine Grassmann Poset is the graph whose double cover is $v$-conditioned $\ell-1,\ell-1$-complement walk. By \pref{claim:grassmann-comp-walk}, this graph is a $\frac{4}{q^{n-2\ell-2}} = \bigO{\frac{1}{q^{\ell}}}$-two-sided spectral expander expander.

    \item \pref{ass:vasa-graphs-a}: Fix $a \in A$. The $\VAS{a}$-graph is the following graph:
    \[L = \sett{v \in V}{v \notin a}, \]
    \[R = \sett{(a',s)}{ a, a' \subset s \ve dim(span(a,a')) = 2\ell+1}, \]
    \[E = \sett{(v,(a',s))}{\set{v},a',a \subset s \ve dim(span(v,a,a')) = 2\ell+2}.\]
    The probability of the edges are uniform. We prove below that this graph is a $\sqrt{\frac{1}{q^\ell}}$-bipartite expander using \pref{lem:structure-trickling-lemma}.

    Consider the following $2$-dimensional $3$-partite simplicial complex:
    \begin{itemize}
        \item The parts of the complex are
        \[ Y[ 1] = \sett{v \in V}{v \notin a}, \]
        \[ Y[ 2] = \sett{a' \in a}{dim(span(a,a')) = 2\ell+1},\]
        \[ Y[ 3] = \sett{(a',s)}{ a',a \subset s \ve dim(span(a,a')) = 2\ell+1}.\]
        \item We connect $(a',v,(a'',s)) \in Y(2)$ if $a' = a''$, $\set{v},a',a \subset s \ve dim(span(v,a,a')) = 2\ell+2$ and $\set{v},a',a \subset s$. The probability of choosing some triangle $(a',v,(a'',s))$ in uniform, but we view it as the following: choosing $s$ given that $a \subset s$, and then choosing $a', v$ given that $\set{v},a',a \subset s \ve dim(span(v,a,a')) = 2\ell+2$.
    \end{itemize}
    Denote by $A^{i,j}$ the bipartite walks between $Y[ i]$ and $Y[ j]$.
    We notice the following:
    \begin{enumerate}
        \item $A^{1,3}$ is the bipartite operator of the bipartite walk between $L,R$ in the $\VAS{a}$-graph.
        \item $A^{1,2}$ is the $a$-conditioned $0,(\ell-1)$-complement walk in the Grassmann. It is a $\bigO{\frac{1}{q^{n-\ell-1}}}=\bigO{\frac{1}{q^\ell}}$-bipartite expander.

        \item Assume without loss of generality, that $a$ is linear. for every $a' \in Y\sqbrack{2}$, the link of $a'$ is the following bipartite graph:
        \[ L = \sett{v \in V}{dim(span(v,a,a')) = 2\ell+2 },\]
        \[ R \cong \sett{s \in S}{a,a' \subset s}.\]
        This walk is similar to the $0,d-2\ell$-containment walk in $Y' = \Grassmann{q}{\mathbb{F}^n / span(a,a')}{d}$ (but instead of a single line in $[v] \in X'$ we have a set of points $v_1,...,v_j$ whose projection to $\mathbb{F}^n / span(a,a')$ go in to the one dimensional space $[v]$). Hence this walk is a $\bigO{\frac{1}{\sqrt q^{d-2\ell-2}}}=\bigO{\frac{1}{\sqrt{q^\ell}}}$-bipartite expander.
    \end{enumerate}
    Hence we can apply \pref{lem:structure-trickling-lemma} and conclude that
    \[ \lambda(A^{1,3}) \leq \bigO{\frac{1}{\sqrt{q^\ell}}} + \bigO{\frac{1}{\sqrt{q^\ell}}} = \bigO{\frac{1}{\sqrt{q^\ell}}}.\]

    \item \pref{ass:a-not-large}: In the Grassmann, $\adj{a}$ are all the points $v \notin a$. For $d \geq \ell + 1$, the probability of choosing a point that is not contained in $a$ is $\approx 1 - \frac{1}{q^2} > \frac{1}{2}$.
\end{enumerate}

\bigskip

Thus by \pref{thm:main-STAV-agreement-theorem}, we are promised a function $G:X(0) \to \Sigma$ so that
\[\prob{f_s \ane{r q^{-\ell}\delta } \rest{G}{s}} = \bigO{\left ( 1 + \frac{1}{r} \right ) \varepsilon}.\]

\end{proof}

The proof in the linear case is very similar:
\begin{proof}[Proof of \pref{thm:agreement-on-Grassmann}]
As in the proof of \pref{thm:main-agreement-theorem-two-sided}, we have an ensemble of functions $\FF$ that has $\disagr{LGR_{d,\ell}}(\FF) = \varepsilon$. We need to find a global function $G$ so that
\[ \Prob[s]{f_{[s]} \ane{r \delta q^{-\ell+1}} \rest{G}{s}} = \bigO{\left (1 + \frac{1}{r} \right )\varepsilon}. \]
We do so using \pref{thm:main-STAV-agreement-theorem}.

Consider the following STAV-structure
\[S = S_{lin} \cong X(d), \; T = X(\ell), \; A = X(\ell-1)\; V = X(0),\]
where we abuse the notation and identify $s$ or $f_{[s]}$ with $s \in X(d)$ and $f_{s}$.
The distribution is choosing:
\begin{enumerate}
    \item $s \in X(d)$ uniformly at random.
    \item $t \in X(\ell)$ given that $t \subset s$.
    \item A pair $(a,v)$ given that $a \oplus v = t$ (here $\oplus$ means direct sum).
\end{enumerate}

By \pref{lem:delta-surprise}, and the fact that our $T$-lower graph is the containment graph in the Grassmann, which is a $\bigO{\frac{1}{\sqrt{q^{-\ell+1}}}}$-bipartite expander, any $(\ell, \delta)$-distance ensemble has $\surp(\FF) = \bigO{\delta^{-1} q^{-\ell+1}}$.

Next, we are to show that the STAV-structure defined in the theorem is $\bigO{\gamma}$-good, for $\gamma = \frac{1}{q^{\ell-1}}$. Namely, that it fulfils the assumptions in \pref{def:good-stav-structure}:
\begin{enumerate}

    \item \pref{ass:global-A-V-graph}: The bipartite graph whose edges are the pairs $(a,v)$ in the Grassmann Poset is the complement walk graph in the Grassmann poset. By \pref{claim:grassmann-comp-walk}, this graph is a $\bigO{\frac{\ell}{q^{n-\ell-2}}} = \bigO{\frac{1}{q^{\ell-1}}}$-bipartite expander.

    \item \pref{ass:edge-expander}: Consider the Grassmann Poset of the quotient $Y = \Grassmann{q}{\mathbb{F_q}^n / a}{d - \ell + 1}$. The $STS_a$-graph for the Grassmann, is (isomorphic to) the graph whose vertices are $Y(d-\ell+1)$, and where two subspaces $s_1,s_2$ share an edge if they intersect on a $1$-dimensional subspace. This graph is the two step version of the $0,d-\ell+1$-containment walk in the Grassmann. By \pref{claim:one-dim-containment-walk-is-a-good-expander}, this graph is a $\bigO{\frac{1}{q^{d-\ell}}}$-two-sided spectral expander. In particular it is a $\frac{1}{3}$-edge expander.

    \item \pref{ass:v-a-graph}: As in the simplicial complex case, once we condition on $(a,v)$, there is only one space in $t = a \oplus v \in T$ that contains both $a$ and $v$. Thus the graph in the assumption is a clique with self loops, and in particular has $\frac{1}{q^{\ell-1}}$-spectral expansion.

    \item \pref{ass:vasa-graphs} Consider the following $\VASA$-distribution:
    \begin{enumerate}
        \item Choose $s \in S$.
        \item Choose two $a_1,a_2 \subset s$ so that $a_1 \oplus a_2 \subset s$.
        \item Choose $v \subset s$ so that $v \cap (a_1 \oplus a_2) = \set{0}$.
        \item out put $(v,a_1,s,a_2)$ or $(v,a_2,s,a_1)$ with probability $\frac{1}{2}$.
    \end{enumerate}
    This distribution is symmetric in $a_1,a_2$, and its marginal is exactly the choice of $(v,a,s)$ in the STAV-structure above.

    \item \pref{ass:vasa-graphs-v}: Fix some $v \in V$. The $\vASA{v}$-graph in the Linear Grassmann Poset is the graph whose double cover is the $v$-conditioned $\ell-1,\ell-1$-complement walk. By \pref{claim:grassmann-comp-walk}, this graph is a $\bigO{\frac{1}{q^{n-2\ell-2}}} = \bigO{\frac{1}{q^{\ell-1}}}$-two sided expander.

    \item \pref{ass:vasa-graphs-a}: Fix $a \in A$. The $\VAS{a}$-graph is the following graph:
    \[L = \sett{v \in V}{a \cap v = \set{0}}, \]
    \[R = \sett{(a',s)}{ a \oplus a' \subset s}, \]
    \[E = \sett{(v,(a',s))}{v \oplus a' \oplus a \subset s}.\]
    The probability of the edges are uniform. We prove below that this graph is a $\sqrt{\frac{1}{q^{\ell-1}}}$-bipartite expander using \pref{lem:structure-trickling-lemma}.

    Consider the following $2$-dimensional $3$-partite simplicial complex:
    \begin{itemize}
        \item The parts of the complex are
        \[ Y[ 1] = \sett{v \in V}{a \cap v = \set{0}}, \]
        \[ Y[ 2] = \sett{a' \in a}{a \cap a' = \set{0}},\]
        \[ Y[ 3] = \sett{(a',s)}{ a \oplus a' \subset s}.\]
        \item We connect $(a',v,(a'',s)) \in Y(2)$ if $a' = a''$ and $v \oplus a' + a \subset s$. The probability of choosing some triangle $(a',v,(a'',s))$ in uniform, but we view it as the following: choosing $s$ given that $a \subset s$, and then choosing $a', v$ given that $a \oplus a' \oplus v \subset s$.
    \end{itemize}
    Denote by $A^{i,j}$ the bipartite walks between $Y[ i]$ and $Y[ j]$.
    We notice the following:
    \begin{enumerate}
        \item $A^{1,3}$ is the bipartite operator of the bipartite walk between $L,R$ in the $\VAS{a}$-graph.
        \item $A^{1,2}$ is the $a$-conditioned $0,(\ell-1)$-complement walk in the Grassmann. It is a $\bigO{\frac{1}{q^{n-\ell}-1}}=\bigO{\frac{1}{q^{\ell-1}}}$-bipartite expander.

        \item for every $a' \in Y\sqbrack{2}$, the link of $a'$ is the following bipartite graph:
        \[ L = \sett{v \in V}{a \oplus a' \subset s},\]
        \[ R \cong \sett{s \in S}{a \oplus a' \subset s}.\]
        This walk is similar to the $0,d-2\ell$-containment walk in $X' = \Grassmann{q}{\mathbb{F}^n / (a\oplus a')}{d}$ (but instead of a single vertex in $[v] \in X'$ we have a set of vertices $v_1,...,v_j$ whose projection to $\mathbb{F}^n / (a\oplus a')$ is $[v]$). Hence this walk is a $\bigO{\frac{1}{\sqrt q^{d-2\ell-2}}}=\bigO{\frac{1}{\sqrt{q^{\ell-1}}}}$-bipartite expander.
    \end{enumerate}
    Hence we can apply \pref{lem:structure-trickling-lemma} and conclude that
    \[ \lambda(A^{1,3}) \leq \bigO{\frac{1}{\sqrt{q^{\ell-1}}}} + \bigO{\frac{1}{\sqrt{q^{\ell-1}}}} = \bigO{\frac{1}{\sqrt{q^{\ell-1}}}}.\]

    \item \pref{ass:a-not-large}: In the Grassmann, $\adj{a}$ are all the subspaces $v$ so that they are not contain in $a$. For $d \geq \ell + 1$, the probability of choosing a subspace that is not contained in $a$ is $\approx 1 - \frac{1}{q^2} > \frac{1}{2}$.
\end{enumerate}

\bigskip

Thus by \pref{thm:main-STAV-agreement-theorem}, we are promised a function $G:X(0) \to \Sigma$ so that
\[\prob{f_s \ane{r q^{-\ell+1}\delta } \rest{G}{s}} = \bigO{\left ( 1 + \frac{1}{r} \right ) \varepsilon}.\]

\end{proof}

\section{The Complement Random Walk}
\label{sec:complement-walk}
This section is dedicated the so-called complement random walk, as described in \pref{def:complement-walk}, and which we repeat now for ease of reading:
\restatedefinition{def:complement-walk}
%Let $X$ be a $d$-dimensional simplicial complex and let $k+\ell+1\leq d$. The complement walk, is a walk between $X(k)$ and $X(\ell)$, where we go from $s \in X(k)$ to $t \in X(\ell)$ by $t \dunion s \in X(k+\ell+1)$. We associate with this walk a bipartite graph whose vertices are $X(k)$ and $X(\ell)$ and the probability of choosing an edge $(s,t)$ is the probability of choosing $s \dunion t \in X(k+\ell+1)$ and then choosing $s \in X(k)$, given that we chose $s \dunion t$.
%
\begin{theorem}
\torestate{
\label{thm:complement-walk-is-a-good-expander}
\begin{enumerate}
    \item Let $X$ be a $\lambda$ two-sided $d$-dimensional link-expander. Let $\ell_1,\ell_2$ integers so that $\ell_1 + \ell_2 +1 \leq d$. Denote by $\comp{\ell_1,\ell_2}$, the bipartite operator of the $\ell_1,\ell_2$-complement walk. Then
            \[\lambda(\comp{\ell_1,\ell_2}) \leq (\ell_1+1)(\ell_2+1)\lambda.\]
    \item Let $X$ be a $d+1$-partite $\frac{\lambda}{(d+1)\lambda + 1}$-one-sided link expander, where $\lambda < \frac{1}{2}$. Let $I,J \subset [d]$ be two disjoint colors. Denote by $\comp{I,J}$ the $I,J$-colored walk. Then
            \[\lambda(\comp{I,J}) \leq |I| |J| \lambda.\]

\end{enumerate}
}
\end{theorem}

In \pref{sec:complement-walk-prelims} we give some additional definitions and preliminaries for this section. In \pref{sec:two-sided-complement-walk} we prove the two-sided complement walk's expansion, and the $d+1$-partite colored walk's expansion respectively. In \pref{sec:grassman-comp-walk} we extend the result to Grasmann Posets. Finally in \pref{sec:random-walks-fixed-union-size} we give additional applications of the complement walk: we analyze random walks on high dimensional expanders with a fixed intersection size, and in \pref{sec:applications-of-complement-walks} we prove a high dimensional expander mixing lemma.

\subsection{Preliminaries for this Section}
\label{sec:complement-walk-prelims}
In a finite measured space we have an inner product on the space of real functions. Thus for any $f,g \in \ell_2(X(k))$

    \[\iprod{f,g} = \Ex[s \in X(k)]{f(s)g(s)}.\]

    In addition, we define two sets of operators that connect the different levels of functions by averaging.

    \begin{definition}[Up and Down Operators]\torestate{\label{def:up-down-operators}
    Define the up operator $\Up_{k,k+1} : \CX{k} \to \CX{k+1}$, and the down operator $\Down_{k+1,k} \colon \CX{k+1} \to \CX{k}$, by
    \[\Up_{k,k+1} f(s) = \Ex[t \subset s; \, t \in X(k)]{f(t)},\]
    \[\Down_{k+1,k} g(t) = \Ex[s \supset t; \, s \in X(k+1)]{g(s)}.\]}
    \end{definition}
    One can show that $\Down_{k+1} = (\Up_{k})^*$, the adjoint with respect to the inner product above.

    Recall the $k+1,k$-lower walk defined in \pref{sec:agreement-on-hdx}. $D_{k+1,k}$ is it's bipartite operator.

    %Note that these operators are the bipartite averaging operators of the graph where one side is $X(k+1)$, the other is $X(k)$, and the edges denote containment.

\subsubsection{Localization}

    Given a function in $f:X(k) \to \RR$ there are two natural operations that give us a function in the link.
    \begin{definition}[Localization]\torestate{\label{def:localization}
        Let $\ell \leq k$ be two integers, $f \in C^k(X)$ and $s \in X(\ell)$. The localization of $f$ denoted by $\scloc{f}{s} : X_{s}(k-|\sigma|) \to \RR$, is defined by:
        \[\scloc{f}{s}(t) = f(s \dunion t).\]}
    \end{definition}

    \begin{definition}[Restriction]\torestate{\label{def:restriction}
        Let $\ell, k$ be two integers s.t. $\ell+k+1 \leq d$, $f \in C^k(X)$ and $s \in X(\ell)$. The restriction of $f$ denoted by $\scres{f}{s} : X_{s}(k) \to \RR$, is defined by:
        \[\scres{f}{s}(t) = f(t).\]}
    \end{definition}

\subsection{Proving the Complement Walk Theorem}
\label{sec:two-sided-complement-walk}

First we prove \pref{thm:complement-walk-is-a-good-expander}. Our main technical tools is \pref{lem:structure-trickling-lemma}, which was already stated in \pref{sec:agreement-on-hdx}. We restate it here:

\restatelemma{lem:structure-trickling-lemma}

\begin{proof}[Proof of \pref{thm:complement-walk-is-a-good-expander}, item 1]
We begin with the two sided case, and prove the statement by induction on $\ell_1 + \ell_2 = k$. The base case is $\ell_1 + \ell_2 = 0$, i.e. $\ell_1=\ell_2=0$. This is exactly the assumption that $X$ is a $\lambda$-two sided link expander.

Assume the statement is true for any $\ell_1,\ell_2$ s.t. $\ell_1 + \ell_2 \leq k$, and consider the graph operator of the complement walk graph $\comp{\ell_1,\ell_2+1}:\RR^{X(\ell_1)} \to \RR^{X(\ell_2+1)}$, for some $\ell_1,\ell_2$ s.t. $(\ell_1+1)+\ell_2 = k+1$. We need to prove that
\[ \lambda(\comp{\ell_1+1,\ell_2})  \leq (\ell_1+2)(\ell_2+1)\lambda,\]

Note that it is enough to prove for the case where we take $\ell_1+1$ since the adjoint of $\comp{\ell_1+1,\ell_2}$ is $\comp{\ell_1,\ell_2+1}$.
It might be easy to keep in mind the first non-trivial case where $\ell_1+1=1$ and $\ell_2=0$.

Consider the following $2$-dimensional $3$-partite simplicial complex $Y$:
\begin{itemize}
    \item The vertices are $Y \sqbrack{1} = X(\ell_1), Y \sqbrack{2} = X(\ell_1 + 1), Y \sqbrack{3} = X(\ell_2)$.
    \item We connect $(y_1,y_2,y_3) \in Y(2)$ if $y_1 \subset y_2$ and $y_2 \dunion y_3 \in X(2)$. The probability of choosing some $(y_1,y_2,y_3)$ is the probability of choosing the edge $y_2 \dunion y_3$ and then choosing $y_1 \subset y_2$. In other words,
    \[ \Prob[Y]{(y_1,y_2,y_3)} = \Prob[X(\ell_1 + \ell_2)]{y_2 \dunion y_3}\cProb{X}{y_2}{y_2 \dunion y_3}\cProb{X}{y_1}{y_2}. \]
\end{itemize}
We notice the following:
\begin{enumerate}
    \item $\comp{2,3}$ is the bipartite operator of the bipartite walk between $X(\ell_1+1), X(\ell_2)$.
    \item $\comp{1,3}$ is the bipartite operator of the bipartite walk between $X(\ell_1), X(\ell_2)$. By induction $\lambda(\comp{1,3}) \leq (\ell_1+1)(\ell_2+1) \lambda$.
    \item for every $s \in Y\sqbrack{1}$, the bipartite operator of the link of $s$ is the complement walk for $\ell_1' = 0, \ell_2' = \ell_2$ in the link of $s$. as $\ell_1' + \ell_2' < (\ell_1 + 1) + \ell_2$, we may use the induction assumption to conclude that $\lambda(M_s) \leq (\ell_2 + 1) \lambda$.
\end{enumerate}
Hence we can apply \pref{lem:structure-trickling-lemma} and conclude that
\[ \lambda(\comp{2,3}) \leq (\ell_2 + 1) \lambda + (\ell_1+1)(\ell_2+1) \lambda \norm{\comp{1,2}} \leq (\ell_2+1)(\ell_1+2) \lambda.\]
\end{proof}

Towards proving the second item in \pref{thm:complement-walk-is-a-good-expander}, we need the following lemma, that we shall prove in \pref{sec:partite-trickling-down}:

\begin{lemma}
\label{lem:colorful-trickling-down-lemma}
Let $X$ be a $(d+1)$-partite simplicial complex, and suppose that for all $v \in X(0)$ the underlying graph is a $\lambda$-one sided $d$-partite expander, for $\lambda < \frac{1}{2}$. Suppose that the underlying graph of $X$ is connected. Then for every $\set{i}, \set{j} \subset [d+1]$, the bipartite graph between $X[i],X[j]$ is a $\frac{\lambda}{1-\lambda}$-bipartite expander.
\end{lemma}

The links of $s \in X(d-2)$ in a $d$-partite simplicial complexes are bipartite graphs. Thus by iterating this lemma we get the following corollary:
\begin{corollary}
\label{cor:colorful-trickling-corollary}
Let $\lambda < \frac{1}{2}$. Let $X$ be a simplicial complex s.t. every link of $X$ is connected and that for every $s \in X(d-2)$, $X_s$ is a $\frac{\lambda}{(d-1)\lambda+1}$-bipartite expanders. Then for every two colors $\set{i},\set{j}$, and every $s \in X$ s.t. $i,j \notin col(s)$ the graph between the two colors $\comp{\set{i},\set{j}}_s$ is a $\lambda$-bipartite expander.
\end{corollary}

\begin{proof}[Proof of \pref{thm:complement-walk-is-a-good-expander}, item 2]
The proof of the colored version is similar to the two-sided case, as is done by induction on $k := |I_1| + |I_2|$. The base case is where $|I_1| + |I_2| = 2$, i.e. $|I_1| = |I_2| = 1$. This case is true due to \pref{cor:colorful-trickling-corollary}.

Take some disjoint color sets $I_1,I_2$ s.t. $|I_1| + |I_2| = k+1$, and suppose the wlog $I_1 = J \dunion \set{i}$ where $J$ is non-empty.

Consider the following $2$-dimensional $3$-partite simplicial complex $Y$:
\begin{itemize}
    \item The vertices are $Y[1] = X[I_1], Y[2] = X[J], Y[3] = X[I_2]$.
    \item We connect $(y_1,y_2,y_3) \in Y(2)$ if $y_1 \subset y_2$ and $y_2 \dunion y_3 \in X[I_1 \dunion I_2]$. The probability of choosing some $(y_1,y_2,y_3)$ is the probability of choosing the edge $y_2 \dunion y_3$ and then choosing $y_1 \subset y_2$. In other words,
    \[ \Prob[Y]{(y_1,y_2,y_3)} = \Prob[X(\ell_1 + \ell_2)]{y_2 \dunion y_3}\cProb{X}{y_2}{y_2 \dunion y_3}\cProb{X}{y_1}{y_2}. \]
\end{itemize}
We notice the following:
\begin{enumerate}
    \item $\comp{2,3}$ is the bipartite operator of the bipartite walk between $X[I_1], X[I_2]$.
    \item $\comp{1,3}$ is the bipartite operator of the bipartite walk between $X[J], X[I_2]$. By induction $\lambda(\comp{1,3}) \leq \abs{J} \abs{I_2} \lambda$.
    \item for every $s \in Y\sqbrack{1}$, the bipartite operator of the link of $s$ is the complement walk for $\set{i}, I_2$ in the link of $s$. as $\abs{\set{i}} + \abs{I_2} < \abs{I_1} + \abs{I_2}$, we may use the induction assumption to conclude that $\lambda(M_s) \leq \abs{I_2} \lambda$.
\end{enumerate}
Hence we can apply \pref{lem:structure-trickling-lemma} and conclude that
\[ \lambda(\comp{2,3}) \leq \abs{I_2} \lambda + \lambda \abs{J} \abs{I_2}  \norm{\comp{1,2}} \leq \abs{I_1} \abs{I_2} \lambda.\]
\end{proof}

\subsubsection{Proof of \pref{lem:structure-trickling-lemma}}

\begin{proof}[Proof of \pref{lem:structure-trickling-lemma}]
Consider two functions $f:X[2] \to \RR$, $g: X[3] \to \RR$ s.t. $f,g$ are orthogonal to the space of constant functions, and s.t. $\norm{f} = \norm{g} = 1$. We need to prove that $\iprod{A^{2,3} f,g} \leq \eta + \lambda(A^{1,2}) \lambda(A^{2,3})$.

The following claim allows us to calculate the inner product in a simplicial complex locally.
\begin{claim}[Localization]
\label{claim:localization-lemma}
Let $X$ be a $d+1$-partite complex, $I_1, I_2, I_3$ disjoint colors. Then for any $f: X \sqbrack{I_1} \to \RR, g: X \sqbrack{I_2} \to \RR$
\[ \iprod{\comp{I_1,I_2} f, g} = \Ex[r \in I_3]{\iprod{M_s^{I_1,I_2} \scres{f}{r}, \scres{g}{r}}},\]
where $M_s^{I_1,I_2}$ is the bipartite operator for $I_1,I_2$ in the link of $s$.
\end{claim}

For every $v \in X[1]$, we denote it's bipartite operator by $A_v$. By \pref{claim:localization-lemma}
\[ \iprod{A^{2,3} f,g} = \Ex[v \in X\sqbrack{1} ]{\iprod{A_v \scres{f}{v}, \scres{g}{v}}}. \]
We decompose $\scres{f}{v} = \scres{f}{v}^0 + \scres{f}{v}^\perp$ where $\scres{f}{v}^0$ is constant and $\scres{f}{v}^\perp$ is orthogonal to $\scres{f}{v}^0$, and similarly $\scres{g}{v} = \scres{g}{v}^0 + \scres{g}{v}^\perp$. Note that $A_v \scres{f}{v}^0$ is also constant and $A_v \scres{f}{v}^\perp$ is also orthogonal to the constant part, because $A_v$ is an averaging operator. Thus
\begin{equation} \label{eq:structure-trickling-lemma-main}
\Ex[v \in X\sqbrack{1} ]{\iprod{A_v \scres{f}{v}, \scres{g}{v}}} = \Ex[v \in X\sqbrack{1} ]{\iprod{A_v \scres{f}{v}^0, \scres{g}{v}^0}} + \Ex[v \in X\sqbrack{1} ]{\iprod{A_v \scres{f}{v}^\perp, \scres{g}{v}^\perp}}.
\end{equation}
We bound each part in the righthand side of \eqref{eq:structure-trickling-lemma-main} separately.

\begin{itemize}
        \item From Cauchy-Schwartz:
        \[ \Ex[v \in X\sqbrack{1}]{\iprod{A_v \scres{f}{v}^\perp,\scres{g}{v}^\perp}} \leq \Ex[v \in X\sqbrack{1}]{\lambda(A_v) \norm{\scres{f}{v}^\perp} \cdot \norm{\scres{g}{v}^\perp}}.\]
        From the assumption for every $v \in X\sqbrack{1}$,  $\lambda(M_v) \leq \eta$, thus:
        \begin{align*}
        \Ex[v \in X(0)]{\lambda(A_v) \norm{\scres{f}{v}^\perp} \cdot \norm{\scres{g}{v}^\perp}} &\leq&
        \eta \Ex[v \in X\sqbrack{1}]{\norm{\scres{f}{v}^\perp} \cdot \norm{\scres{g}{v}^\perp]}} \\
        &\leq& \eta \Ex[v \in X\sqbrack{1}]{\frac{1}{2}(\norm{\scres{f}{v}^\perp}^2 + \norm{\scres{g}{v}^\perp}^2)}
        \\ & \leq& \eta,
        \end{align*}
        where the second inequality is achieved by taking arithmetic mean instead of geometric mean.

        \item Next we bound $\Ex[v \in X\sqbrack{1}]{\iprod{A_v \scres{f}{v}^0,\scres{g}{v}^0}}$. Notice that
        \[\scres{f}{v}^0 \equiv \Ex[u \in X_v \sqbrack{2}]{\scres{f}{v}(u)} = A^{1,2}f(v), \]
        \[\scres{g}{v}^0 \equiv \Ex[u \in X_v \sqbrack{3}]{\scres{g}{v}(u)}  = A^{1,3}g(v). \]
        Hence
        \[\Ex[v \in X\sqbrack{1}]{\iprod{A_v \scres{f}{v}^0,\scres{g}{v}^0}} = \Ex[v \in X\sqbrack{1}]{\iprod{A^{1,2} f(v) A^{1,3}g(v)}} = \iprod{A^{1,2} f(v), A^{1,3}g(v)}.\]
        From Cauchy-Schwarz
        \[ \iprod{A^{1,2} f(v), A^{1,3}g(v)} \leq \lambda(A^{1,2}) \lambda(A^{2,3})\norm{f} \norm{g} = \lambda(A^{1,2}) \lambda(A^{1,3}).\]
    \end{itemize}
    Summing up the two terms, we get that the operator is bounded by $\eta + \lambda(A^{1,2}) \lambda(A^{2,3})$.
\end{proof}

\begin{proof}[Proof of \pref{claim:localization-lemma}]
\[ \iprod{\comp{I_1,I_2} f,g} = \Ex[t \in X \sqbrack{I_1}]{g(t)\Ex[s \in X_t \sqbrack{I_2}]{f(s)}} =\] \[\Ex[s \dunion t \in X \sqbrack{I_1 \dunion I_2}]{f(s)g(t)}, \]
where the last expectation is by choosing two faces according to the random walk defined using $\comp{I_1, I_2}$. We condition on choosing some $r \in X \sqbrack{I_3}$:
\[ = \Ex[r \in X \sqbrack{I_3}]{\Ex[s \dunion t \in X_r \sqbrack{I_1 \dunion I_2}]{f(s)g(t)}}\]
\[ = \Ex[r \in X \sqbrack{I_3}]{\Ex[s \dunion t \in X_r \sqbrack{I_1 \dunion I_2}]{\scres{f}{r}(s)\scres{g}{r}(t)}}, \]
Following the previous steps in every link we conclude:
\[ = \Ex[r \in X \sqbrack{I_3}]{\iprod{\comp{I_1,I_2}_{r} \scres{f}{r}, \scres{g}{r}}}.\]

\end{proof}

\subsubsection{Partite trickling down lemma}
\label{sec:partite-trickling-down}
We now go towards proving \pref{lem:colorful-trickling-down-lemma}, since its corollary, \pref{cor:colorful-trickling-corollary} is the base case for proving \pref{thm:complement-walk-is-a-good-expander}, item 2. This lemma is an adaptation of the theorem in \cite{Oppenheim2018}, where the author proved the following:
\begin{theorem}[Theorem 5.2 in \cite{Oppenheim2018}]
\label{thm:Oppenheim-trickling-down-lemma}
Let $X$ be simplicial complex,  Let $-1 < k \leq d-2$ be some integer. For any $s \in X$, denote by $\lambda(X_s)$ the second largest eigenvalue of the underlygraph of $X_s$, in absolute value.

If for all $s \in X(k)$, $\lambda(X_s) \leq \lambda$, for some $\lambda \in (0,\frac{1}{2}]$, then for any $r \in X(k-1)$, s.t. $X_r$'s underlying graph is connected, $\lambda_2(X_r) \leq \frac{\lambda}{1-\lambda}$.
\end{theorem}

We begin by giving another version of the localization claim:
\begin{claim}[second localization lemmata]
\label{claim:partite-localization-lemma}
Let $X$ be any $d+1$-partite simplicial complex, and let $I_1,I_2$ be disjoint color sets, and $I_3 \subsetneq I_1$. Let $f \in \RR^{X\sqbrack{I_1}}, g \in \RR^{X\sqbrack{I_2}}$. Then $\iprod{\comp{I_1,I_2} f,g} = \Ex[r \in X\sqbrack{I_3}]{\iprod{\comp{I_1 \setminus I_3,I_2}_{r} \scloc{f}{r},\scres{g}{r}}}.$
Where $\comp{I_1 \setminus I_3,I_2}_{r}$ is the colored complement walk in the link of $r$.
\end{claim}

The proof is similar to the proof of \pref{claim:localization-lemma} and is therefore omitted.

\begin{proof}[Proof of  \pref{lem:colorful-trickling-down-lemma}]
Fix two colors $i,j$, s.t. $\lambda^{\set{i},\set{j}}$ is maximal, and fix some $k \ne i,j$. Take two functions $f: X\sqbrack{i} \to \RR, g: X\sqbrack{j} \to \RR$, s.t. $\ex{f}=\ex{g}=0$, $||f||=||g||=1$, and that $||\comp{\set{i},\set{j}}||=\iprod{\comp{\set{i},\set{j}} f,g}$.

For every $v \in X[k]$ we decompose $\scloc{f}{v},\scres{g}{v}$ to their constant part and the part that is perpendicular to constant functions:
\[ \scloc{f}{v} = (\scloc{f}{v})^0 + (\scloc{f}{v})^\perp ; \; \scloc{g}{v} = \scres{g}{v}^0 + \scres{g}{v}^\perp. \]
Thus by \cref{claim:partite-localization-lemma}:
\begin{align*}
\iprod{\comp{\set{i},\set{j}} f,g} = \Ex[v \in X\sqbrack{k}]{\iprod{\comp{\set{i},\set{j}}_v (\scloc{f}{v}), \scloc{g}{v}}} &= \\
\Ex[v \in X\sqbrack{k}]{\iprod{\comp{\set{i},\set{j}}_v (\scloc{f}{v})^0, \scres{g}{v}^0} + \iprod{\comp{\set{i},\set{j}}_v (\scloc{f}{v})^\perp, \scres{g}{v}^\perp}} &\leq \\
\Ex[v \in X\sqbrack{k}]{\iprod{\comp{\set{i},\set{j}}_v (\scloc{f}{v})^0, \scres{g}{v}^0} + \lambda \frac{||(\scloc{f}{v})^\perp||^2+||\scres{g}{v}^\perp||^2}{2}} &= \\
\Ex[v \in X\sqbrack{k}]{\iprod{\comp{\set{i},\set{j}}_v (\scloc{f}{v})^0, \scres{g}{v}^0} + \lambda \frac{||\scloc{f}{v}||^2+||\scres{g}{v}||^2}{2} - \lambda \frac{||(\scloc{f}{v})^0||^2+||\scres{g}{v}^0||^2}{2}} &\leq \\
\Ex[v \in X\sqbrack{k}]{\iprod{\comp{\set{i},\set{j}}_v (\scloc{f}{v})^0, \scres{g}{v}^0} + \lambda \frac{||\scloc{f}{v}||^2+||\scloc{g}{v}||^2}{2} - \lambda ||(\scloc{f}{v})^0|| ||\scres{g}{v}^0||}  &\leq \\
(1-\lambda) \Ex[v \in X\sqbrack{k}]{\iprod{\comp{\set{i},\set{j}}_v (\scloc{f}{v})^0, \scres{g}{v}^0}} + \lambda
\end{align*}

The last inequality is by Cauchy-Schwartz.

Notice that the average value that is in all the entries of $(\scloc{f}{v})^0$, is exactly $\comp{\set{i},\set{k}}f(v)$, and similarly $\scres{g}{v}^0$'s entries are $ \comp{\set{j},\set{k}} g(v)$ hence the above is equal to:
\[(1-\lambda)\iprod{\comp{\set{i},\set{k}}f, \comp{\set{j},\set{k}} g} + \lambda  \leq (1-\lambda) ||\comp{\set{i},\set{k}}|| ||\comp{\set{j},\set{k}}|| + \lambda,\]
and since $||\comp{\set{i},\set{j}}||$ is maximal:
\[ \leq (1-\lambda)||\comp{\set{i},\set{j}}||^2 + \lambda.\]
The inequality \[||\comp{\set{i},\set{j}}|| \leq (1-\lambda)||\comp{\set{i},\set{j}}||^2 + \lambda\] indicates that $||\comp{\set{i},\set{j}}|| \geq 1$ or $||\comp{\set{i},\set{j}}|| \leq \frac{\lambda}{1-\lambda}$.
If we show that the walk is connected, then as an immediate conclusion $||\comp{\set{i},\set{j}}|| \leq \frac{\lambda}{1-\lambda}$. We separate the proof that the walk is connected to the following claim:
\begin{claim}
\label{claim:basic-connectivity}
Let $X$ be a $d$-partite simplicial complex s.t. every link of $X$ is connected. Then for every $i,j \in \set{1,...,d}$, the induced graph between vertices of color $i$ and vertices of color $j$ is connected.
\end{claim}

Modulo this claim, the lemma follows.
\end{proof}

\begin{proof}[Proof of \pref{claim:basic-connectivity}]
We prove this by induction on $d$ - the number of parts. The base case of two parts is clear. Assume for $d$ parts and prove for $d+1$ parts:

Take some $v \in X[i], u \in X[j]$, as we already assumed that the whole complex is connected there is a walk $v=w_0,w_1,...,w_t,w_{t+1}=u$. We prove now that if $w_{q} \in X[i] \dunion X[j]$ and $w_{q+1} \notin X[i] \dunion X[j]$ we can substitute it with a walk from $w_{q}$ to $w_{q+2}$, where all the vertices except maybe $w_{q+2}$ are in $w_{q} \in X[i] \dunion X[j]$.

Each edge $\set{w_{q+1},w_{q+2}}$ is contained in some $d+1$-face $s \in X(d)$. We denote by $w^i_{q+1}, w^j_{q+1}$ the vertices in $s$ that are in $X[i], X[j]$ respectively.

Assume without loss of generality that $w_q \in X[I]$. The link of $w_{q+1}$, is a $d$-partite complex. By the induction hypothesis it is color connected, i.e. there is a walk between any two vertices from colors $i,j$ in the link. Specifically we can walk from $w_{q}$ to $w^j_{q+1}$. Also, as $w^j_{q+1}$ and $w_{q+2}$ share a a $d$-face, they also share an edge. Thus the walk between $w_q$ to $w^j_{q+1}$ and the edge $\set{w^j_{q+1},w_{q+2}}$ is the walk between $w_{q}$ and $w_{q+2}$ where all vertices except (maybe) $w_{q+2}$ are in $X[I] \dunion X[J]$.
\end{proof}

\subsection{Complement Walk for the Grassmann }
\label{sec:grassman-comp-walk}
In this subsection, we prove that the complement walk in the Grassmann Poset has good spectral gap, as stated in \pref{claim:grassmann-comp-walk}. We feel that the notion of complement walks could be generalized to many other Posets, however in this paper we merely study the complement walk of the Grassmann Poset.

\restateclaim{claim:grassmann-comp-walk}

\begin{proof}[Proof of the Affine Case]
Let $u \subset U$ be of dimension $\ell_3$. If we denote by $A$ the bipartite operator of the $\ell_1,\ell_2$-affine-complement walk and by $J$ the bipartite operator of just choosing $w_1,w_2$ independently. Denote by $E$ the event that $dim(span(w_1,w_2,u)) = \ell_1 + \ell_2 + \ell_3 + 2$. We can say that
\[ eA = J - (1-e)M. \]
where $e$ is the probability of choosing $w_1,w_2$ independently so that $E$ occurs and $M$ is the operator conditioned that $E$ doesn't occur. Since the spectral norm of $J$ is $0$ when we restrict to the space of functions with expectation $0$, we obtain that
\[\norm{A} \leq \frac{1-e}{e}\norm{M} \leq \frac{1-e}{e}. \]

We calculate a lower bound on $e$. Consider the following process where we choose $\ell_1 + \ell_2$ points $(p_1,...,p_{\ell_1+\ell_2})$ sequentially so that the first $\ell_1$ points span $w_1$, and the other $\ell_2$ span $w_2$. If we choose these points so that in $j$-step $p_{j} \notin span(U,p_1,...,p_{j-1})$, then $E$ occurs.

For every $j$, if we chose $p_1,...p_{j-1}$ so that $span(U,p_1,...,p_{j-1})$ is of maximal dimension, then the probability to choose $p_j \in span(U,p_1,...,p_{j-1})$ is $\frac{q^{\ell_3+j-1}}{q^n} = \frac{1}{q^{n-\ell_3-j+1}}$.

By union bound, we get that the probability that
\[ e \geq 1 - \sum_{j=0}^{\ell_1+\ell_2}\frac{1}{q^{n-\ell_3-j+1}}, \]
Rearranging and taking to infinity the geometric sum, we get that this is greater or equal to
\[ e \geq 1- \frac{1}{q^{n-\ell_3-\ell_2-\ell_1-1}}\frac{q}{q-1} \geq 1-\frac{2}{q^{n-\ell_3-\ell_2-\ell_1-1}}.\]
Hence we get that the expansion is bounded by $\frac{4}{q^{n-\ell_1-\ell_2-\ell_3-1}}$.

\end{proof}

\begin{proof}[Proof of the Linear Case]
Similar to the affine case, let $u \subset U$ be of dimension $\ell_3$. We denote by $A$ the bipartite operator of the $\ell_1,\ell_2$-affine-complement walk and by $J$ the bipartite operator of just choosing $w_1,w_2$ independently. Denote by $E$ the event that $dim(w_1 \oplus w_2 \oplus u)) = \ell_1 + \ell_2 + \ell_3 + 1$. And as before we obtain that
\[\norm{A} \leq \frac{1-e}{e}. \]

where $e$ is the probability of choosing $w_1,w_2$ independently so that $E$ occurs.

We calculate a bound lower on $e$. Consider the following process where we choose $\ell_1 + \ell_2$ lines $(r_1,...,r_{\ell_1+\ell_2})$ sequentially so that the first $\ell_1$ lines span $w_1$, and the other $\ell_2$ span $w_2$. If we choose these lines so that in $j$-step $p_{j} \notin span(U,r_1,...,r_{j-1})$, then $E$ occurs.

For every $j$, if we chose $r_1,...r_{j-1}$ so that $span(U,r_1,...,r_{j-1})$ is of maximal dimension, then the probability to choose $r_j \in span(U,r_1,...,r_{j-1})$ is $\frac{q^{\ell_3+j}-1}{q^n - 1} \leq \frac{1}{q^{n-\ell_3-j}}$.

Similarly to the previous case, by union bound, we get that the probability that
\[ e \geq 1 - \sum_{j=1}^{\ell_1+\ell_2}\frac{1}{q^{n-\ell_3-j}}, \]
Rearranging and taking to infinity the geometric sum, we get that this is greater or equal to
\[ e \geq 1- \frac{1}{q^{n-\ell_3-\ell_2-\ell_1-2}}\frac{q}{1-q} \geq 1-\frac{2}{q^{n-\ell_3-\ell_2-\ell_1-2}}.\]
Hence we get that the expansion is bounded by $\frac{4}{q^{n-\ell_1-\ell_2-\ell_3-2}}$.
\end{proof}

\rem{\begin{proof}[Proof of \pref{claim:grassmann-comp-walk}]
First we reduce to the case where $u = \set{0}$:
Let $u \in X(\ell_3)$ for $\ell_3 > -1$. Choosing $(a_1,s,a_2)$ so that $a_1 \oplus a_2 \oplus u \subset s$ is similar to choosing $(w_1,w_2)$ in the $\ell_1,\ell_2$-complement walk for $\Grassmann{q}{\mathbb{F}^n / u}{d}$, and then choosing \emph{independently} $a_1 \in X(\ell_1), a_2 \in X(\ell_2)$ so that $\pi(a_i) = w_i$ for $i=1,2$, where $\pi$ is the projection $\mathbb{F}^n \mapsto \mathbb{F}^n/u$. As the choice is done independently, the spectral gap is the same as in the $\ell_1,\ell_2$-complement walk in $\Grassmann{q}{\mathbb{F}^n / u}{d}$.

\bigskip

Next, we prove this using \pref{lem:structure-trickling-lemma}, similar to the proof of \pref{thm:complement-walk-is-a-good-expander}. We prove the statement by induction on $\ell_1 + \ell_2 = k$. The base case is $\ell_1 + \ell_2 = 0$, i.e. $\ell_1=\ell_2=0$. In this case, the complement walk is just choosing two different $1$-dimensional spaces uniformly at random. This random walk is a $\frac{1}{q^n-1}$-bipartite expander.

Assume the statement is true for any $\ell_1,\ell_2$ s.t. $\ell_1 + \ell_2 \leq k$, and consider the graph operator of the complement walk graph $\comp{\ell_1,\ell_2+1}:\RR^{X(\ell_1)} \to \RR^{X(\ell_2+1)}$, for some $\ell_1,\ell_2$ s.t. $(\ell_1+1)+\ell_2 = k+1$. We need to prove that
\[ \lambda(\comp{\ell_1+1,\ell_2})  \leq \frac{(\ell_1+1)(\ell_2+1)}{q^{n-\ell_1-\ell_2}-1},\]

Note that it is enough to prove for the case where we take $\ell_1+1$ since the adjoint of $\comp{\ell_1+1,\ell_2}$ is $\comp{\ell_1,\ell_2+1}$.
It might be easy to keep in mind the first non-trivial case where $\ell_1+1=1$ and $\ell_2=0$.

Consider the following $2$-dimensional $3$-partite simplicial complex $Y$:
\begin{itemize}
    \item The vertices are $Y \sqbrack{1} = X(\ell_1), Y \sqbrack{2} = X(\ell_1 + 1), Y \sqbrack{3} = X(\ell_2)$.
    \item We connect $(y_1,y_2,y_3) \in Y(2)$ if $y_1 \subset y_2$ and $y_2 \oplus y_3 \in X(\ell_1+\ell_2+2)$. The probability of choosing some $(y_1,y_2,y_3)$ is uniform. However, we can view it as the probability of choosing $(y_2,y_3)$ in the complement walk, and then choosing a $\ell_1+1$-dimensional space $y_1$ in $y_2$.
\end{itemize}
We notice the following:
\begin{enumerate}
    \item $\comp{2,3}$ is the bipartite operator of the bipartite walk between $X(\ell_1+1), X(\ell_2)$.
    \item $\comp{1,3}$ is the bipartite operator of the bipartite walk between $X(\ell_1), X(\ell_2)$. By induction $\lambda(\comp{1,3}) \leq \frac{(\ell_1+1)(\ell_2+1)}{q^{n-\ell_1-\ell_2}-1}$.
    \item for every $w \in Y\sqbrack{1}$, the bipartite operator of the link of $w$ is the $w$-conditioned $\ell_1',\ell_2'$-complement walk for $\ell_1' = 0, \ell_2' = \ell_2$ in the link of $w$. as $\ell_1' + \ell_2' < (\ell_1 + 1) + \ell_2$, we may use the induction assumption to conclude that $\lambda(M_s) \leq \frac{(\ell_2+1)}{q^{n-\ell_2}-1}$.
\end{enumerate}
Hence we can apply \pref{lem:structure-trickling-lemma} and conclude that
\[ \lambda(\comp{2,3}) \leq \frac{(\ell_1+1)(\ell_2+1)}{q^{n-\ell_1-\ell_2}-1} + \frac{(\ell_2+1)}{q^{n-\ell_2}-1} \leq \frac{(\ell_1+2)(\ell_2+1)}{q^{n-\ell_1-\ell_2}-1}.\]

\end{proof}}

\subsection{Random Walks with Fixed Union Size}
\label{sec:random-walks-fixed-union-size}
As a generalization of the complement walk, we can also define a random walk where we go from $\ell_1 \in X(\ell)$ to $\ell_2 \in X(\ell)$ if their union is of size $\ell+1+j$ for some fixed $j > 0$.

\begin{definition}[Fixed Union Size Walk]\torestate{\label{def:fixed-union-size-walk}
Let $X$ be a $d$-dimensional simplicial complex. Let $\ell \geq 0$ and $1 \leq j \leq \ell+1$ so that $\ell + j + 1 \leq d$. The \emph{$\ell,\ell+j$-fixed union walk} is a random walk on $X(\ell)$, where given $t \in X(\ell)$ we:
\begin{enumerate}
    \item Choose $s \in X(\ell + j)$ given that $t \subset \ell$.
    \item Choose $t' \in X(\ell)$ given that $t \cup t' = s$. Equivalently, we can require that $t' \subset s$ and that $\abs{t \cap t'} = \ell+1-j$.
\end{enumerate}}
\end{definition}

For example, if $j=\ell+1$, this walk is the complement walk. If $j=1$ this is just the non-lazy version of the upper-walk (where we choose $t,t'$ if they are contained in some $s \in X(\ell+1)$.
%
%A question arises: how does this walk behave? Does it expand?

In \cite{DiksteinDFH2018}, the authors proved that in a $\lambda$-two-sided high dimensional expander, the difference between the non-lazy upper walk and the $\ell,\ell-1$-lower walk is bounded by $\lambda$ in spectral norm.
\begin{lemma}[\cite{DiksteinDFH2018} Theorem 5.5 item 1]
Let $X$ be a $\lambda$-two-sided spectral expander, then
\[ \norm{A - L} \leq \lambda, \]
where $A$ is the non-lazy $\ell,\ell+1$-upper walk adjacency operator, and $L$ is the $\ell,\ell-1$ lower-walk adjacency operator.
$\qed$
\end{lemma}

We generalize this result, and show that the difference between the $\ell,j$-fixed union walk and the $\ell,\ell-j$-lower walk is bounded by the spectral gap of the $j,j$-complement walk. In particular, by \pref{thm:complement-walk-is-a-good-expander}, the complement walk is bounded by $j^2\lambda$ for any $\lambda$-two-sided high dimensional expander.

\begin{corollary} \label{cor:spectral-gap-of-conditional-intersection}
Let $X$ be a $\lambda$-two-sided high dimensional expander. Fix some $\ell$ and $1 \leq j \leq \ell+1$ so that $\ell+j+1 \leq d$. Denote by $A$ the adjacency operator for the $\ell,j$-fixed union walk. Denote by $L$ the adjacency operator of the $\ell,\ell-j$-lower walk. Then
\[ \norm{A - L} \leq j^2 \lambda. \]

In particular, $\lambda(A) \leq \frac{\ell+1-j}{\ell+1} + \bigO{\ell^2 \lambda}$.
\end{corollary}

\begin{proof}[Proof of \pref{cor:spectral-gap-of-conditional-intersection}]
The last part of $\lambda(A) \leq \frac{1}{\ell-j} + \bigO{\ell^2 \lambda}$, is just using the first part of the corollary, along with \pref{thm:containment-graph-is-a-good-expander} from which we obtain that
\[ \lambda(L) = \frac{\ell+1-j}{\ell+1} + \bigO{\ell^2 \lambda}.\]

as for the first part, consider two functions $f,g:X(\ell) \to \RR$ so that $\norm{f} = \norm{g} = 1$.
\[ \iprod{Af,g} = \Ex[a \in X(\ell-j)]{\iprod{{\comp{j,j}}_a\scloc{f}{a},\scloc{g}{a}}}, \]
where ${\comp{j,j}}_a$ is the $j,j$-complement walk in $X_a$.
This is true since choosing $t,t'$ by the $\ell,j$-fixed union walk, is the same as choosing the intersection $t \cap t' = a \in X(\ell-j)$, and then choosing $t \setminus a, t' \setminus a$ in the complement walk of $X_a(j)$. For each $a \in X(\ell-j)$ we denote
\[\scloc{f}{a} = \scloc{f}{a,0} + \scloc{f}{a,\perp}; \scloc{g}{a} = \scloc{g}{a,0} + \scloc{g}{a,\perp}\]
where $\scloc{f}{a,0}$ is constant and $\scloc{f}{a,\perp}$ is perpendicular to the constant part (and the same for $g$).

\[ \Ex[a \in X(\ell-j)]{\iprod{{\comp{j,j}}_a\scloc{f}{a},\scloc{g}{a}}} = \Ex[a \in X(\ell-j)]{\iprod{\scloc{f}{a,0},\scloc{g}{a,0}}} + \Ex[a \in X(\ell-j)]{\iprod{{\comp{j,j}}_a\scloc{f}{a,\perp},\scloc{g}{a,\perp}}}. \]
\begin{enumerate}
    \item $\abs{\Ex[a \in X(\ell-j)]{\iprod{{\comp{j,j}}_a\scloc{f}{a,\perp},\scloc{g}{a,\perp}}}} \leq j^2\lambda$ by \pref{thm:complement-walk-is-a-good-expander}, since this is applying the complement walk in $X_a$ to an operator perpendicular to the constant functions.

    \item The constant part
    \[ \scloc{f}{a,0} = \Ex[p \in X_a(j)]{\scloc{f}{a}(p)} = \Ex[a \subset t \in X(\ell)]{f(t)},\]
    and by definition this is
    $D_{\ell,\ell-j}f(a)$ (and the same for $g$). Thus
    \[ \Ex[a \in X(\ell-j)]{\iprod{\scloc{f}{a,0},\scloc{g}{a,0}}} = \Ex[a \in X(\ell-j)]{D_{\ell,\ell-j}f(a),D_{\ell,\ell-j}g(a)} = \iprod{D_{\ell,\ell-j}f,D_{\ell,\ell-j}g}.\]
    By definition of the lower-walk
    \[ \iprod{D_{\ell,\ell-j}f,D_{\ell,\ell-j}g} = \iprod{(D_{\ell,\ell-j})^*D_{\ell,\ell-j}f,g} = \iprod{Lf,g}.\]
\end{enumerate}
Combining the two item from above, we get that for every $f,g$ as above
\[ \abs{\iprod{Af,g} - \iprod{Lf,g}} \leq j^2\lambda,\]
or
\[ \norm{A - L} \leq j^2 \lambda.\]
\end{proof}

\subsection{High Dimensional Expander Mixing Lemma}
\label{sec:applications-of-complement-walks}

We can use our newly constructed complement walks and colored walks to prove high dimensional versions of the expander mixing lemma.

Let $A_1 \subset X(j_1),...,A_m \subset X(j_m)$, and denote by $k = \sum_{t=1}^m j_t + m - 1$. We denote by
\[ F(A_1,...,A_k) \defeq \sett{s \in X(k)}{\forall j \exists s_j \in A_j \; \; s_j \subset s}, \]
i.e. all $k$-faces that contain a subface from each $A_j$. For example, when $m=2$ and $j_1=j_2=0$, $F(A_1,A_2)$ are all edges between $A_1$ and $A_2$, in the underlying graph of $X$.

\begin{lemma}[High dimensional expander mixing lemma - two-sided]\torestate{
\label{lem:two-sided-HDEML}
Let $X$ be a $d$-dimensional $\lambda$-two sided link expander. Let $j_1,j_2,...,j_m \leq d$, and $A_1 \subset X(j_1), A_2 \subset X(j_2),..., A_m \subset X(j_m)$ s.t. for any $j_{\ell_1} \ne j_{\ell_2}$, and any $s \in A_{j_{\ell_1}}, t \in A_{j_{\ell_2}}$, $s \cap t = \emptyset$.
Then
\[ \Abs{\prob{F(A_1,A_2,...,A_k)} - \binom{k+1}{j_1+1,j_2+1,...,j_m+1} \prod_{j=1}^m \prob{A_j}} \leq C \lambda \sqrt[m]{\prod_{j=1}^m \prob{A_j}}\]
where $C$ depends on $m,d$ only.\footnote{here $\binom{k+1}{j_1+1,j_2+1,...,j_m+1}$ is the number of partitions of a set of size $k+1$ to sets of size $j_1+1,j_2+1,...,j_m+1$.}}
\end{lemma}

\begin{lemma}[High dimensional expander mixing lemma - one-sided $d+1$-partite]
\torestate{\label{lem:one-sided-HDEML}
Let $X$ be a $\lambda$-one sided $d+1$-partite link expander. Let $I_1,...,I_m \subset [d+1]$ be pairwise disjoint colors, and let $A_1 \subset X[I_1],..., A_m \subset X[I_m]$.
Then
\[ \Abs{\prob{F(A_1,A_2,...,A_k)} - \prod_{j=1}^m \cProb{}{A_j}{X\sqbrack{I_j}}} \leq C \lambda \sqrt[m]{\prod_{j=1}^m \cProb{}{A_j}{X \sqbrack{I_j}}}\]
where $C$ depends on $m,d$ only.}
\end{lemma}

\paragraph{Comparison with previous results} There are other suggested expander mixing lemmas for high dimensional expanders. For example, the lemma in \cite{oppenheim3} states that on a $\lambda$-two-sided high dimensional expander, for $A_1,...,A_m \subset X(0)$ we get that
\[ \Abs{\prob{F(A_1,A_2,...,A_k)} - (k+1)! \prod_{j=1}^m \prob{A_j}} \leq C \lambda \sqrt{\min_{j\ne i} \prob{A_i}\prob{A_j}}.\]

The lemma in \cite{coloring}, had a similar statement for a special case of Ramanujan complexes.

Our lemma generalizes these results. It deals with faces of all sizes, and not only vertices. This shows that link expanders have pseudorandom behavior in all levels of the complex.

\medskip

We give the proof for the two-sided case. The one sided case's proof is similar.

\begin{proof}[Proof of \pref{lem:two-sided-HDEML}]
The proof is by induction on $m$. The base case where $m=1$ is obvious from the definition.

Let $X$ and $A_1 \subset X(j_1),...,A_{m+1} \subset X(j_{m+1})$ be as above. It is enough to prove that for any $A_{i}$ that
\[ \abs{\prob{F(A_1,A_2,...,A_k)} - \prod_{j=1}^m \prob{A_j}} \leq C \lambda \sqrt{\prob{A_i} \sqrt[m]{\prod_{i \ne j=1}^{m+1} \prob{A_j}}},\]
because the geometric mean of RHS is
\[\prod_{i=1}^{m+1}\left( C \lambda \sqrt{\prob{A_i} \sqrt[m]{\prod_{i \ne j=1}^{m+1} \prob{A_j}}}\right)^\frac{1}{m+1} =  C \lambda \sqrt[m+1]{\prod_{j=1}^{m+1} \prob{A_j}}. \]

Indeed denote by $\one_{F(A_1,...,A_m)}, \one_{A_{m+1}} : X(k) \to \RR$ the indicators of $F(A_1,...,A_m)$ and $A_{m+1}$ respectively. Consider the expression \[\iprod{M \one_{A_{m+1}}, \one_{F(A_1,...,A_m)}},\]
where the operator $M \defeq \comp{j_{m+1},k - j_{m+1} - 1}$ is the complement walk operator.
As we can see
\[ \iprod{M \one_{A_{m+1}}, \one_{F(A_1,...,A_m)}} = \]
\[  \Ex[s_1 \in X(j_{m+1}), s_2 \in X(k-j_{m+1} - 1); s_1 \dunion s_2 \in X(k)]{\one_{A_{m+1}}(s_1) \one_{F(A_1,...,A_m)}(s_2)} = \]
\[\prob{F(A_1,...,A_{m+1})}\frac{1}{\binom{k+1}{j_{m+1}+1, k-j_{m+1}}},\]
As this is exactly the probability to get a face $ t \in F(A_1,...,A_{m+1})$, and partition it to $s_1, s_2$ (there is only one such partition so that $s_1 \in A_{m+1}$ and $s_2 \in F(A_1,...,A_m)$, because of the mutual disjointness property of the $A_{j_i}$'s).

On the other hand, we can decompose
\[\one_{A_{m+1}} = \one_{A_{m+1}}^0 + \one_{A_{m+1}}^\perp\]
and
\[\one_{F(A_1,...,A_{m})} = \one_{F(A_1,...,A_{m})}^0 + \one_{F(A_1,...,A_{m})}^\perp,\]
to the constant part and the part perpendicular to it. Thus
 \[ \iprod{M \one_{A_{m+1}}, \one_{F(A_1,...,A_m)}} = \iprod{M \one_{A_{m+1}}^0, \one_{F(A_1,...,A_m)}^0} + \iprod{M \one_{A_{m+1}}^\perp, \one_{F(A_1,...,A_m)}^\perp} .\]
 Thus from Cauchy-Schwartz:
 \[\Abs{\iprod{M \one_{A_{m+1}}, \one_{F(A_1,...,A_m)}} - \iprod{M \one_{A_{m+1}}^0, \one_{F(A_1,...,A_m)}^0} } \leq \lambda(M) \norm{\one_{A_{m+1}}^\perp} \norm{\one_{F(A_1,...,A_m)}^\perp}. \]
 The product between constant parts is equal to the product of probabilities and by induction:
 \[ \iprod{M \one_{A_{m+1}}^0, \one_{F(A_1,...,A_m)}^0} = \prob{A_{m+1}}\prob{F(A_1,...,A_m)}. \]

 Thus
\[\Abs{\iprod{M \one_{A_{m+1}}, \one_{F(A_1,...,A_m)}} - \prob{A_{m+1}}\prob{F(A_1,...,A_m)} } \leq  \lambda(M) \norm{\one_{A_{m+1}}^\perp} \norm{\one_{F(A_1,...,A_m)}^\perp}.\]

By the triangle inequality
\[\Abs{\iprod{M \one_{A_{m+1}}, \one_{F(A_1,...,A_m)}} - \prod_{j=1}^{m+1} \prob{A_j}} \leq \]
\[\lambda(M) \norm{\one_{A_{m+1}}^\perp} \norm{\one_{F(A_1,...,A_m)}^\perp} + \Abs{\prob{A_{m+1}}\prob{F(A_1,...,A_m)}  - \prod_{j=1}^{m+1} \prob{A_j}} \leq \]
\[ C \lambda \sqrt{\prob{A_{m+1}} \sqrt[m]{\prod_{j=1}^m \prob{A_j}} } .\]
\end{proof}

\section*{Acknowledgement}
We wish to thank Prahladh Harsha for many helpful discussions.

\addcontentsline{toc}{section}{References}
\bibliographystyle{alpha}
\bibliography{ms}

%\printbibliography
\clearpage

\appendix
\section{Standard Definitions and Claims} \label{sec:preliminaries}

In this appendix we give the necessary background and conventions we use throughout the paper. Most results and claims in this section are standard, and thus given without proof.
\subsection{Expander graphs}

    Every weighted undirected graph induces a random walk on its vertices: Let $G = (V,E)$ be a finite weighted graph with a probability weight function $\mu: E \to [0,1]$. The transition probability from $v$ to $u$ is
    \[\frac{\mu(\set{u,v}) }{ \sum_{w \sim v} \mu(\set{v,w}) } .\]
    Denote by $A=A(G)$ the Markov operator associated with this random walk. We call this operator the \emph{adjacency operator}.

    $A$ is an operator on real valued functions on the vertices, where
    \[ \forall v\in V \; Af(v) = \Ex[u \sim v]{f(u)}. \]
    The expectation is taken with respect to the graph's probability on vertices, conditioned on being adjacent to $v$.

    %Recall that a finite weighted graph as above gives rise to a probability measure on $V$, i.e. $\prob{v} = \sum_{w \sim v} \mu(\set{v,w})$, and this defines $\ell_2(V)$ as the space of real valued functions, with the inner product
    %\[ \forall f,g \in \ell_2(V) \; \iprod{f,g} = \Ex[v \in V]{f(v)g(v)} .\]
    %One can prove the adjacency operator is self adjoint with respect to $\ell_2(V)$ and thus diagonalizable, with real eigenvalues.

    $A$'s eigenvalues are in the interval $[-1,1]$. We denote its eigenvalues by $\lambda_1 \geq \lambda_2 \geq ... \geq \lambda_n$ (with multiplicities). The largest eigenvalue is always $\lambda_1 = 1$, and it is obtained by the constant function. The second eigenvalue is strictly less than $1$ if and only if the graph is connected.

    \begin{definition}[spectral expanders]\torestate{\label{def:spectral-expanders}
        Let $G$ be a graph. $G$ is a \emph{$\lambda$-one sided spectral expander} for some $0 \leq \lambda < 1$, if
        \[ \lambda_2 \leq \lambda.\]
        $G$ is a \emph{$\lambda$-two sided spectral expander} for some $0 \leq \lambda < 1$, if
        \[ \max (\abs{\lambda_2},\abs{\lambda_n}) \leq \lambda.\]}
    \end{definition}

    There is another notion of graph expansion that we'll need in this paper, called edge expansion. Intuitively, an edge expander is a graph where every set of vertices has a large number of outgoing edges.

    \begin{definition}[edge expansion]\torestate{\label{def:edge-expansion}
    Let $G$ be a weighted graph. The \emph{edge expansion} of $G$ is
    \[\Phi(G) = \min \sett{\frac{\prob{E(S,V \setminus S)}}{\prob{S}}}{S \subset V, \; 0 < \prob{S} \leq \frac{1}{2}},\]
    where $E(S,V \setminus S)$ is the set of all edges between $S$ and $V \setminus S$.}
    \end{definition}

    There is a connection between spectral expansion and edge expansion:
    \begin{theorem}[Cheeger's inequality]
    Let $G$ be any weighted graph. Then
    \[ \frac{1-\lambda_2}{2} \leq \Phi(G) \leq \sqrt{2(1-\lambda_2)} .\]
    $\qed$
    \end{theorem}

\subsubsection{Bipartite Graphs and Bipartite Expanders}
    A bipartite graph is a graph where the vertex set can be partitioned to two independent sets $V = L \dunion R$, called sides. Bipartite graphs are sometimes easier to analyze than graphs, and arise naturally when studying STAV-structures.

    %A bipartite graph can be a $\lambda$-one sided spectral expander, but it is never a non-trivial two-sided spectral expander. The smallest eigenvalue of its adjacency operator is always $\lambda_n =-1$.

    %\paragraph{Double Covers} \label{sec:double-cover}
    %A double cover of a (not necessary bipartite) graph $G = (V_0,E_0)$, is the bipartite graph $D(G) = (L,R,E)$ where $L,R$ are two copies of $V_0$. For any $v \in V_0$ we denote its copies in $L,R$ by $v_L, v_R$ respectfully.
    %We connect $v_{L} \in L, u_R \in R$ if their original vertices $v,u \in V_0$ share an edge in $G$. The weights for edges are:
    %\[ \mu_{D(G)}(\set{v_L,u_R}) = \mu_{D(G)}(\set{v_R,u_L}) = \frac{1}{2}\mu_{G}(\set{v,u}).\]

    %A double cover is sometimes convenient in analyzing two-sided expansion in graphs. Concretely:
    %\begin{claim}
    %Let $G$ be any weighted graph. $\lambda$ is an eigenvalue of $G$'s adjacency operator, if and only if $\pm \lambda$ both are eigenvalues of $D(G)$'s adjacency operator. $\qed$
    %\end{claim}
    %In particular, $G$ is a $\lambda$-two-sided spectral expander if and only if $D(G)$ is a $\lambda$-one sided spectral expander.

    \paragraph{The Bipartite Adjacency Operator}
    In a bipartite graph, we view each side as a separate probability space, where for any $v \in L$ (resp. $R$), $\prob{v} = \sum_{w \sim v} \mu(\set{v,w})$. We can define the \emph{bipartite adjacency operator} as the operator  $B: \ell_2(L) \to \ell_2(R)$ by
    \[\forall f \in \ell_2(L), v \in R, \; Bf(v) = \Ex[w \sim v]{f(u)}\]
    where the expectation is taken with respect to the probability space $L$, conditioned on being adjacent to $v$.

    %We may define $B^*: \ell_2(R) \to \ell_2(L)$ as the bipartite operator for the opposite side. As the notation suggests, $B^*$ is adjoint to $B$ with respect to the inner products of $\ell_2(L), \ell_2(R)$.

    We denote by $\lambda(B)$ the spectral norm of $B$ when restricted to $\ell_2^0(L) = \set{\one}^\bot$, the orthogonal complement of the constant functions (according to the inner product the measure induces). Namely
    \[ \lambda(B) = \sup \sett{\iprod{Bf,g}}{\norm{g},\norm{f}=1}.\]
    \begin{definition}[Bipartite Expander]\torestate{\label{def:bipartite-expander}
    Let $G$ be a bipartite graph, let $\lambda < 1$. We say $G$ is a \emph{$\lambda$-bipartite expander}, if
    $\lambda(B) \leq \lambda$.}
    \end{definition}
    %One may show that $\lambda(B) = \sqrt{\lambda_2(B^*B)}$. Note that this root is always well defined over $\RR$, since $B^*B$ is non-negative semidefinite.

    %There is a connection between singular values of the bipartite adjacency operator, and the eigenvalues of the adjacency operator of the graph:
    %\begin{claim}
    %Let $G = (L,R,E)$ be a weighted bipartite graph. $\lambda \in \RR$ is an eigenvalue of $B^*B$ if and only if $\pm \sqrt{\lambda}$ are eigenvalues of $A(G)$ - the (non-bipartite) adjacency operator of $G$. $\qed$
    %\end{claim}
    %% Proof: if f \in eig(B^*B) then (f,\pm (1/sqrt \lambda) f) is the eigenvalue of A%%%.
    \medskip
    \paragraph{Sampling Graph}We also define a sampling graph, a notion close in some sense to expanders.
    \begin{definition}[Sampling Graph]
    \torestate{\label{def:sampling-graph}
        Let $G = (L,R,E)$ be a bipartite graph, and $\delta < 1$. We say that $G$ has the \emph{$\delta$-sampling property} if the following holds:
        For any set $B \subset V$ of size greater than $\prob{C} \geq \delta$, the set $T = \set{a: \cProb{v \in V}{v \in C}{v \in \adj{a}} \geq \frac{1}{3} \delta}$ has size at least $\frac{1}{3}$.}
    \end{definition}

\subsection{Properties of Expander Graphs}

In this subsection we develop the necessary properties of expander graphs, that we will need in \pref{sec:main-thoerems-proof}.

    \paragraph{Edge-Expander Partition Property}
    The following claim is also useful in the proof of the main theorem. It says that if we partition the vertices, and there are few edges between the partition's parts, then one set in the partition is larger than $\frac{1}{2}$.

    \begin{claim}[Edge-Expander Partition Property]
    \label{claim:edge-expander-partition-property}
    Let $G = (V,E)$ be a $c$-edge expander. Let $V = B_1 \dunion ... \dunion B_n$, partitioned into sets, and suppose that there are less than $\frac{c}{2}$ edges between parts of the partition, namely:
    \[ \frac{1}{2}\sum_{i=1}^n \prob{E(B_i,B_i^c)} < \frac{c}{2}.\]
    Then there exists $i$ such that $\prob{B_i} \geq \frac{1}{2}$.
    \end{claim}

    \begin{proof}[Proof of \pref{claim:edge-expander-partition-property}]

    Assume towards contradiction that for all $1 \leq i \leq n$, $\prob{B_i} <\frac{1}{2}$.

    From our assumption, there are less than $\frac{c}{2}$ edges between parts of the partition, namely
    \[ \frac{c}{2} > \frac{1}{2}\sum_{i=1}^n \prob{E(B_i,B_i^c)} \geq \frac{c}{2}\sum_{i=1}^n\prob{B_i},\]
    where the second inequality is from edge expansion. $B_i$'s are a partition of the vertices, thus $\sum_{i=1}^n\prob{B_i} = 1$, a contradiction.

    \end{proof}

    \paragraph{Expander Mixing Lemma}
    A classical result in expander graphs is the \emph{expander mixing lemma}, that intuitively says that the weight of the edges between any two vertex sets $S,T \subset V$ is proportionate to the probabilities of $S,T$.
    \begin{lemma}[Expander Mixing Lemma]
    \label{lem:EML}
    Let $G = (V,E)$ be a $\lambda$-two sided spectral expanders. Then for any $S,T \subset V$
    \[ \abs{\prob{E(S,T)} - \prob{S} \prob{T} } \leq \lambda \sqrt{\prob{S} \prob{T} (1-\prob{S})(1-\prob{T})} . \]
    $\qed$
    \end{lemma}

    Bipartite graphs have their own type of expander mixing lemma:
    \begin{lemma}[Bipartite Expander Mixing Lemma]
        \label{lem:bipartite-EML}
        Let $G = (L,R,E)$ be a bipartite $\lambda$-one sided spectral expander. Then for any $S \subset L, T \subset R$
        \[ \abs{\prob{E(S,T)} - \Prob[v \in L]{v \in S} \Prob[w \in R]{w \in T} } \leq \lambda \sqrt{\prob{S} \prob{T} (1-\prob{S})(1-\prob{T})} . \]
        $\qed$
    \end{lemma}

    \paragraph{Expander Sampler Property}
    In \cite{DinurK2017} the authors showed that bipartite $\lambda$-one sided spectral expander has the following useful sampler property.
    \begin{lemma}[Sampler Property, by \cite{DinurK2017}]
    \label{lem:sampler-lemma}
        Let $G = (L,R,U)$ be a bipartite $\lambda$-one sided spectral expander. Let $B \subset R$ be any set of vertices, and $c > 0$. then $T = \sett{v \in L}{\abs{\cProb{w \in R}{w \in S}{w \sim v} - \prob{S}} > c}$ of vertices who view $S$ as "large", satisfies:
        \[\prob{T} \leq  \frac{\lambda^2}{c^2}\prob{S}.\]
    \end{lemma}

    \paragraph{Almost Cut Approximation Property}
    As a corollary to the expander mixing lemma, we get the following useful approximation property. In an expander graph, if the number of outgoing edges from some $A \subset V$, is an approximation to the size of $A$ or $V \setminus A$. The following claim generalizes this fact to the setting where we count only outgoing edges from $A$ to a (large) set $B \subset V \setminus A$.

    \begin{claim}[Almost Cut Approximation Property] \torestate{ \label{claim:almost-cut}
        Let $G = (V,E)$ be a $\lambda$-two sided spectral expander. Let $V = A \dunion B \dunion C$, s.t. $\prob{A} \leq \prob{B}$. Then
        \begin{equation} \label{eq:eml-bound}
                \prob{A} \leq \frac{1}{(1-\lambda)\prob{B}} \left ( \prob{E(A,B)} + \lambda \prob{C} \right ).
        \end{equation}

        In particular, if  \( \prob{A}, 1-\lambda = \Omega(1)\) then
        \[ \prob{A} = \bigO{  \prob{E(A,B)} + \lambda \prob{C} }. \]
    }
    \end{claim}

For bipartite expanders we have an analogues almost approximation cut property, similar to \pref{claim:almost-cut}.

\begin{claim}[Almost Cut Approximation Property - Bipartite expanders] \torestate{ \label{claim:almost-cut-bipartite}
    Let $G = (L,R,E)$ be a $\lambda$-bipartite expander for $\lambda < \frac{1}{2}$. Let $V = A \dunion B \dunion C$, s.t. $\prob{A} \leq \prob{B}$ (where the probability is taken over all the graph). Then
    \begin{equation} \label{eq:eml-bound-bipartite}
            \prob{A} \leq \frac{1}{2(1-2\lambda)\prob{B}} \left ( \prob{E(A,B)} + \lambda 4\prob{C} \right ).
    \end{equation}

    In particular, if  \( \prob{A}, 1-\lambda = \Omega(1)\) then
    \[ \prob{A} = \bigO{  \prob{E(A,B)} + \lambda \prob{C} }. \]
}
\end{claim}

    \begin{proof}[Proof of \pref{claim:almost-cut}]
        By the expander mixing lemma
        \[ \prob{A} \prob{B} \leq \prob{E(A,B)} + \lambda \sqrt{\prob{A}\prob{B}(1-\prob{A})(1-\prob{B})}. \]
        The expression inside the square root is equal \(\prob{A}\prob{B}(\prob{C} + \prob{A}\prob{B}) \), since \(\prob{C} = 1 - \prob{A} - \prob{B}\). Thus we may write
        \[ \prob{A} \prob{B} \leq \prob{E(A,B)} +  \lambda (\prob{A}\prob{B} + \prob{C}). \]
        The claim easily follows by direct calculation.
    \end{proof}

    \begin{proof}[Proof of \pref{claim:almost-cut-bipartite}]
        Denote the restriction of a set to $L$ or $R$ by $X_L$ or $X_R$ respectively. Denote $a_L = \Prob[L]{A_L}$ and the same for $b_L,c_L,a_R,b_R,c_R$.
        By the bipartite expander mixing lemma
        \[ a_Lb_R \leq \prob{E(A_L,B_R)} + \lambda \sqrt{a_Lb_R(1-a_L-b_R+a_Lb_R)}, \]
        and
        \[ a_Rb_L \leq \prob{E(A_R,B_L)} + \lambda \sqrt{a_Rb_L(1-a_R-b_L+a_Rb_L)}. \]
        The expressions inside both square roots are less or equal to
        \[(a_Rb_L + a_Lb_R)((1-a_R-b_R) + (1-a_L-b_L) +(a_Rb_L + a_Lb_R)). \]
        This in turn, is less or equal than
        \[((1-a_R-b_R) + (1-a_L-b_L) +(a_Rb_L + a_Lb_R))^2.\]
        Notice that we may write $(1-a_R-b_R) + (1-a_L-b_L) = c_L + c_R = 2\prob{C}$. Thus by combining both inequalities we obtain:
        \[ (1-2\lambda) (a_Rb_L + a_Lb_R) \leq E(A,B) + 4\lambda\prob{C} .\]
        Wlog $a_L \geq a_R$ thus we obtain that
        \[(a_Rb_L + a_Lb_R) \geq a_L(b_L+b_R) \geq a_L(2\prob{B}) \geq 2\prob{A}\prob{B}.\]
        Thus
        \[\prob{A} \leq \frac{1}{2(1-2\lambda)\prob{B}}\prob{E(A,B)} + 4\lambda\prob{C}.\]
    \end{proof}
\subsection{Simplicial Complexes and high dimensional expanders}
\label{sec:HDX}
We include here the basic definitions needed for our results. For a more comprehensive introduction to this topic we refer the reader to \cite{DinurK2017} and the references therein. 

A simplicial complex is a hypergraph that is closed downward with respect to containment. It is called $d$-dimensional if the largest hyperedge has size $d+1$. We refer to $X(\ell)$ as the hyperedges (also called faces) of size $\ell+1$. $X(0)$ are the vertices.

We define a weighted simplicial complex. Suppose we have a $d$-dimensional simplicial complex $X$ and a probability distribution $\mu : X(d) \to [0,1]$. We consider the following probabilistic process for choosing lower dimensional faces:
    \begin{enumerate}
    \item Choose some $d$-face $s_d \in X(d)$ with probability $\mu(s_d)$.
    \item Given the choice of $s_d$, choose sequentially a chain of faces contained in $s_d$, $(\emptyset \subset s_1 \subset ... \subset s_d)$ uniformly, where $s_i \in X(i)$.
\end{enumerate}

%For example, if $X$ is a weighted graph, with weights that are  a probability distribution $\mu: X(1) \to [0,1]$, then this process amounts to choosing an edge according to its weight, and then choosing one of its vertices, with equal probability. The probability of choosing a vertex $v$ is
%\[ \prob{v} = \frac{1}{2} \sum_{e \in X(1), v \in e} \mu(e). \]

 For any $s \in X(k)$ we denote by
\[ \prob{s} = \cProb{}{\set{(\emptyset \subset s_0 \subset ... \subset s_d)}}{ s_k = s}  .\]
For all $s_k \in X(k), s_\ell \in X(l)$, we will write $\cProb{}{s_k}{s_\ell}$ the probability of the $k$-face in the sequence is $s$, given that the $l$-face is $s_\ell$.

%More formally we define this process as follows:
%\begin{definition}[Measured simplicial complex]
%    \label{def:meausred-simplicial-complex}
%    Let $X$ be a pure $d$-dimensional simplicial complex, and let $\mu : X(d) \to [0,1]$ such that $\sum_{t \in X(d)} \mu(t) = 1$. We define a probability measure on sequences $\set{(\emptyset \subset s_0 \subset ... \subset s_d) : \forall -1\leq k \leq d \, \, s_k \in X(k) }$, by
%    \[ \prob{ \set{(\emptyset \subset s_0 \subset ... \subset s_d)} } = \frac{\mu(s_d)}{(d+1)!} .\]
%\end{definition}

%One $\mu$ commonly used is the uniform measure $\mu(s_d) = \frac{1}{|X(d)|}$. In this case $Pr[s]$ is proportionate to the number of $d$-faces containing $s$. When we discuss the complete complex, the measure we use is uniform on $X(d)$, and it induces a uniorm measure on every $X(i), i < d$.

From here throughout the rest of the paper, when we refer to a simplicial complex $X$, we always assume that there is a probability measure on it constructed as above.

A link of a face in a simplicial complex, is a generalization of a neighbourhood of a vertex in a graph:

\begin{definition}[link of a face]\torestate{\label{def:link}
Let $s \in X(k)$ be some $k$-face. The \emph{link} of $s$ is a $d-(k+1)$-dimensional simplicial complex defined by:
\[X_s = \set{ t \backslash s : s \subseteq t \in X }.\]
The associated probability measure $Pr_{X_s}$, for the link of $s$ is defined by
\[ \Prob[X_s]{t} = \cProb{X}{t \cup s}{s},\]
where $Pr_{X}$ is the measure defined on $X$.}
\end{definition}

%For example, if we have a $2$-dimensional simplicial complex $X$, then for every vertex $v \in X(0)$, the $1$-dimensional simplicial complex $X_v$ is a graph whose vertices are $X_v(0)$ and edges are $X_v(1)$.
%\begin{itemize}
%    \item $X_v(0) = \set{ u : \set{u,v} \in X(1) }$ - all vertices connected to $v$ by an edge.
%    \item $X_v(1) = \set{ \set{u,w} : \set{u,w,v} \in X(2) }$ - two vertices share an edge if together with $v$ they share a triangle.
%\end{itemize}

\begin{definition}[underlying graph]\torestate{\label{def:underlying-graph}
    The \emph{underlying graph} of a simplicial complex $X$ with some probability measure as define above, is the graph whose vertices are $X(0)$ and edges are $X(1)$, with (the restriction of) the probability measures of $X$ to the vertices and edges.}
\end{definition}

We are ready to define our notion of high dimensional expanders: the one-sided and two-sided link expander.
\begin{definition}[one-sided and two-sided link expander] \torestate{\label{def:prelim-link-expander}
    Let $0 \leq \lambda < 1$. A simplicial complex $X$ is a \emph{$\lambda$-two sided link expander} (or $\lambda$-two sided HDX) if for every $-1 \leq k \leq d-2$ and every $s \in X(k)$, the underlying graph of the link $X_s$ is a $\lambda$-two sided spectral expander.

    Similarly, $X$ is a \emph{$\lambda$-one sided link expander} (or $\lambda$-one sided HDX) if for every $-1 \leq k \leq d-2$ and every $s \in X(k)$, the underlying graph of the link $X_s$ is a $\lambda$-one sided spectral expander.}
\end{definition}
When $X$ is a graph, this definition coincides with the definition of a spectral expander.

We remark that it is a deep theorem that there exist good one-sided and two-sided high dimensional expanders with bounded degree \cite{LubotzkySV2005a}. %\begin{example}
%   The complete complex $K_n^d$ is a $\frac{1}{n-d-2}$-two sided HDX. This is because the underlying graph of every link $s \in X(k)$ of the complete complex, is a complete graph with $n-k-1$ vertices (recall that the complete graph on $m$ vertices is a $\frac{1}{m-1}$-two sided spectral expander).
%\end{example}
%\cite{LubotzkySV2005a} and \cite{LubotzkySV2005b} papers, showed the existence of one-sided high dimensional expanders, that are \emph{sparse}, namely, that every vertex $v \in X(0)$ is contained in a bounded number of top level faces $s \in X(d)$. In \cite{DinurK2017}, the authors used the previous result to show the existence of two-sided high dimensional expanders with the same properties.

\subsubsection*{$d+1$-partite simplicial complexes}
\label{sec:partite-simplicial-complexes}
A $d+1$-partite simplicial complex is a generalization of a bipartite graph. We say a $d$-dimensional simplicial complex is \emph{$d+1$-partite} if we can partition the vertex set
\[V = V_0 \dunion V_1 \dunion ... \dunion V_d, \]
s.t. any $d$-face $s \in X(d)$, contains a vertex from each $V_i$, i.e. $\abs{s \cap V_i} = 1$.

The \emph{color} of a $k$-face $s \in X(k)$, is the set of all indexes of $V_i$'s, that intersect with $s$. I.e.
\[ \col(s) = \set{j \in [d] : \abs{s \cap V_j} = 1}. \]

For any $J \subset [d]$, we denote
\[ X \sqbrack{J} = \set{s \in X : \col(s) = J} . \]
When $J = \set{i}$, we abuse the notation and write $X \sqbrack{i}$ instead of $X \sqbrack{\set{i}}$ (not to be confused with $X(i)$).

\section{From Independent Choice to Expanding Choice}
\label{sec:independent--to-expanding-distributions}
In \pref{sec:agreement-on-hdx}, \pref{sec:agreement-on-nbrhoods} and \pref{sec:agreement-in-the-grassmann} we showed that a number of agreement tests were sound. The agreement test's distributions had in common the following property: given the choice of intersection $t$, we chose the sets $s_1,s_2$ independently.
This property is very helpful in analyzing the expansion of the conditioned $STS_{a,v}$-graph, as required when showing that \pref{ass:v-a-graph} holds.

In this appendix, we show that if the choice of $s_1,s_2$ given $t$, is done according to an expanding graph, then we can get a similar result.

\begin{definition}[$STS_t$-graph]\torestate{\label{def:sts-t-graph}
Let $X = (S,T,A,V)$ be any STAV-structure. For a fixed $t \in T$, an $\sts_{t}$-Graph is has vertex set $\set{s\supset t}$ and the probability of choosing an edge $\set{s_1,s_2}_t$ is given by $2\cProb{STS}{(s_1,s_2)}{s_1,s_2 \supset t}$.}
\end{definition}

\begin{claim}
\label{claim:independent-vs-expander}
Let $X = (S,T,A,V)$ be any STAV-structure. Let $D_1,D_2$ be two $STS$-distributions on $X$ so that for all $t \in T$:
\begin{enumerate}
    \item The choice of $s_1,s_2 \sim D_1$ given $t$ is independent.
    \item The $\sts_t$-graph for $D_2$ is a $\frac{1}{3}$-two-sided spectral expander.
\end{enumerate}
Denote by $\varepsilon_i = \disagr{D_i}(\FF)$, namely, the probability to sample $(t,s_1,s_2) \sim D_i$ so that $\rest{f_{s_1}}{t} \ne \rest{f_{s_2}}{t}$. Then
\[ \frac{1}{6} \varepsilon_1  \leq \varepsilon_2 \leq 6 \varepsilon_1. \]
\end{claim}

The constant $\frac{1}{3}$ is arbitrary, any constant bounded away from $1$ will suffice.

As a corollary to this claim,
\begin{corollary}
Let $X = (S,T,A,V)$ be any STAV-structure. Let $D_1,D_2$ be two $STS$-distributions on $X$ so that for all $t \in T$, the $\sts_t$-graphs are $\frac{1}{3}$-edge spectral expanders for both $D_1$ and $D_2$.
Denote by $\varepsilon_i = \disagr{D_i}(\FF)$, namely, the probability to sample $(t,s_1,s_2) \sim D_i$ so that $\rest{f_{s_1}}{t} \ne \rest{f_{s_2}}{t}$. Then
\[ \varepsilon_2 \leq 36 \varepsilon_1. \]
In particular, $D_1$ yields a $\gamma$-approximate $c$-sound agreement test if and only if $D_1$ yields a $\gamma$-approximate $36c$-sound agreement test (including the exact case where $\gamma = 0$).
$\qed$
\end{corollary}

The proof of the corollary is by two uses of the claim above. We leave the details to the reader.

\begin{example}[Simplicial Complexes]
We recall that for a simplicial complex $X$ we can define agreement tests for the ground set $V = X(0)$ and $S = X(d)$. Previously we defined the $D_{d,\ell}$ distribution where we choose $s_1,s_2$ independently given that they contain some $\ell$-face $t \in X(\ell)$.

Observe the following test distribution $Up_{2k,\frac{k}{4}}$ for a $2k$-dimensional simplicial complex.
\begin{enumerate}
    \item Sample $r \in X(2k)$ and $t \in X(\frac{k}{4})$.
    \item Sample $s_1,s_2 \in X(k)$, given that $t \subset s_1,s_2 \subset r$.
\end{enumerate}
Given any $t \in X(\frac{k}{4})$ the $STS_t$-graph above is two steps in the $\frac{k}{2},\frac{3k}{2}$-containment walk, thus an edge expander. By \pref{claim:independent-vs-expander}, we can immediately obtain that $\disagr{UP_{2k,\frac{k}{4}}} = \bigO{\disagr{D_{k,\frac{k}{4}}}}$. By \pref{thm:main-agreement-theorem-two-sided} this agreement test is exact $c$-sound.

\medskip

We can take this argument one step further. Consider the following test distribution $UP_{2k}$, where we only condition on $s_1,s_2 \subset r$, namely:
\begin{enumerate}
    \item Sample $r \in X(2k)$.
    \item Sample $s_1,s_2 \in X(k)$, given that $s_1,s_2 \subset r$.
\end{enumerate}
This distribution was the main distribution analyzed in the agreement theorem in \cite{DinurK2017}.

We expect that $s_1$ and $s_2$ intersect on a set of size $\frac{k}{2}$. Thus by a simple Markov argument, $\Prob[s_1,s_2 \sim Up_{2k}]{\abs{s_1 \cap s_2} \geq \frac{1}{4}k}  = \Omega(1)$.
Thus if $\disagr{Up_{2k}} \leq \varepsilon$, then conditioned on intersecting on a set of size $\frac{k}{4}$, the rejection probability is still $\bigO{\varepsilon}$. In conclusion, we get that
\[\disagr{UP_{2k}} = \bigO{\disagr{UP_{2k,\frac{k}{4}}}} = \bigO{\disagr{D_{k,\frac{k}{4}}}}.\]
By \pref{thm:main-agreement-theorem-two-sided}, we obtain a new proof to the theorem in \cite{DinurK2017} that this distribution gives rise to a $c$-sound agreement test, for a good enough two-sided spectral expander.
\end{example}
\Ynote{cite DK theorem formally.}

\begin{proof}[Proof of \pref{claim:independent-vs-expander}]
For any $t \in T$ and $i=1,2$ we denote by $\varepsilon_{i,t}$ the probability of sampling $s_1,s_2 \supset t$ who disagree on $t$.
It is easy to see that $\Ex[t]{\varepsilon_{i,t}} = \varepsilon_i$, so it will suffice to show that $\frac{1}{6} \varepsilon_{1,t} \leq  \varepsilon_{2,t} \leq  6\varepsilon_{1,t}$ for every $t \in T$.

We begin by showing that $\frac{1}{6}\varepsilon_{1,t} \leq  \varepsilon_{2,t}$ or equivalently that $\varepsilon_{1,t} \leq  6\varepsilon_{2,t}$. If $\varepsilon_{2,t} \geq \frac{1}{6}$, then $\varepsilon_{1,t} \leq 1 \leq  6\varepsilon_{2,t}$.

Otherwise observe the partition of $\ul{t}{s}$ into $V_1,...,V_n$ where
\[ V_i = \set{\rest{f_s}{t} = h_i},\]
for all possible assignments $h_i:t \to \Sigma$. By the edge expander partition property \pref{claim:edge-expander-partition-property}, there is a set $V_i$ such that $\prob{V_i} \geq \frac{1}{2}$. Without loss of generality it is $V_1$. By edge expansion we get that
\[\prob{V_i^c} \leq 3\prob{E(V_i,V_i^c)} \leq 3 \varepsilon_{2,t}.\]

Observe that the $(s,t)$ marginal according to $D_1$ and $D_2$ are identical since they are both $STS$-test distributions of the same STAV. Thus in particular when we write $\prob{V_i}$ it doesn't matter whether we are sampling $s$ in the $STS_t$-graph according to $D_1$ or according to $D_2$.

Returning to $STS_t$-graph of $D_1$, the probability of choosing $s_1,s_2 \in V_1$ according to $D_1$ is just
\[\prob{V_i}^2 \geq \left( 1 - 3 \varepsilon_{2,t} \right )^2 \geq 1 - 6 \varepsilon_{2,t}. \]
If we choose $s_1,s_2 \sim D_1$ that disagree, then at least on of them is not in the majority set, hence
\[ \varepsilon_{1,t} \leq 6 \varepsilon_{2,t}. \]

Next we show that $\varepsilon_{2,t}  \leq 2 \varepsilon_{1,t}$. If $\varepsilon_{1,t} \geq \frac{1}{6}$ then $\varepsilon_{2,t} \leq 1 \leq 6\varepsilon_{1,t}$ so assume otherwise.

Consider again $V_1$, the set of all $f_s$ that agree with the most popular assignment. From independence
\[\prob{V_1}\prob{V_1^c} = \prob{s_1 \in V_1, s_2 \notin V_1} \leq \Prob[s_1,s_2]{\rest{f_{s_1}}{t} \ne \rest{f_{s_2}}{t}} = \varepsilon_{2,t}.\]
The graph where we sample $s_1,s_2$ independently is also a $\frac{1}{3}$-edge expander. By the same argument as in the other direction, we can get that $\prob{V_1} \geq \frac{1}{2}$, thus
\[\prob{V_1^c} = 2\varepsilon_{1,t}. \]
Recall that this inequality is true also when sampling $s_1 \in $. If we chose $s_1,s_2 \sim D_2$ such that they disagree, then at least one vertex is in $V_1^c$. Thus
$\varepsilon_{1,t} \leq \prob{V_1^c} \leq 2\varepsilon_{2,t}.$
\end{proof} 

\section{List of Abbreviations for STAV-Structures}
\def\arraystretch{1.25}
\begin{tabular}{| p{3cm} | p{8cm} | c |}\hline
 \textbf{Name}& \textbf{Definition} & \textbf{Reference}\\\hline
STAV-Structure &
A system of sets with four layers: S - sets, T - intersections, A - amplification, V - vertices.

It is accompanied by a distribution

\((s,t,(a,v)) \sim D_{stav}\).
& \pref{def:STAV} \\ \hline
STS-distribution & 
A distribution where we sample \(t \in T\), and then \(s_1,s_2 \in S\) so that \(s_1 \cap s_2 \supset t\). The marginal \((s_i,t)\) is the same as the marginal in \(D_{stav}\).&
\pref{def:STAV} \\\hline
VASA-distribution &
A distribution \((v,a,s,a') \sim D_{vasa}\) where the marginals \((v,a,s), (v,a',s)\) are the same as \(D_{stav}\).&
\pref{def:STAV} \\\hline
Reach Graph &
The bipartite graph between \(V\) and \(A\) where we choose an edge \((v,a)\) according to the STAV-distribution. 

We denote by \(\adj{a}\) or \(\adj{v}\) then neighbours of \(a\) or \(v\) in this graph, respectively.&
\pref{def:reach-graph}.  \\\hline
Local Reach Graph (\(\AV{s}\)-graph) &
For a fixed \(s_0 \in S\), the \(\AV{s_0}\)-graph is a bipartite graph where \(L = \sett{a}{a \subset s_0}\) and \(R = \sett{v}{v \in s_0}\). The edges are chosen according to the STAV-distribution given that \(s=s_0\).& \pref{def:local-reach-graph} \\ \hline
\(\sts_a\)-Graph &
For a fixed \(a_0 \in A\), the \(\sts_{a_0}\)-graph is a graph whose elements are \(\sett{s}{s \supset a_0}\). We connect \(s,s'\) when there exists \(t \in T\) so that \(a_0 \subset t \subset s \cap s'\).&
\pref{def:sts-a-graph} \\ \hline
\(\sts_{a,v}\)-Graph &
For a fixed \(a_0 \in A\) and \(v_0 \in \adj{a_0}\), the \(\sts_{a_0,v_0}\)-graph is a graph whose elements are \(\sett{s}{s \supset (a_0,v_0)}\). We connect \(s,s'\) when there exists \(t \in T\) so that \((a_0,v_0) \subset t \subset s \cap s'\). &
\pref{def:sts-a-v-graph} \\ \hline
\(\vASA{v}\)-graph &
For a fixed \(v_0 \in V\) the \(\vASA{v_0}\)-graph is a graph whose elements are \(a \in \adj{v_0}\). We connect \(a,a'\) with a labeled edge \((a,s,a')\) if \((v_0,a,s,a')\) is in the support of \(D_{vasa}\).
& \pref{def:v-asa-graph} \\ \hline
 Bipartite 
 
 \(\VAS{a}\)-Graph & 
 For a fixed \(a_0 \in A\), the \(\VAS{a_0}\)-Graph is a bipartite graph where one side is \(L = \adj{a_0}\). The other side is the set of \((s,a')\) so that \((a_0,s,a')\) is in the support of the marginal of \(D_{vasa}\).
 
 We sample an edge in this graph by sampling \((v,a,s,a')\) given that \(a = a_0\).
 &
 \pref{def:bip-vas-a-graph} \\ \hline
 Surprise &
 Let \(\set{f_s}_{s\in S}\) be some local ensemble. The surprise of the ensemble is the probability over \((s,a,v)\) that \(\rest{f_s}{a} = \rest{f_{s'}}{a}\) but \(f_s(v) \ne f_{s'}(v)\).&
 \pref{def:surprise} \\ \hline
\end{tabular}

\section{List of Results}
\label{sec:list-of-results}
\subsection{Main Theorem}
\begin{theorem}[Restatement of \pref{thm:main-STAV-agreement-theorem}] 
Let $\Sigma$ be some finite alphabet (for example $\Sigma = \set{0,1}$).
Let \( X = (S,T,A,V)\) be a $\gamma$-good STAV-structure for some \(\gamma < \frac{1}{3}\). Let \( \FF = \sett{f_s :\dl{s}{V} \to \Sigma }{s \in S}\) be an ensemble such that
\begin{enumerate}
    \item Agreement:
    \begin{equation*} %\label{eq:agreement}
        \disagr{X}(f) \leq \varepsilon,
    \end{equation*}
    \item Surprise:
    \begin{align} %\label{eq:surprise}
        \surp(X,f)\le O(\gamma)
    \end{align}
\end{enumerate}
Then assuming either \pref{ass:AVS-graph-good-sampler} for $r=1$ or \pref{ass:a-not-large}, 
\[ \dist_\gamma(f,\glob)\le O(\varepsilon).
\]
More explicitly, there exists a global function $G: V \to \Sigma$ s.t.
\[ \Prob[s \in S]{f_s \ane{\gamma} \rest{G}{\dl{s}{V}}} \defeq \Prob[s \in S]{ \cProb{v \in V}{f_s(v) \ne \rest{G}{\dl{s}{V}}}{v \in s} \geq \gamma } = \bigO{\varepsilon}. \]
Moreover, for any $r > 0$, if either \pref{ass:AVS-graph-good-sampler} or \pref{ass:a-not-large} holds then
\begin{equation*} %\label{eq:thm-then}
    \Prob[s \in S]{f_s \ane{r \gamma} \rest{G}{\dl{s}{V}}} = \bigO{\left (1+\frac{1}{r} \right )\varepsilon}.
\end{equation*}

The O notation \emph{does not} depend on any parameter including $\gamma,\varepsilon$, the size of the alphabet, the size of $|S|,|T|,|A|,|V|$ and, size of any $s \in S$.

\end{theorem}

\subsection{Applications of Main Theorem}
\begin{enumerate}
    \item Agreement tests on two-sided HDX.
    \restatetheorem{thm:main-agreement-theorem-two-sided}
    
    \item Agreement tests on one-sided HDX.
    \restatetheorem{thm:main-agreement-theorem-one-sided}
    
    \item Agreement tests on vertex neighbourhoods.
    \restatetheorem{thm:agreement-on-links}
    
    \item Agreement tests on the Affine and Linear Grassmann Posets:
    \restatetheorem{thm:agreement-on-Grassmann-affine}
    \restatetheorem{thm:agreement-on-Grassmann}
\end{enumerate}

\subsection{Analysis of the Complement Walk}
\restatetheorem{thm:complement-walk-is-a-good-expander}
\subsection{High Dimensional Expander Mixing Lemma}
\begin{enumerate}
    \item Two sided case:
    \restatetheorem{lem:two-sided-HDEML}
    \item One sided partite case:
    \restatetheorem{lem:one-sided-HDEML}
\end{enumerate}

\end{document}